\documentclass[smallcondensed,a4paper]{svjour3}
\pdfoutput=1
\smartqed
\usepackage{graphicx}
\usepackage{newtxtext,newtxmath}
\usepackage[utf8]{inputenc}
\usepackage[T1]{fontenc}
\usepackage{llncsstuffthm}
\usepackage{cwdefs}
\usepackage{amsfonts}
\usepackage{bm}
\usepackage{multirow}
\usepackage{longtable}
\usepackage{booktabs}
\usepackage{amsmath}
\usepackage{bigdelim}
\usepackage{comment}
\usepackage{xspace}
\usepackage{booktabs}
\usepackage{setspace}
\usepackage{upgreek}
\usepackage{enumitem}
\usepackage{rotating}
\usepackage{tikz}
\usetikzlibrary{shapes.geometric}
\usetikzlibrary{shapes.symbols}
\usetikzlibrary{positioning}
\usetikzlibrary{arrows}
\usepackage{array}
\usepackage{xcolor}
\usepackage{colortbl}
\usepackage{capt-of}
\usepackage[labelfont=bf,hypcap=false,font=small]{caption}
\usepackage{float}
\usepackage{pdfpages}

\usepackage{hyperref}
\hypersetup{
hyperfootnotes=false,
breaklinks=true,
hidelinks,
linktoc=all,
colorlinks=false
}
\usepackage{plabels}

\tikzset{
itria/.style={
  draw,
  solid,
  thin,
  isosceles triangle,
  isosceles triangle apex angle=60,
  shape border rotate=90,yshift=-6.8ex}
}

\tikzset{
itriaw/.style={
  draw,
  solid,
  thin,
  isosceles triangle,
  isosceles triangle apex angle=90,
  shape border rotate=90,yshift=-5.1ex}
}

\newcommand{\triaw}[2]
{{ node[itriaw]
    {\small \makebox[#1]{\raisebox{-1.5ex}[0pt][5pt]{\rule[-0.55ex]{0pt}{2.2ex}$\,$#2} }}}}

\newcommand{\tableauscale}{0.7}

\newcommand{\vbar}{\raisebox{-0.60ex}{\rule{0pt}{2.35ex}}}
\definecolor{tcolbbbbg}{rgb}{0.8,0.8,0.8}
\definecolor{tcolaaabg}{rgb}{1.0,1.0,1.0}
\newcommand{\taaa}[1]{\colorbox{tcolaaabg}{$#1\vbar$}}
\newcommand{\tbbb}[1]{\colorbox{tcolbbbbg}{$#1\vbar$}}

\newcommand{\tbbbtxt}[1]{\setlength{\fboxsep}{1pt}\colorbox{tcolbbbbg}{#1\vbar}}
\newcommand{\tccc}[1]{\fbox{$#1\vbar$}}
\newcommand{\tccctxt}[1]{{\setlength{\fboxsep}{1pt}\fbox{#1\vbar}}}
\newcommand{\extabld}{9ex}
\newcommand{\extabscale}{0.85}

\newcommand{\nhphantom}[1]{\sbox0{#1}\hspace*{-\the\wd0}}
\newcommand{\nannot}[1]%
           {\hspace{0.2em}{[}#1{]}\nhphantom{\hspace{0.2em}{[}#1{]}}}
\newcommand{\nannotw}[1]%
           {\hspace{0.6em}{[}#1{]}\nhphantom{\hspace{0.6em}{[}#1{]}}}

\newcommand{\nannotmm}[1]%
           {\hspace{0.2em}\underline{{[}#1{]}}\nhphantom{\hspace{0.2em}\underline{{[}#1{]}}}}
\newcommand{\nannotwmm}[1]%
           {\hspace{0.6em}\underline{{[}#1{]}}\nhphantom{\hspace{0.6em}\underline{{[}#1{]}}}}

\title{Craig Interpolation with Clausal First-Order Tableaux}

\author{Christoph Wernhard}
\institute{\email{info@christophwerhard.com}}
\date{Version: March 21, 2021}
\renewcommand{\makeheadbox}{\vspace{1cm}}

\makeatletter
\newlength{\negph@wd}
\DeclareRobustCommand{\negphantom}[1]{%
  \ifmmode
    \mathpalette\negph@math{#1}%
  \else
    \negph@do{#1}%
  \fi
}
\newcommand{\negph@math}[2]{\negph@do{$\m@th#1#2$}}
\newcommand{\negph@do}[1]{%
  \settowidth{\negph@wd}{#1}%
  \hspace*{-\negph@wd}%
}
\makeatother

\newcommand{\ex}{\exists}
\newcommand{\all}{\forall}

\newcommand{\nand}{\uparrow}

\newcommand{\invsubst}[2]{#1 \la #2^{-1} \ra}
\newcommand{\invsubstpost}[1]{\la #1^{-1} \ra}

\newcommand{\emptysubst}{\varepsilon}

\newcommand{\m}[1]{\mathit{#1}}

\newcommand{\dom}[1]{\m{\mathcal{D}\hspace{-0.08em}om}(#1)}
\newcommand{\rng}[1]{\m{\mathcal{R}\hspace{-0.02em}ng}(#1)}
\newcommand{\varrng}[1]{\m{\mathcal{VR}\hspace{-0.02em}ng}(#1)}

\newcommand{\var}[1]{\m{\mathcal{V}\hspace{-0.11em}ar}(#1)}
\newcommand{\free}[1]{\m{\mathcal{F}\hspace{-0.11em}ree}\m{\mathcal{V}\hspace{-0.11em}ar}(#1)}
\newcommand{\voc}[1]{\m{\mathcal{V}\hspace{-0.11em}oc}^\pm{(#1)}}
\newcommand{\vocplain}[1]{\m{\mathcal{V}\hspace{-0.11em}oc}{(#1)}}

\newcommand{\pred}[1]{\m{\mathcal{P}\hspace{-0.11em}red}^\pm(#1)}
\newcommand{\predplain}[1]{\m{\mathcal{P}\hspace{-0.11em}red}(#1)}
\newcommand{\fun}[1]{\m{\mathcal{F}\hspace{-0.18em}un}(#1)}
\newcommand{\lit}[1]{\m{\mathcal{L}\hspace{-0.01em}iterals}(#1)}

\newcommand{\sterm}{\text{-term}}
\newcommand{\sterms}{\text{-terms}}

\renewcommand{\miff}{\textit{iff}}
\renewcommand{\mimp}{\textit{implies}}

\newcommand{\algoinput}{\smallskip \noindent\textsc{Input: }}
\newcommand{\algooutput}{\smallskip \noindent\textsc{Output: }}
\newcommand{\algomethod}{\smallskip \noindent\textsc{Method: }}

\newenvironment{arrayprf}
{\begin{array}{Z{2.5em}@{\hspace{1em}}X{32em}}}
{\end{array}}

\newenvironment{arrayprfeq}
{\begin{array}{Z{2.5em}@{\hspace{1em}}Y{1.5em}X{30.5em}}}
{\end{array}}

\newenvironment{arrayprfmeq}
{\begin{array}{Z{2.5em}@{\hspace{1em}}Z{3.0em}@{\hspace{1em}}X{28.0em}}}
{\end{array}}

\newcommand{\tup}[1]{\boldsymbol{#1}}

\renewcommand{\tts}{\tup{t}}

\newcommand{\nlit}[1]{\f{lit}(#1)}
\newcommand{\nclause}[1]{\f{clause}(#1)}

\newcommand{\nside}[1]{\f{side}(#1)}
\newcommand{\ntgt}[1]{\f{tgt}(#1)}
\newcommand{\nipol}[1]{\f{ipol}(#1)}
\newcommand{\npath}[2]{\f{path}_{#1}(#2)}
\newcommand{\ncopy}[1]{\f{copy}(N)}

\newcommand{\stv}{\upsilon}   
\newcommand{\stren}{\rho}     
\newcommand{\stmerge}{\delta} 
\newcommand{\stsk}{\phi}      
\newcommand{\stz}{\theta}     
\newcommand{\sth}{\eta}       
\newcommand{\stt}{\sigma}     

\renewcommand{\stv}{\theta_{\textsc{exp}}}   
\renewcommand{\stren}{\theta_{\textsc{lft}}} 
\renewcommand{\stmerge}{\phi_{\textsc{exp}}} 
\renewcommand{\stsk}{\phi_{\textsc{lft}}}    
\newcommand{\stski}{\phi}     
\renewcommand{\stt}{\sigma}               
\renewcommand{\stz}{\eta_{\textsc{lft}}}     
\renewcommand{\sth}{\eta_{\textsc{exp}}}     

\newcommand{\FA}{F_{\textsc{top}}}
\newcommand{\FE}{F_{\textsc{exp}}}
\newcommand{\FG}{F_{\textsc{Grd}}}
\newcommand{\FQ}{F_{\textsc{lft}}}
\newcommand{\GE}{G_{\textsc{exp}}}
\newcommand{\HG}{H_{\textsc{grd}}}
\newcommand{\HQ}{H_{\textsc{lft}}}

\newcommand{\symset}[1]{\mathcal{#1}}

\renewcommand{\us}{\symset{U}}
\renewcommand{\vs}{\symset{V}}
\newcommand{\xwf}{\symset{X}}  
\newcommand{\ywg}{\symset{Y}}  
\newcommand{\fgs}{\symset{F\!G}}
\newcommand{\ffs}{\symset{F}}
\newcommand{\ggs}{\symset{G}}

\newcommand{\XS}{\symset{X}}

\newcommand{\al}{\raisebox{-0.48ex}{\rule{0pt}{1.95ex}}}

\newcommand{\CTIG}{CTI\xspace}
\newcommand{\CTIF}{CTIF\xspace}

\newcommand{\aaa}{\f{F}}
\newcommand{\bbb}{\f{G}}
\newcommand{\sided}{two-sided\xspace}
\newcommand{\Sided}{Two-Sided\xspace}

\newcommand{\npathL}[1]{\npath{\aaa}{#1}}
\newcommand{\npathR}[1]{\npath{\bbb}{#1}}

\newcommand{\FL}{F}

\newcommand{\cpi}[1]{\f{pi}(#1)}
\newcommand{\huang}[1]{\f{pi}(#1)}
\newcommand{\CUTS}{\mathit{CUTS}}
\newcommand{\fct}{\f{ct}}
\newcommand{\ct}[1]{\fct(#1)}

\newcommand{\mm}[1]{#1}
 
\newcommand{\LBase}{Lifting Base\xspace}
\newcommand{\LBases}{Lifting Bases\xspace}

\newcommand{\lbase}{lifting base\xspace}

\newcounter{ibcounter}

\newcommand{\ibn}[1]{\prl{#1}}
\newcommand{\ibm}[1]{(\prefNumber{#1}^\prime)}
\newcommand{\ibrefm}[1]{(\prefGlobalNumber{def-ib:#1}$'$)}
\newcommand{\ibref}[1]{(\prefGlobalNumber{def-ib:#1})}

\newcommand{\hypertableau}{hypertableau\xspace}
\newcommand{\hypertableaux}{hypertableaux\xspace}
\newcommand{\Hypertableau}{Hypertableau\xspace}

\newcommand{\strictsubterm}[2]{#1 \lhd #2}

\renewcommand{\GR}{G}
\newcommand{\lnotGR}{\lnot \GR}

\newcommand{\NQ}{\overline{Q}}

\newcommand{\w}{v'}
\newcommand{\R}{Q'}

\newcommand{\CMProver}{\name{CMProver}\xspace}

\makeatletter
\def\denseparagraph{\@startsection{paragraph}{4}{\z@}%
  {-9pt plus-2pt}{\z@}{\normalsize\itshape}}
\makeatother

\begin{document}

\maketitle

\begin{abstract}
  We develop foundations for computing Craig-Lyndon interpolants of two given
  formulas with first-order theorem provers that construct clausal tableaux.
  Provers that can be understood in this way include efficient
  machine-oriented systems based on calculi of two families: goal-oriented
  such as model elimination and the connection method, and bottom-up such as
  the \hypertableau calculus. We present the first interpolation method for
  first-order proofs represented by closed tableaux that proceeds in two
  stages, similar to known interpolation methods for resolution proofs.  The
  first stage is an induction on the tableau structure, which is sufficient to
  compute propositional interpolants.  We show that this can linearly simulate
  different prominent propositional interpolation methods that operate by an
  induction on a resolution deduction tree. In the second stage, interpolant
  lifting, quantified variables that replace certain terms (constants and
  compound terms) by variables are introduced.  We justify the correctness of
  interpolant lifting (for the case without built-in equality) abstractly on
  the basis of Herbrand's theorem and for a different characterization of the
  formulas to be lifted than in the literature.  In addition, we discuss
  various subtle aspects that are relevant for the investigation and practical
  realization of first-order interpolation based on clausal tableaux.

\end{abstract}

\section{Introduction}
\label{sec-intro}

By Craig's interpolation theorem \cite{craig:linear,craig:uses},\footnote{The
  development of Craig's interpolation theorem is described in
  \cite{craig:2008:road}.} for two first-order formulas $F$ and $G$ such that
$F$ entails $G$ there exists a third first-order formula $H$ that is entailed
by $F$, entails $G$ and is such that all predicate and function symbols and
all free variables occurring in it occur in both $F$ and $G$. Such a
\name{Craig interpolant}~$H$ can be \emph{constructed} from given formulas $F$
and $G$ with a calculus that allows to extract~$H$ from a proof that $F$
entails $G$, or, equivalently, a proof that $F \imp G$ is valid, or, again
equivalently, a refutation of $F \land \lnot G$.  Automated construction of
interpolants has many applications, in the area of computational logic most
notably in symbolic model checking, initiated with \cite{mcmillan:2003}, and
in query reformulation
\cite{segoufin:vianu:determinacy:05,marx:2007,nash:2010,borgida:2010,toman:wedell:book,benedikt:guarded,benedikt:etal:2014:generating,toman:2015:tableaux,benedikt:book,benedikt:2017,toman:2017}.
The foundation for the latter application field is the observation that a
reformulated query can be viewed as a \emph{definiens} of a given query where
only symbols from a given set, the target language of the reformulation, occur
in the definiens.  The existence of such definientia, that is, definability
\cite{tarski:35}, or \name{determinacy} as it is called in the database
context, can be expressed as validity and their synthesis as interpolant
construction \cite{craig:uses}. For example, a \name{definiens~$H$ of a unary
  predicate~$\fp$ within a first-order formula~$F$} can be characterized by
the following conditions:
\begin{enumerate}
\item $F$ entails $\forall x\, (\fp(x) \equi H)$.
\item $\fp$ does not occur in  $H$.
\end{enumerate}
The variable $x$ is allowed there to occur free in $H$. We further assume that
$x$ does not occur free in $F$ and let $F^\prime$ denote $F$ with $\fp$
replaced by a fresh symbol $\fp^\prime$. Now the characterization of
\name{definiens} by the two conditions given above can be equivalently
expressed as
\begin{center}
\name{$H$ is a
  Craig interpolant of the two formulas $F \land \fp(x)$ and $\lnot (F^\prime
  \land \lnot \fp^\prime(x))$}.
\end{center}
A definiens $H$ exists if and only if it is valid that the first formula
implies the second one.

There are many known ways to strengthen Craig's interpolation theorem by
ensuring that if formulas~$F$ and~$G$ where $F$ entails $G$ satisfy certain
syntactic restrictions, then there exists an interpolant~$H$ that also
satisfies certain syntactic restrictions. We consider here in particular that
predicates occur in~$H$ only with \emph{polarities} with which they occur in
both $F$ and $G$.  (A predicate occurs with \name{positive} (\name{negative})
polarity in a formula if it occurs there in the scope of an even (odd) number
of negation operators.)  The respective strengthened interpolation theorem has
been explicated by Lyndon \cite{lyndon}, hence we call Craig interpolants that
meet this additional restriction \name{Craig-Lyndon interpolants}.

\subsection{Craig Interpolation and Clausal Tableaux}
\label{sec-intro-tab}
  
The construction of Craig interpolants of given first-order formulas has been
elegantly specified in the framework of analytic tableaux by Smullyan
\cite{smullyan:book:68,fitting:book}. Although this has been taken as
foundation for applications of interpolation in query reformulation
\cite{toman:wedell:book,benedikt:book}, it has been hardly used as a basis for
the practical computation of first-order interpolants with automated reasoning
systems, where so far mainly the interpolant extraction from resolution
proofs, also with paramodulation and superposition, has been considered (see
\cite{bonacina:15:ground,bonacina:15:on,kovacs:17} for recent overviews and
discussions).  Some of these techniques only work on specially constrained
proofs.  In general they support equality handling with superposition and
integrated theory reasoning, targeted at applications in verification.

Here we approach the computation of interpolants from another paradigm of
automated reasoning, the framework of \defname{clausal tableaux} (or
\name{clause tableau} \cite{handbook:ar:haehnle})
\cite{letz:clausal:deduktion,letz:habil,handbook:tableaux:letz,letz:stenz:handbook}. It
has been established in the 1990s as a tool to develop and investigate methods
for fully automated (in contrast to interactive) first-order theorem proving
that operate not by manipulating clause \emph{sets} (like resolution), but by
creating a \emph{tree structure}, the clausal tableau, that arranges
\emph{copies of the input clauses}.  It is this explicit representation of the
proof as a structure, the clausal tableau, that lends itself to inductive
extraction of an interpolant.

Expectations are that, on the one hand, the elegance of Smullyan's
interpolation method for analytic tableaux can be utilized and, on the other
hand, a foundation for efficient practical implementations can be laid.

The practically efficient theorem proving methods that can be viewed as
operating by constructing a clausal tableau can be roughly divided into two
major families. First, methods that are goal-sensitive, typically proceeding
with the tableau construction ``top-down'', by ``backward reasoning'',
starting with clauses from the theorem in contrast to the axioms.  Aside of
clausal tableaux in the literal sense, techniques to specify and investigate
such methods include model elimination \cite{loveland:1978}, the connection
method \cite{bibel:ar:1982,bo:C46}, and the Prolog Technology Theorem Prover
\cite{pttp}.  One of the leading first-order proving systems of the 1990s,
\name{SETHEO} \cite{setheo:92}, followed that approach.  The \name{leanCoP}
system \cite{leancop} along with its recent derivations
\cite{kaliszyk15:tableaux,femalecop} as well as the \CMProver component of
\name{PIE} \cite{cw-mathlib,cw-pie,cw-pie:2020} are implementations in active
duty today.

The second major family of methods constructs clausal tableaux ``bottom-up'',
in a ``forward reasoning'' manner, by starting with positive axioms and
deriving positive consequences.  With the focus of their suitability to
construct model representations, these methods have been called
\name{bottom-up model generation (BUMG)} methods \cite{bumg}.  They include,
for example, \name{SATCHMO} \cite{satchmo} and the \hypertableau calculus
\cite{hypertab}, with implementations such as \name{Hyper}, formerly called
\name{\mbox{E-KRHyper}} \cite{cw-ekrhyper,cw-krhyper,hyper:2013}.
\Hypertableau methods are also used in high-per\-for\-mance description logic
reasoners \cite{dl:hypertab}.  It appears that the family of chase methods
from the database field \cite{chase:maier:79,chase:fagin:2005}, which recently
got attention anew in knowledge representation (see, e.g.,
\cite{grau:2013:acyclicity}), can also be understood as such a bottom-up
tableau construction.

Depending on the method, the constructed clausal tableaux have particular
structural characteristics, such as for top-down techniques typically the
connectedness condition, which ensures that each inner node in the tableau has
a child with complementary literal label.  Bottom-up techniques are often
applied in a way such that nodes labeled with a negative literal only appear
as leaves.  For a systematic overview of different variants of tableau
structures and methods, including clausal tableaux realizing both major
paradigms considered above see \cite{handbook:ar:haehnle}.

An essential feature that makes tableau methods particularly suitable for
interpolation compared to methods based on resolution and
paramodulation/superposition is that they generate only formulas which are
instances of subformulas of the input.  For \emph{clausal} tableaux this means
that only \emph{instances of input clauses} are created. More generally,
methods of the instance-based approach to theorem proving (see
\cite{jacobs:inst,baumgartner:2010:ibased} for overviews) realize this
principle. A resolution step, in contrast, recombines fragments of different
clauses into a new clause.
Nevertheless, clausal tableau methods might be complemented by preprocessors
that perform operations which recombine clause fragments, for example
predicate elimination, where all clauses with a given predicate are replaced
by all non-tautological resolvents upon literals with that predicate if this
results in a smaller clause set (see, e.g,
\cite{biere:elim,blocked:fol:2017,kiesl:suda:clause:elim:2017}).

An essential feature that distinguishes clausal from analytic tableau is that
with the clausal form only a particularly simple formula structuring has to be
considered, in essence sets of clauses. Through preprocessing with conversion
to prenex form and Skolemization, the handling of quantifications amounts for
clausal tableau methods just to the handling of free variables, which are
considered implicitly as universally quantified.  There are different
approaches for this: In \name{free-variable} tableaux variables are treated in
a \name{rigid} manner, which means their scope is the whole tableau and they
are considered as placeholders for arbitrary ground terms (see, e.g.,
\cite[Sect.~4]{handbook:tableaux:letz}).  This is typical for the mentioned
top-down methods.  It may, however, complicate tableau construction, because
the instantiation of a variable has effect on all its occurrences throughout
the partially constructed tableau.  Hence the mentioned bottom-up methods
typically handle variables in other ways (see, e.g., \cite{hypertab} and
\cite[Sect.~3.2.4]{handbook:ar:haehnle}).  For our interpolant construction
from given clausal tableaux it seems that the rigid view of variables is the
most simple and straightforward.

For propositional logic there are various known methods to extract an
interpolant by an induction on a tree representation of a resolution proof
(e.g., \cite{huang:95,mcmillan:2003}), which are surveyed and presented in a
common framework in \cite{bonacina:15:ground}.  As we will see, these
interpolation methods for resolution proofs can be linearly simulated by the
extraction of interpolants with an induction on a clausal tableau.  These
simulations are based on a known linear simulation of tree resolution, that
is, resolution with the proof constrained to form a tree in contrast to an
unrestricted dag, by clausal tableaux \cite[Sect.~7.3]{letz:habil}, which is
related to earlier results \cite[Sect.~5.6]{reckhow:thesis}. Different
interpolation methods for resolution proofs are reflected in slight variations
of the involved translation of resolution deduction trees into clausal
tableaux.

\subsection{The Two-Stage Approach to First-Order Interpolation}
\label{sec-intro-lifting}

The \name{two-stage approach} to interpolation proceeds by first computing
with an induction on a proof structure an intermediate formula that is
constrained in certain ways and can be \name{lifted} in a second stage to an
interpolant. \name{Lifting} means there to replace terms by variables and
prepending a quantifier prefix upon these. The two-stage approach seems to
originate in \cite{huang:95}, has been systematically investigated in
\cite{bonacina:15:on}, and is used also in \cite{baaz:11,kovacs:17}.  The
lifting step is justified in \cite{huang:95,baaz:11} on the basis of a proof
structure, with resolution/paramodulation and natural deduction, respectively,
as underlying calculi. In \cite{bonacina:15:on} it is justified independently
of a calculus, but applies only if the terms to be replaced by variables are
constants.

It will be shown here that also the general case, where the terms to be
replaced can be compound terms, leading to dependency constraints on the
quantifications in the computed prefix, can be justified independently of a
calculus and of proof structures on fundamental techniques of automated
theorem proving, Skolemization and Herbrand's theorem.

The intermediate formula computed in the first stage is called in
\cite{huang:95,baaz:11,bonacina:15:on} \name{relational}, \name{weak} and
\name{provisional interpolant}, respectively, and is characterized with
respect to the two interpolation inputs as satisfying the semantic constraints
of a Craig interpolant, but the syntactic constraints only with respect to
predicates (relational and weak interpolant) or not at all (provisional
interpolant).  Our view of lifting based on Herbrand's theorem allows to
characterize the intermediate formula computed in the first stage as an actual
Craig-Lyndon interpolant of two formulas that relate to the interpolation
inputs.  To justify interpolant lifting it is just required that these two
formulas \emph{exist}. They need not to be materialized at interpolant
computation.

In the first stage of the two-stage method presented here a ground interpolant
is extracted from a closed clausal ground tableau that represents a proof of
the unsatisfiability of a set of first-order clauses.  The construction
process of the tableau is entirely independent from its use for interpolation,
permitting to use existing provers without modification and utilizing their
optimizations.  Variables in the tableau returned from the prover must be
rigid.  Before the inductive interpolant extraction they are instantiated by
ground terms, which is then a simple linear operation.  Preprocessing has to
be restricted, as for interpolation it is not sufficient to preserve just
unsatisfiability of the input clauses.

\subsection{Summary of Contributions and Structure of the Paper}

The contributions of the paper can be summarized as:
\begin{enumerate}

\item An extension of the two-stage approach to Craig interpolation to
  first-order tableaux-based methods with rigid variables.
  
\item A method for ground Craig-Lyndon interpolation by an induction on
  clausal tableaux obtained as proof representation from clausal first-order
  theorem provers.  The method adapts a known interpolation method for
  analytic tableaux to clausal tableaux.

\item For propositional logic: Linear simulations of various known inductive
  interpolation methods for resolution proofs with the ground interpolation
  method based on clausal tableaux.  Differences in the methods for resolution
  proofs correspond to variations in the translation of resolution deduction
  trees into clausal tableaux.
  
\item In the context of the two-stage approach to interpolation: A new
  justification for the second stage, interpolant lifting, that is based on
  Herbrand's theorem and independent of a calculus or proof structure. It
  provides a new characterization of the results of the formulas computed in
  the first stage as actual Craig-Lyndon interpolants of certain intermediate
  formulas for which a construction can be specified, but which need not to be
  actually constructed in the computation of first-order interpolants.
  
\item Discussion of various subtle aspects that are relevant for the
  investigation and practical realization of first-order interpolation, such
  as adequacy of the color-based terminology, interpolation-specific choices
  in the algorithms, and adaptations of preprocessing.

\item A basis for implementations of interpolation with efficient
  machine-oriented theorem provers for first-order logic that can be
  understood as constructing clausal tableaux. With methods and systems of two
  main families, goal-oriented ``top-down'' and forward reasoning
  ``bottom-up'', there is a wide range of potential applications.
\end{enumerate}

The rest of the paper is structured as follows: After notation and basic
terminology have been specified in Sect.~\ref{sec-notation}, precise accounts
of \name{clausal tableau} and related notions are given in
Sect.~\ref{sec-tableaux}.  In Sect.~\ref{sec-ipol-basic} the inductive
extraction of ground interpolants from closed clausal ground tableaux is
specified and proven correct.  The relationship to other propositional
interpolation methods including the simulation of methods based on resolution
proofs is shown in Sect.~\ref{sec-simul}.  We then turn to the lifting stage.
In Sect.~\ref{sec-lift-intro} additional notation is specified and the main
theorem about interpolant lifting is stated in Sect.~\ref{sec-lift-thm}.  It
is applied in Sect.~\ref{sec-lift-procedure} to specify a two-stage
interpolation procedure and proven in
Sect.~\ref{sec-lift-proof}. Section~\ref{sec-cli-related} discusses related
work, refinements of our method and some open issues for further research.
Section~\ref{sec-conclusion} concludes the paper, summarizing its main
contributions and indicating related additional research perspectives.

A work-in-progress poster of this research at an earlier stage was
presented at the\linebreak \name{TABLEAUX~2017} conference.

\enlargethispage{0.1cm}
\vspace{-0.1cm}
\section{Notation and Basic Terminology}
\label{sec-notation}

\paragraph{Formulas.}
We basically consider first-order logic without
equality.\footnote{\label{foot-eq}This does not preclude to represent equality
  as a predicate with axioms that express reflexivity, symmetry, transitivity
  and substitutivity.}  Atoms are of the form $p(t_1,\ldots,t_n)$, where $p$
is a \defname{predicate symbol} (briefly \defname{predicate}) with associated
arity $n \geq 0$ and $t_1,\ldots,t_n$ are terms formed from \defname{function
  symbols} (briefly \name{functions}) with associated arity~$\geq 0$ and
\defname{individual variables} (briefly \name{variables}).  Function symbols
with arity $0$ are also called \defname{individual constants} (briefly
\name{constants}).  Unless especially noted, a \defname{formula} is understood
as a formula of first-order logic without equality, constructed from atoms,
constant operators $\true$, $\false$, the unary operator $\lnot$, binary
operators $\land, \lor$ and quantifiers $\forall, \exists$ with their usual
meaning.  Further binary operators~$\imp$, $\equi$ as well as $n$-ary versions
of $\land$ and $\lor$ can be understood as meta-level shorthands.  Also
quantification upon a finite set of variables is used as shorthand for
successive quantification upon each of its elements.  The operators $\land$
and~$\lor$ bind stronger than $\imp$ and $\equi$. The scope of $\lnot$, the
quantifiers, and the $n$-ary connectives is the immediate subformula to the
right.  Formulas in which no predicates with arity larger than zero and no
quantifiers occur are called \defname{propositional}.

\paragraph{Clausal Formulas.}

A \defname{literal} is an atom or a negated atom.  If $A$ is an atom, then the
\name{complement} of $A$ is $\lnot A$ and the complement of $\lnot A$ is $A$.
The complement of a literal~$L$ is denoted by~$\du{L}$.  A literal and its
complement are said to be \defname{complementary}.  A \defname{clause} is a
(possibly empty) disjunction of literals.  A \defname{clausal formula} is a
(possibly empty) conjunction of clauses, called the \defname{clauses in} or
the \defname{clauses of} the formula.  A clausal formula is a special case of
a formula.

\paragraph{Variables, Vocabulary and Literals of a Formula.}

If $E$ is a term or a quantifier-free formula, then the set of variables
occurring in $E$ is denoted by $\var{E}$.  If $F$ is a formula, then the set
of variables that occur \defname{free} in~$F$ is denoted by $\free{F}$.  A
formula without free variables is called a \defname{sentence}.  A term or
quantifier-free formula in which no variable occurs is called
\defname{ground}. A ground formula is thus a special case of a sentence.
Symbols not present in the formulas and other items under discussion are
called \defname{fresh}.

If $E$ is a term or a formula, then the set of functions (including constants)
occurring in~$E$ is denoted by $\fun{E}$.  A subformula occurrence has in a
given formula \defname{positive (negative) polarity}, or is said to occur
\defname{positively (negatively)} in the formula, if it is in the scope of an
even (odd) number of negations.  If $F$ is a formula, then the set of
predicates occurring in $F$ is denoted by $\predplain{F}$, and $\pred{F}$
denotes the set of pairs $\la p, \mathit{pol}\ra$, where $p$ is a predicate
and $\mathit{pol} \in \{{+},{-}\}$, such that an atom with predicate~$p$
occurs in $F$ with the polarity indicated by $\mathit{pol}$.  For example,
$\pred{\fp} = \{\la \fp, {+}\ra\}$ and $\pred{\fp \lor \lnot \fp} = \{\la \fp,
{+}\ra, \la \fp, {-}\ra \}$.  We define $\vocplain{F} \eqdef \predplain{F}
\cup \fun{F}$ and $\voc{F} \eqdef \pred{F} \cup \fun{F}$.

If $F$ is a formula, then the set of pairs $\la A, \mathit{pol}\ra$, where $A$
is an atom and $\mathit{pol} \in \{{+},{-}\}$ such that $A$ occurs in $F$ with
the polarity indicated by $\mathit{pol}$ is denoted by $\lit{F}$.  Observe
that since $\true$ and $\false$ are not considered as atomic formulas but as
constant logical operators it holds that $\lit{\true} = \lit{\false} =
\emptyset$.

The notation $\var{E}$ and $\fun{E}$ is also used with sets~$E$ of terms or
formulas, where it stands for the union of values of the respective function
applied to each member of~$E$.

\paragraph{Semantic Relationships.}

We write $F \entails G$ for \name{$F$ entails $G$}, and $F \equiv G$ for
\name{$F$ is equivalent to $G$}, that is, \name{$F \entails G$ and $G \entails
  F$}.  On occasion we write a sequence of statements with these operators
where the right and left, respectively, arguments of subsequent statements are
identical in a chained way, such as, for example, $F \entails G \entails H$
for \name{$F \entails G$ and $G \entails H$}.

\section{Clausal First-Order Tableaux}
\label{sec-tableaux}

The following definition makes the variant of clausal tableaux that we use as
basis for interpolation precise. It is targeted at modeling tableau structures
produced by efficient fully automated first-order proving systems based on
different calculi.\hspace{-1em}
\begin{defn}[Clausal Tableau and Related Notions]

\sdlab{def-tab}
Let $F$ be a clausal formula.  A \defname{clausal tableau} (briefly
\name{tableau}) \defname{for} $F$ is a finite ordered tree whose nodes~$N$
with exception of the root are labeled with a literal, denoted by $\nlit{N}$,
such that the following condition is met: For each node~$N$ of the tableau the
disjunction of the labels of all its children in their left-to-right order,
denoted by $\nclause{N}$, is an instance of a clause in~$F$. A value of
$\nclause{N}$ for a node~$N$ in a tableau is called a \defname{clause of} the
tableau.

\smallskip

\sdlab{def-tab-closed} A branch of a tableau is called \defname{closed} if and
only if it contains nodes with complementary literals.  A node~$N$ is called
\defname{closed} if and only if all branches through $N$ are closed. A tableau
is called \defname{closed} if and only if its root is closed.

\smallskip

\sdlab{def-tab-closing} A node of a tableau is called \defname{closing} if and
only if it has an ancestor with complementary literal.  With a closing
node~$N$, a particular such ancestor is associated as \defname{target of}~$N$,
written~$\ntgt{N}$. A tableau is called \defname{leaf-closed} if and only if
all of its leaves are closing.

\smallskip
 
\sdlab{def-tab-ground} A tableau is called \defname{ground} if and only if for
all its nodes~$N$ it holds that $\nlit{N}$ is ground.
\end{defn}
The most immediate relationship of clausal tableaux to the semantics of
clausal formulas is that the universal closure of a clausal formula is
unsatisfiable if and only if there exists a closed clausal tableaux for the
clausal formula. Knowing that there are sound and complete calculi that
operate by constructing a closed clausal tableau for an unsatisfiable clausal
formula, and taking into account Herbrand's theorem we can state the following
proposition.%
\begin{prop}[Unsatisfiability and Computation of Closed Clausal Tableaux]
\label{prop-tab-complete}
There is an effective method that computes from a clausal formula~$F$ a closed
clausal tableau for~$F$ if and only if $\forall x_1 \ldots \forall x_n\, F$,
where $\{x_1,\ldots,x_n\} = \var{F}$, is unsatisfiable.  Moreover, this also
holds if terms in the literal labels of tableau nodes are constrained to
ground terms formed from functions occurring in~$F$ and, in case there is no
constant occurring in $F$, an additional fresh constant.

\end{prop}

Our objective is here interpolant construction on the basis of clausal
tableaux produced by fully automated systems. This has effect on some aspects
of our formal notion of \name{clausal tableau}: All occurrences of variables
in the literal labels of a tableau according to Definition~\ref{def-tab} are
free and the scope of these variables spans all literal labels of the whole
tableau.  In more technical terms, this means that the tableaux are
\name{free-variable tableaux} (see \cite[p.~158ff]{handbook:tableaux:letz})
with \name{rigid} variables (see \cite[p.~114]{handbook:ar:haehnle}).

The \name{closing} property of nodes, the \name{target} function, and the
\name{leaf-closed} property of tableaux are introduced to facilitate the
definition of interpolants by induction on the tableau structure, bottom-up,
with leaves considered in the base cases. It is easy to see that any closing
node is also a closed node and that any leaf-closed tableau is also a closed
tableau.  Any closed tableau for a clausal formula $F$ can be easily converted
to a \defname{leaf-closed} tableau for~$F$ by removing all edges that
originate in nodes that are closing.  Practical proving methods typically
construct a closed tableau that is already leaf-closed, since for the
construction of a closed tableau it is pointless to attach children to a
closing, and hence closed, node.

\section{Ground Interpolation with Clausal Tableaux}
\label{sec-ipol-basic}

Craig's interpolation theorem along with Lyndon's observation ensures the
existence of \name{Craig-Lyndon interpolants}, defined as follows.
\begin{defn}[Craig-Lyndon Interpolant]
\label{def-cli}
Let $F, G$ be formulas such that $F \entails G$.  A \defname{Craig-Lyndon
  interpolant} of $F$ and $G$ is a formula $H$ such that
\begin{enumerate}
\item \label{def-cli-sem} $F \entails H \entails G$.
\item \label{def-cli-pred} 
$\voc{H} \subseteq \voc{F} \cap \voc{G}$.
\item \label{def-cli-var}
$\free{H} \subseteq \free{F} \cap \free{G}$.
\end{enumerate}
\end{defn}
If condition (\ref{def-cli-pred}) is replaced by the weaker condition
$\vocplain{H} \subseteq \vocplain{F} \cap \vocplain{G}$, that is, the polarity
of predicate occurrences in~$H$ is not taken into account, then we call~$H$ a
\defname{Craig interpolant}.  Many automated reasoning techniques, including
the clausal tableau methods considered here, validate an entailment $F
\entails G$ by showing the unsatisfiability of $F \land \lnot G$. In the
literature on interpolation in verification this is reflected in the notion of
\name{reverse interpolant}.
\begin{defn}[Reverse Craig-Lyndon Interpolant]
  Let $F, \GR$ be formulas such that $F \land \GR \entails \false$.  A
  \defname{reverse Craig-Lyndon interpolant} of $F$ and $\GR$ is then defined
  as a Craig-Lyndon interpolant of $F$ and $\lnot \GR$.
\end{defn}
Our interpolant construction is based on a variant of clausal tableaux where
nodes have an additional \name{side} label that is shared by siblings and
indicates whether the tableau clause is an instance of an input clause derived
from the formula $F$ or the formula $G$ of the statement $F \land \GR \entails
\false$ underlying the reverse interpolant.

\begin{defn}[\Sided Clausal Tableau and Related Notions]

\sdlab{def-coltab} Let $\FL, \GR$ be clausal formulas.  A \defname{\sided
  clausal tableau for} $\FL$ \defname{and} $\GR$ (or briefly \defname{tableau}
for the \defname{two} formulas) is a clausal tableau for $\FL \land \GR$ whose
nodes~$N$ with exception of the root are labeled additionally with a
\name{side} $\nside{N} \in \{\aaa, \bbb\}$, such that the following conditions
are met by all nodes $N, N'$ of the tableau:
\begin{enumerate}
\item If $N$ and
$N^\prime$ are siblings, then $\nside{N} = \nside{N^\prime}$.
\item If $N$ has a child $N'$ with $\nside{N^\prime} = \aaa$, then
  $\nclause{N}$ is an instance of a clause in $\FL$.  If~$N$ has a child $N'$
  with $\nside{N^\prime} = \bbb$, then $\nclause{N}$ is an instance of a
  clause in $\GR$.
\end{enumerate}
We also refer to the side of the children of a node~$N$ as \defname{side of}
$\nclause{N}$.

\medskip

\sdlab{def-path} For $\mathit{side} \in \{\aaa,\bbb\}$ and all nodes $N$ of
a \sided clausal tableau define
\[\npath{\mathit{side}}{N}\; \eqdef \bigwedge_{N^\prime \in \mathit{Path} \text{ and }
\nside{N^\prime} = \mathit{side}}\hspace{-3em} \nlit{N^\prime},\] where
$\mathit{Path}$ is the union of the set of the ancestors of $N$ and $\{N\}$.
\end{defn}
For examples of \sided tableaux see Figure~\ref{fig-tab-bu}
and~\ref{fig-tab-td} in Example~\ref{ex-ipol-prop} below (ignore the node
annotations in brackets).  For examples of $\f{path}$ consider
Figure~\ref{fig-tab-td}.  Let $N$ be the rightmost node labeled with~$\fe$
there (above the rightmost leaf labeled with \tbbbtxt{$\lnot \fe$}). Then
$\npath{\aaa}{N} = \lnot \fa \land \fe$ and $\npath{\bbb}{N} = \lnot \fd \land
\lnot \fc$.

\label{page-occurrences}
The literature on interpolation targeted at applications in verification uses
a conceptualization with varying terminology (see \cite{bonacina:15:ground}
for an overview), which is related to the \name{side} labeling but would not
exactly match our needs. The main conceptual difference is that our labeling
applies to \emph{occurrences} of symbols, terms and formulas, identified by
tableau nodes, in contrast to symbols, terms and formulas themselves.
Analogously, first-order interpolation methods based on resolution with
paramodulation are specified with labelings that refer to \emph{occurrences}
in proof trees, \name{is from} \cite{huang:95} and \name{label}
\cite{bonacina:15:on}.%
\footnote{\label{foot-color}Symbols that appear only in $F$ or $\GR$ are
  called \name{$F$-colored} and \name{$\GR$-colored}, respectively, or
  \name{colored} in general (also \name{local} has been used synonymously to
  \name{colored}), while the other symbols are called \name{transparent} (or
  \name{grey}).  It appears that the association with occurrences is in
  particular necessary to take account of predicate polarity required by
  Craig-Lyndon (in contrast to just \name{Craig}) interpolants.  Another
  reason for the occurrence-based labeling is that the possibility that some
  clause $C$ is an instance of a clause in $F$ as well as an instance of a
  clause in $\GR$ should be retained.  An \emph{occurrence} of~$C$ then can be
  associated with either side label.  A perhaps confusing aspect of the
  color-based terminology is that the term \name{colored} on the one hand
  indicates for a \emph{symbol} that it appears only in one of $F$ or $\GR$, a
  property called \name{isolated} in \cite{craig:uses}, while, on the other
  hand, the \name{colored} property for compound structures, that is, terms
  and formulas, permits occurrences of transparent symbols.  We apply the
  color-based terminology in Sect.~\ref{sec-simul} and discuss it further in
  the context of first-order interpolation at the end of
  Sect.~\ref{sec-lift-thm}. Limitations of the color-based approaches are also
  discussed in \cite[Sect.~3]{bonacina:15:on}.}

\begin{defn}[Interpolant Extraction from a Clausal Tableau]
  \label{def-ipol}
Let $N$ be a node of a leaf-closed \sided clausal tableau.  The value of
\defname{$\nipol{N}$} is a quantifier-free formula, defined inductively as
follows:
\begin{enumerate}[label={\roman*}.,leftmargin=2em]
\item If $N$ is a leaf, then the value of $\nipol{N}$ is determined by the
  values of $\nside{N}$ and $\nside{\ntgt{N}}$ as specified in the
  following table:
\end{enumerate}
\[
\begin{array}{c@{\hspace{1em}}c@{\hspace{1em}}c}
\nside{N} &  \nside{\ntgt{N}} & \nipol{N}\\\midrule
\aaa & \aaa & \false\\[0.5ex]
\aaa & \bbb & \nlit{N}\\[0.5ex]
\bbb & \aaa & \du{\nlit{N}}\\[0.5ex]
\bbb & \bbb & \true
\end{array}
\]

\begin{enumerate}[label={\roman*}.,leftmargin=2em]
\setcounter{enumi}{1}
\item \label{item-children} If $N$ is an inner node with children $N_1,
  \ldots, N_n$ where $n \geq 1$, then the value of $\nipol{N}$ is composed
  from the values of $\f{ipol}$ for the children, disjunctively or
  conjunctively, depending on the side label of the children (which is the
  same for all of them), as specified in the following table:
\end{enumerate}
\[
\begin{array}{c@{\hspace{1em}}c}
\nside{N_1} & \nipol{N}\\\midrule
\aaa & \bigvee_{i=1}^n \nipol{N_i}\\[1ex]
\bbb & \bigwedge_{i=1}^n \nipol{N_i}
\end{array}
\]
\end{defn}

The following lemma gives semantic and syntactic properties of the formula
returned by applying $\f{ipol}$ to the root of a leaf-closed ground tableau
for ground formulas.\footnote{The $\f{ipol}$ function is also defined on
  non-ground tableaux, as it is independent of occurrences of
  variables. However, the association with properties relevant for
  interpolation is simplified if we assume a ground tableau, which is without
  loss of generality: A calculus may construct as proof a closed free-variable
  tableau with occurrences of free rigid variables. Any ground instantiation
  of these variables yields a ground tableau that provides a proof of the same
  formula. In particular, an instantiation of each variable by a dedicated
  constant. For the considered properties, tableaux with variables can be
  represented by ground tableaux with such constants.  The restriction that
  the tableau is for \emph{ground} formulas provides another simplification
  that, however, does not restrict the applicability of the lemma in a
  first-order context: A ground tableau that provides a proof of a clausal
  \emph{first-order} formula also provides a proof of a clausal \emph{ground}
  formula, the conjunction of the tableau clauses. In the setting of \sided
  tableaux this can be stated more precisely as: A leaf-closed \sided ground
  tableau for two clausal \emph{first-order} formulas is also a leaf-closed
  \sided ground tableau for two clausal \emph{ground} formulas, the
  conjunction of the tableau clauses with side~$\aaa$ and the conjunction of
  the tableau clauses with side~$\bbb$.}  These properties imply the
conditions required from a Craig-Lyndon interpolant as specified in
Definition~\ref{def-cli}.
\begin{lem}[Correctness of Interpolant Extraction from Clausal Ground Tableaux]
\label{lem-ground-ipol-correct}
Let $\FL, \GR$ be clausal ground formulas and let $N_0$ be the root of a
leaf-closed \sided clausal ground tableau for $\FL$ and $\GR$. Then
\begin{enumerate}
\item $\FL \entails \nipol{N_0} \entails \lnotGR$.
\item $\lit{\nipol{N_0}} \subseteq \lit{\FL} \cap \lit{\lnotGR}$.
\end{enumerate}
\end{lem}
\begin{proof}
We show that the following invariant of $\f{ipol}$ holds for all nodes~$N$ of
the tableau, including the root, which immediately implies the proposition:
\begin{enumerate}[label={(\alph*)},leftmargin=2em]
\item \label{item-lem-gi-sem}
 $\FL \land \npathL{N} \entails \nipol{N}
 \entails \lnotGR  \lor \lnot \npathR{N}$.
\item \label{item-lem-gi-syn}
$\lit{\nipol{N}} \subseteq
\lit{\FL \land \npathL{N}} \cap
\lit{\lnotGR  \lor \lnot \npathR{N}}$.
\end{enumerate}

\noindent
This is proven by induction on the tableau structure, proceeding from leaves
upwards.  We prove the base case, where $N$ is a leaf, by showing
\ref{item-lem-gi-sem} and \ref{item-lem-gi-syn} for all possible values of
$\nside{N}$ and $\nside{\ntgt{N}}$:

  \begin{itemize}
  \item Case $\nside{N} = \aaa$:
    \smallskip
    \begin{itemize}
    \item Case $\nside{\ntgt{N}} = \aaa$: Immediate since then
      $\npathL{N} \entails \false = \nipol{N}$.
    \item Case $\nside{\ntgt{N}} = \bbb$: Then $\nipol{N} = \nlit{N}$.
      Properties \ref{item-lem-gi-sem} and \ref{item-lem-gi-syn} follow
      because $\nlit{N}$ is a conjunct in $\npathL{N}$ and $\du{\nlit{N}}$
      is a conjunct in $\npathR{N}$.
    \end{itemize}
    \smallskip
  \item Case $\nside{N} = \bbb$:
    \smallskip
    \begin{itemize}
      \item Case $\nside{\ntgt{N}} = \aaa$: Then $\nipol{N} =
        \du{\nlit{N}}$.  Properties \ref{item-lem-gi-sem} and
        \ref{item-lem-gi-syn} follow because $\du{\nlit{N}}$ is a conjunct in
        $\npathL{N}$ and $\nlit{N}$ is a conjunct in $\npathR{N}$.
      \item Case $\nside{\ntgt{N}} = \bbb$: 
        Immediate since then
        $\nipol{N} = \true \entails \lnot \npathR{N}$.
    \end{itemize}
  \end{itemize}

\noindent
To show the induction step, assume that $N$ is an inner node with children
$N_1, \ldots, N_n$. Consider the case where the side of the children is
$\aaa$.  The induction step for the case where the side of the children is
$\bbb$ can be shown analogously.  By the induction hypothesis we can assume
that for all $i \in \{1,\ldots, n\}$ it holds that
\[\FL \land \npathL{N_i} \entails \nipol{N_i}
 \entails \lnotGR \lor \lnot \npathR{N_i},\]
which, since $\nside{N_i} = \aaa$, is equivalent to
\[\FL \land \npathL{N} \land \nlit{N_i} \entails \nipol{N_i}
 \entails \lnotGR  \lor \lnot \npathR{N}.\]
This implies
\[\FL \land \npathL{N} \land \bigvee_{i=1}^{n} \nlit{N_i} \entails
  \bigvee_{i=1}^{n} \nipol{N_i}
 \entails \lnotGR  \lor \lnot \npathR{N}.\]
 Since $\nipol{N} = \bigvee_{i=1}^{n} \nipol{N_i}$,
 according to the definition of $\f{ipol}$ for
 nodes $N$ whose children have side $\aaa$,
it follows that
\[\FL \land \npathL{N} \land \bigvee_{i = 1}^{n}{\nlit{N_i}}
\entails \nipol{N}
 \entails \lnotGR \lor \lnot \npathR{N}.\] 
 Because $\bigvee_{i = 1}^{n}{\nlit{N_i}} = \nclause{N}$ is by construction of
 the tableau a clause in $\FL$ and thus entailed by $\FL$ the semantic
 requirement \ref{item-lem-gi-sem} of the induction conclusion follows:
\[\FL \land \npathL{N} \entails \nipol{N}
 \entails \lnotGR \lor \lnot \npathR{N}.\]
The syntactic requirement \ref{item-lem-gi-syn} follows from the induction
hypothesis and because in general for all nodes~$N$ of a \sided clausal ground
tableau for clausal ground formulas $\FL$ and $\GR$ it holds that all literals
in $\npathL{N}$ occur in some clause of $\FL$ and all literals in $\npathR{N}$
occur in some clause of $\GR$.  \qed
\end{proof}

Lemma~\ref{lem-ground-ipol-correct} immediately yields a construction method
for Craig-Lyndon interpolants of propositional and, more general, first-order
formulas that are ground (and without equality, except if represented as
predicate, see Sect.~\ref{sec-cli-related-equality}).
We call the procedure \name{\CTIG}, suggesting \name{Clausal Tableau
  Interpolation}.  In Sect.~\ref{sec-lift-procedure} below it will be
generalized to first-order sentences in full.

\begin{proc}[The \CTIG Method for Craig-Lyndon Interpolation on
    Ground Formulas]\hspace{-10pt}
  \label{proc-ctig}

  \algoinput Ground formulas $F$ and $G$ such that $F \entails G$.

  \algomethod
  \begin{enumerate}
  \item \name{Clausification.}
    Convert $F$ and $\lnot G$ to equivalent
    clausal ground formulas $F'$ and $G'$, respectively, such that
    $\voc{F'} \subseteq \voc{F}$ and $\voc{G'} \subseteq \voc{\lnot
      G}$.

    \item \name{Tableau computation.}  Compute a closed clausal ground tableau
      for $F' \land G^{\prime}$.  If the tableau is not already leaf-closed,
      convert it to leaf-closed form by removing all edges that originate in
      closing nodes.
      
    \item \label{step-sa} \name{Side assignment.}  Convert the ground tableau
      to a \sided tableau for $F'$ and $G'$ by attaching appropriate
      \name{side} labels to all nodes except the root. This is always possible
      because every clause of the tableau is in $F'$ or in $G'$. (As $F$ and
      $G$ can be arbitrary ground formulas it is possible that a clause of the
      tableau appears in both $F'$ \emph{and} in ~$G'$. For siblings
      corresponding to such a clause, the side labels can be either set all to
      $\aaa$ or all to $\bbb$. As shown in Example~\ref{examp-sides} below,
      this choice can have an effect on the computed interpolant. The issue is
      discussed further in Sect.~\ref{sec-grounding}.)

    \item \name{Interpolant extraction.}  Let $H$ be the value of
      $\nipol{N_0}$, where $N_0$ is the root of the tableau.

  \end{enumerate}
  
  \algooutput Return~$H$.  The output is a ground formula that is a
  Craig-Lyndon interpolant of the input formulas.
\end{proc}

\noindent
That the procedure is correct, which means that it outputs indeed a
Craig-Lyndon interpolant of the input formulas, follows from
Lemma~\ref{lem-ground-ipol-correct} and Definition~\ref{def-cli}.  The
procedure computes an interpolant on the basis of \emph{any} leaf-closed
\sided clausal ground tableaux for the clausified inputs $F' \land G'$, which
follows from the definition of $\f{ipol}$ and
Lemma~\ref{lem-ground-ipol-correct}.  (For interpolation systems that operate
on resolution proofs, the analog to this property is termed \name{complete
  relative to an inference system} in \cite[Def.~10]{bonacina:15:ground}.)
Hence, if the invoked method for tableau computation is complete, that is, it
outputs a closed clausal tableau for all unsatisfiable inputs, then the
overall interpolation procedure is also complete, that is, it outputs a
Craig-Lyndon interpolant of its inputs $F$ and $G$ whenever $F \entails G$.

The size of the resulting formula is linear in the size of the tableau, or
more precisely in the number of its leaves whose target has the opposite side
label (assuming that truth value simplification\footnote{\label{foot-tv}
  Exhaustively rewriting with $F \land \false \equiv \false$, $F \lor \true
  \equiv \true$, $F \land \true \equiv F$, $F \lor \false \equiv F$, modulo
  commutativity.} is integrated into the computation of $\f{ipol}$). The size
of the tableau itself is not polynomially bounded in the size of the clausal
formulas underlying the tableau construction.\footnote{An argument for this is
  that even clausal tableaux with atomic cut, which can polynomially simulate
  tree resolution, can not polynomially simulate unrestricted resolution where
  proofs may have the form of dags \cite{letz:habil}.}

The potential blow up in the transformation to clausal form can be avoided by
replacing this step with the transformation to a structure preserving (also
known as \name{definitional}) normal form, which needs to be applied for
interpolation in a way such that the sets of auxiliary symbols introduced for
transforming~$F$ and~$\lnot G$ are disjoint. This is discussed further in
Sect.~\ref{sec-preproc}.

\begin{examp}[Propositional Interpolation with Clausal Tableaux]
  \label{ex-ipol-prop}
  To facilitate comparison with other methods for computing reverse
  interpolants for propositional clausal formulas we consider interpolation
  inputs from \cite[Example~2]{bonacina:15:ground}: Let $\FL = \FL' = (\fa
  \lor \fe) \land (\lnot \fa \lor \fb) \land (\lnot \fa \lor \fc)$ and let
  $\lnot G = \GR' = (\lnot \fb \lor \lnot \fc \lor \fd) \land \lnot \fd \land
  \lnot \fe$.  As there are usually quite different leaf-closed clausal
  tableaux for a given clausal formulas, we illustrate the assignment of
  values of $\f{ipol}$ for two different such tableaux, both for $\FL'$ and
  $\GR'$. Figure~\ref{fig-tab-bu} represents a typical result of a bottom-up
  calculus and Figure~\ref{fig-tab-td} of a top-down calculus, as further
  discussed below. In these figures, nodes with exception of the root are
  represented by their literal label.  Nodes with \tbbbtxt{side $\bbb$} are
  indicated by gray background.  The remaining nodes have side $\aaa$.  For
  each node the value of $\f{ipol}$ is annotated in brackets, where values are
  shown exactly as specified in Definition~\ref{def-ipol} except that
  truth-value simplification (see footnote~\ref{foot-tv}) is applied.

  \smallskip
  
  \noindent
  \begin{minipage}[t]{0.345\textwidth}
      \centering
      \hspace{-12pt}%
  \begin{tikzpicture}[scale=0.85,
      baseline=(a.north),
      sibling distance=9em,level distance=\extabld,
      every node/.style = {transform shape,anchor=mid}]]
      \node (a) {\vbar\textbullet\nannotw{$(\fb\land\fc) \lor \fe$}}
      child { node {$\taaa{\fa}$\nannot{$\fb\land\fc$}}
        [sibling distance=5em]
        child { node {$\taaa{\lnot \fa}$\nannot{$\false$}} }
        child { node {$\taaa{\fb}$\nannot{$\fb\land\fc$}}
          child { node {$\taaa{\lnot \fa}$\nannot{$\false$}} }
          child { node {$\taaa{\fc}$\nannot{$\fb\land\fc$}}
            [sibling distance=4em]
            child { node {$\tbbb{\lnot \fb}$\nannot{$\fb$}} }
            child { node {$\tbbb{\lnot \fc}$\nannot{$\fc$}} }
            child { node {$\tbbb{\fd}$\nannot{$\true$}}
              child { node {$\tbbb{\lnot \fd}$\nannot{$\true$}} }
            }
          }
        }
      }
      child { node {$\taaa{\fe}$\nannot{$\fe$}}
        child { node {$\tbbb{\lnot \fe}$\nannot{$\fe$}}}
      };
  \end{tikzpicture}
  
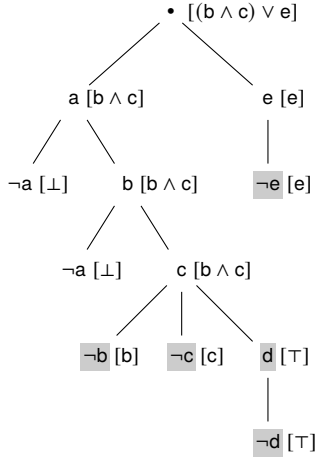
\captionof{figure}{Interpolation with
   a bottom-up constructed tableau.}
  \label{fig-tab-bu}
  \end{minipage}
   \hspace{18pt}%
   \begin{minipage}[t]{0.605\textwidth}
     \centering
     \hspace{-12pt}%
  \begin{tikzpicture}[scale=\extabscale,
      baseline=(a.north),
      sibling distance=9em,level distance=\extabld,
      every node/.style = {transform shape,anchor=mid}]]
      \node (a) {\vbar\textbullet%
        \nannotw{$(\fe \lor \fb) \land (\fe \lor \fc)$}}
      child { node {$\tbbb{\lnot \fd}$%
          \nannot{$(\fe \lor \fb) \land (\fe \lor \fc)$}}
        child { node {$\tbbb{\lnot \fb}$\nannot{$\fe \lor \fb$}}
          [sibling distance=5em]
          child { node {$\taaa{\lnot \fa}$\nannot{$\fe$}}
            child { node {$\taaa{\fa}$\nannot{$\false$}}}
            child { node {$\taaa{\fe}$\nannot{$\fe$}}
              child { node {$\tbbb{\lnot \fe}$\nannot{$\fe$}}
              }}
          }
          child { node {$\taaa{\fb}$\nannot{$\fb$}} }
        }
        child { node {$\tbbb{\lnot \fc}$\nannot{$\fe \lor \fc$}}
          [sibling distance=5em]
          child { node {$\taaa{\lnot \fa}$\nannot{$\fe$}}
            child { node {$\taaa{\fa}$\nannot{$\false$}}}
            child { node {$\taaa{\fe}$\nannot{$\fe$}}
              child { node {$\tbbb{\lnot \fe}$\nannot{$\fe$}}
              }
            }
          }
          child { node {$\taaa{\fc}$\nannot{$\fc$}} }
        }
        child { node {$\tbbb{\fd}$\nannot{$\true$}} }
      };
  \end{tikzpicture}
  
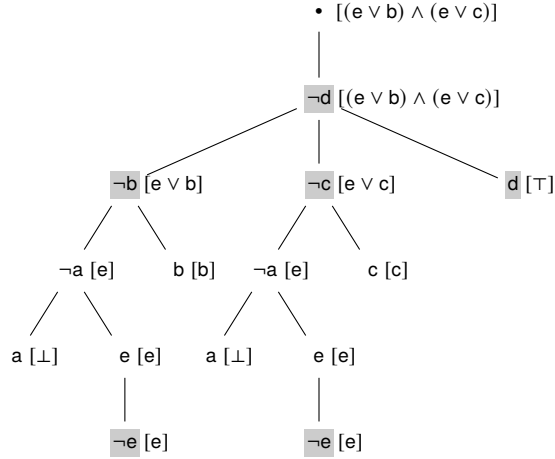
\captionof{figure}{Interpolation with a top-down constructed tableau.}
  \label{fig-tab-td}
  \end{minipage}
   \vspace{0.2cm}
\end{examp}

The way in which a tableau was constructed step-by-step by a calculus is
irrelevant for interpolant extraction. Nevertheless, the tableaux in
Figure~\ref{fig-tab-bu} and~\ref{fig-tab-td} suggest two different specific
construction possibilities which we will sketch now. Both might have been
constructed with calculi that start with the root and repeatedly extend an
open (i.e., not closed) branch by attaching a clause. In the tableau of
Figure~\ref{fig-tab-bu} a clause is only attached if the complements of all
its negative literals appear on the branch, with the effect that the tableau
never has an open branch that ends with a negative literal.  The clause $\fa
\lor \fe$, for example, has no negative literal and can thus be attached
directly below the root.  The clause $\lnot \fb \lor \lnot \fc \lor \fd$, for
example, is attached to a branch in which $\fb$ and $\fc$ appear. Branches
ending in $\lnot \fb$ and $\lnot \fc$ are created but immediately closed.
Although presented as a tableau construction that starts with the root node,
calculi that proceed in the indicated way may be understood as bottom-up
methods, as they start from disjunctions of facts (positive clauses) and the
extension of a branch by a clause is like the forward application of a rule:
The antecedent is on the branch in the form of the complements of the negative
clause literals and the consequent corresponds to its positive literals, each
forming the end of a new open branch.  Extension with a goal (negative clause)
closes an open branch without introducing a new one.

The tableau of Figure~\ref{fig-tab-td} might have been constructed with a
top-down method. Some clause is chosen as start clause and attached below the
root, $\lnot \fd$ in the example. Open branches are only extended with clauses
such that the \emph{connection condition} is preserved, that is, the
complement of the last literal on the branch appears in the clause.  The
complement of the last literal on an open branch is viewed as an open subgoal.
Extending the branch by attaching a clause replaces it with other subgoals or,
if all literals of the clause have complements on the branch, solves it.
Top-down tableau construction typically may lead to the necessity to solve the
same subgoal more than once, as can be seen in the example with the replicated
subtree at the two nodes labeled with $\lnot \fa$.

Returning to the question of interpolant extraction from a given clausal
tableaux, the following example shows a case where the obtained interpolant
depends on a choice of the \name{side assignment} step of
Procedure~\ref{proc-ctig}.
\begin{examp}[Alternate Possible Side Assignments]
  \label{examp-sides}
   Let $F = \fa \land \fb \land (\fb \imp \fc)$ and $G = \fc \lor \lnot (\fb
   \imp \fc) \lor \fd$. Figures~\ref{fig-side-one} and~\ref{fig-side-two} each
   show a leaf-closed \sided clausal tableau for the corresponding clausal
   formulas $F' = \fa \land \fb \land (\lnot \fb \lor \fc)$ and $G' = \lnot
   \fc \land (\lnot \fb \lor \fc) \land \lnot \fd$ according to
   Procedure~\ref{proc-ctig}.  The tableaux are identical, except that for the
   clause $\lnot \fb \lor \fc$, which is in $F'$ and also in $G'$, the side
   assigned to its occurrence in the tableau is different. The symbolism for
   indicating sides and values of $\f{ipol}$ is as in
   Example~\ref{ex-ipol-prop}. The overall interpolants of~$F$ and~$G$
   returned by the procedure are $\fc$ and~$\fb$, respectively.

  \noindent
  \hfill
  \begin{minipage}[t]{0.4\textwidth}
      \centering
      \hspace{-4pt}%
  \begin{tikzpicture}[scale=0.85,
      baseline=(a.north),
      sibling distance=9em,level distance=\extabld,
      every node/.style = {transform shape,anchor=mid}]]
      \node (a) {\vbar\textbullet\nannotw{$\fc$}}
      child { node {$\tbbb{\lnot \fc}$\nannot{$\fc$}}
        [sibling distance=5em]
        child { node {$\taaa{\lnot\fb}$\nannot{$\false$}}
          child { node {$\taaa{\fb}$\nannot{$\false$}}}}
        child { node {$\taaa{\fc}$\nannot{$\fc$}} }
      };
  \end{tikzpicture}
  \captionof{figure}{Interpolation with $\aaa$ as side of $\lnot \fb \lor
    \fc$.}
  \label{fig-side-one}
  \end{minipage}
  \hfill%
  \begin{minipage}[t]{0.4\textwidth}
      \centering
      \hspace{-4pt}%
      \begin{tikzpicture}[scale=0.85,
      baseline=(a.north),
      sibling distance=9em,level distance=\extabld,
      every node/.style = {transform shape,anchor=mid}]]
      \node (a) {\vbar\textbullet\nannotw{$\fb$}}
      child { node {$\tbbb{\lnot \fc}$\nannot{$\fb$}}
        [sibling distance=5em]
        child { node {$\tbbb{\lnot\fb}$\nannot{$\fb$}}
          child { node {$\taaa{\fb}$\nannot{$\fb$}}}}
        child { node {$\tbbb{\fc}$\nannot{$\true$}} }
      };
      \end{tikzpicture}
      \captionof{figure}{Interpolation with $\bbb$ as side of $\lnot \fb \lor
    \fc$.}
      \label{fig-side-two}
      \end{minipage}
  \hspace*{\fill}
\end{examp}

\section{Related Methods for Propositional Interpolation}
\label{sec-simul}

The extraction of interpolants from clausal tableaux with the $\f{ipol}$
function emerged from a straightforward adaptation of the propositional part of
Smullyan's interpolation method for analytic tableaux described in
\cite[Chap.~XV]{smullyan:book:68} and \cite[Chap.~8.12]{fitting:book}.
Differently from these works, the adaptation specifies the interpolant
extraction not in terms of tableau manipulation rules that deconstruct the
tableau bottom-up, but inductively, as a function that maps each node to a
formula.
It turned out that the clausal adaptation of Smullyan's method has some
striking similarities to a family of interpolation methods for resolution
proofs including, for example, the algorithms from \cite{huang:95} and
\cite{mcmillan:2003}, surveyed in a common framework in
\cite{bonacina:15:ground}.  Like the tableau-based method, these interpolation
methods for resolution proofs are specified by an induction on a proof
structure, which led to the approach being called \name{inductive}
\cite{bonacina:15:ground}.

The role of the clausal tableau for an unsatisfiable clausal formula~$F$ is
taken by a deduction tree \cite{chang:lee} for $F$, that is, a finite ordered
tree that represents a resolution proof. Its nodes are labeled with clauses
such that each leaf is labeled by a clause in~$F$, each inner node has exactly
two children and is labeled by a binary resolvent of the clauses of its
children, and the root is labeled with the empty clause $\Box$.  We assume
that the involved clauses do not contain duplicate literals by expecting this
from the clauses in~$F$ and by assuming that the used variant of binary
resolution incorporates merging of duplicate literals.

The invariants that justify the involved inductions for tableaux and for
deduction trees are actually quite similar.  For clausal tableaux, they have
been stated above as~(a) and~(b) in the proof of
Lemma~\ref{lem-ground-ipol-correct}. For deduction trees they can be expressed
as follows: Let~$\f{pi}$ be a function that, analogously to $\f{ipol}$, maps
each node in the deduction tree to a formula, the \name{partial interpolant}
\cite{bonacina:15:ground} associated with the node, let $\f{clause}(M)$ denote
the clause label of deduction tree node~$M$, let $C|_{\aaa}$ denote the clause
$C$ with all literals deleted whose predicate does not occur in $\FL$, and,
analogously, let $C|_{\bbb}$ denote $C$ with all literals deleted whose
predicate does not occur in $\GR$.  Based on \cite[Proof of
  Theorem~2]{huang:95} and \cite[Proposition~9]{bonacina:15:ground}, the
invariants expected to hold for all nodes $M$ of the deduction tree can then
be stated as
\begin{enumerate}[label={(\alph*$^{\prime}$)},leftmargin=2em]
\item $\FL \land \lnot \f{clause}(M)|_{\aaa} \entails \cpi{M} \entails \lnotGR
  \lor \f{clause}(M)|_{\bbb}$.
\item $\predplain{\cpi{N}} \subseteq \predplain{\FL} \cap \predplain{\GR}$.
\end{enumerate}
If $M_0$ is the root of the deduction tree, labeled with the empty clause
$\Box$, then $\f{clause}(M_0)|_{\aaa} = \f{clause}(M_0)|_{\bbb} = \false$ and
the invariants~(a$^\prime$) and~(b$^\prime$) imply that $\cpi{M_0}$ is an
interpolant of~$\FL$ and~$\lnotGR$.  There is a parallelism between the roles
of the clause projections $\f{clause}(M)|_{\aaa}$ and $\f{clause}(M)|_{\bbb}$
in invariant (a$^\prime$) for deduction trees and the path projections
$\npath{\aaa}{N}$ and $\npath{\bbb}{N}$ in invariant~(a) of
Lemma~\ref{lem-ground-ipol-correct} for clausal tableaux.

A profound account of these similarities is an issue for future research. As a
step in this direction, we show that various inductive propositional
interpolation methods for resolution proofs can be linearly simulated with our
clausal tableau method. We show this in detail for Huang's algorithm
\cite{huang:95}, an early method, whose fundamental role has been observed
only later with \cite{bonacina:11,bonacina:15:ground,bonacina:15:on}.  The
simulations suggest a new way to systematize interpolation methods for
resolution proofs and, as they involve just linear-time conversions from
tree-shaped resolution proofs, may be useful for practical implementation.

The basis for these simulations are clausal tableaux \emph{with atomic cuts}.
The \name{atomic cut} rule \cite{letz:cut:1994} permits to extend a clausal
tableau under construction by attaching to a node two successors, labeled with
$A$ and $\lnot A$, respectively, where $A$ is an arbitrary atom.
Equivalently, this can be considered as attaching siblings whose clause is the
tautology $\lnot A \lor A$, also if that tautology is not an input clause.
The atomic cut rule is of interest for theorem proving as it permits
substantially shorter proofs.  Simulation by clausal tableaux with atomic cut
is a central tool in the analysis of advanced construction techniques for
clausal tableaux as well as investigations of relationships of clausal tableau
methods to other calculi \cite{letz:habil,letz:cut:1994}. Of particular
relevance for our simulations of interpolation methods is that, for
propositional logic, clausal tableaux with atomic cut can linearly simulate
\name{tree resolution}, that is, resolution where the proof graph forms a tree
in contrast to an unrestricted dag, which has been shown via the semantic tree
method as an intermediate method in \cite[Prop~7.33
  and~7.40]{letz:habil}.\footnote{That semantic trees and tree resolution can
  simulate each other polynomially has been shown already in
  \cite{reckhow:thesis}, where also relationships to many other propositional
  systems are investigated.} The simulation is achieved
 with an encoding of a deduction tree as a closed clausal tableau with
atomic cuts, that is, a tableau that permits interspersed ``cuts'',
tautological clauses of the form $\lnot A \lor A$ that are not necessarily in
the input formula. Moreover, the clausal tableaux that encode resolution
deduction trees are in \defname{cut normal form} \cite{letz:habil}, that is,
all their clauses except those of nodes whose children are all leaves are
``cuts'' of the form $\lnot A \lor A$.

This encoding can be modified, preserving linearity and the cut normal form,
to produce a leaf-closed clausal tableau such that each node~$M$ of the
deduction tree can be mapped to a node~$N$ of the tableau, where the partial
interpolant $\cpi{M}$ associated with $M$ according to Huang's method is
syntactically identical to the value of $\nipol{N}$ as specified in
Definition~\ref{def-ipol}. The root of the deduction tree, labeled with
$\Box$, whose partial interpolant is the overall interpolation result, is
mapped there to the root of the clausal tableau.

We now show this simulation for Huang's method (in the special case where it
is applied to propositional formulas) in detail.  To specify Huang's method,
we use the ``color-based'' terminology, which is adequate here since we only
consider propositional interpolation, it facilitates the incorporation of
merging duplicate literals into resolution, and the objective is to compute
only \emph{Craig} (in contrast to Craig-\emph{Lyndon}) interpolants.
\begin{defn}[Huang's Partial Interpolant for Propositional Inputs]
  \label{def-huang}
  Let $\FL,\GR$ be propositional clausal formulas with no common clause and
  such that $\FL \land \GR$ is unsatisfiable.  Call an atom
  \defname{$\aaa$-colored} if it occurs in $\FL$ but not in $\GR$,
  \defname{$\bbb$-colored} if it occurs in $\GR$ but not in $\FL$, and
  \name{transparent} if it occurs in both $\FL$ and $\GR$.  For all nodes $M$
  of a deduction tree for $\FL \land \GR$ define the \defname{partial
    interpolant associated with $M$ according to Huang's method}~$\huang{M}$
  as follows:

  \begin{enumerate}[label={(\roman*)},leftmargin=2.5em]
  \item If $M$ is a leaf and its clause is in $\FL$, then $\huang{M} \eqdef
    \false$.
  \item If $M$ is a leaf and its clause is in $\GR$, then $\huang{M} \eqdef
    \true$.
  \item \label{huang-case-inner} If $M$ is an inner node with children $M_1,
    M_2$ and its clause is a resolvent of the clauses of $M_1$ and $M_2$ upon
     $A$ in the clause of $M_1$ and $\lnot A$ in the
    clause of $M_2$, then
    \begin{enumerate}[label={(\alph*)},leftmargin=2em]
    \item If $A$ is $\aaa$-colored, then $\huang{M} \eqdef \huang{M_1} \lor
      \huang{M_2}$.
    \item If $A$ is $\bbb$-colored, then $\huang{M} \eqdef \huang{M_1} \land
      \huang{M_2}.$
    \item \label{huang-case-transp} If $A$ is transparent, then $\huang{M}
      \eqdef (\huang{M_1} \land \lnot A) \lor (A \land \huang{M_2}).$
    \end{enumerate}
  \end{enumerate}
\end{defn}
If $M_0$ is the root of a deduction tree for $\FL \land \GR$ as specified in the
preconditions of Definition~\ref{def-huang}, then $\huang{M_0}$ is a Craig
interpolant of $\FL$ and $\lnot \GR$ \cite{huang:95,bonacina:15:ground}, which
can be shown via the invariants~(a$^\prime$) and~(b$^\prime$). For the case
\ref{huang-case-inner}.\ref{huang-case-transp} in Definition~\ref{def-huang} we
have, compared to \cite{huang:95}, flipped the sides of the first disjunct to
achieve an exact syntactic correspondence with our clausal tableaux encoding.
The original formula reads: $(\lnot A \land \huang{M_1}) \lor (A \land
\huang{M_2})$.

\begin{thm}[Simulation of Huang's Method for Propositional Inputs
    with Clausal Tableaux]
  \label{thm-simul-huang}
  Let $\FL,\GR$ be propositional clausal formulas with no common clause and
  such that $\FL \land \GR$ is unsatisfiable. Let $\CUTS_{\aaa} \eqdef
  \textstyle\bigwedge_{A \in \predplain{\FL}} (\lnot A \lor A)$ and
  $\CUTS_{\bbb} \eqdef \textstyle\bigwedge_{A \in \predplain{\GR}} (\lnot A
  \lor A)$. Then there is a linear-time algorithm that computes from any
  deduction tree for $\FL \land \GR$ a \sided leaf-closed clausal tableau for
  $\FL \land \CUTS_{\aaa}\;$ and $\;\GR \land \CUTS_{\bbb}$ with the property
  that if $M_0$ is the root of the deduction tree and $N_0$ is the root of the
  tableau, then
  \[\huang{M_0} = \nipol{N_0}.\]
\end{thm}

\begin{proof}[Sketch]
  Assume a given deduction tree for $\FL \land \GR$ and let $M_0$ be its root.
  We specify a function $\fct$ that maps each node~$M$ of the deduction tree
  to a node $N = \ct{M}$ of a leaf-closed \sided clausal tableau for $\FL
  \land \CUTS_{\aaa}\;$ and $\;\GR \land \CUTS_{\bbb}$ whose root is $\ct{M_0}$.
  It can be verified that for all nodes $M$ of the deduction tree it holds that
  \begin{equation}
    \label{eq-huang-ipol}
    \huang{M} = \nipol{\ct{M}}.\tag{$*$}
  \end{equation}
  It follows that $\ct{M_0}$ has all the properties in the claim,
  since~(\ref{eq-huang-ipol}) implies
  \[ \huang{M_0} = \nipol{\ct{M_0}}.
  \]
  The specification of $\fct$ below characterizes structural properties as
  well as the side and literal labels of the tableau nodes with the same case
  distinctions as the definition of Huang's method
  (Definition~\ref{def-huang}).  It is not difficult to derive from $\fct$ a
  linear-time procedure that constructs a clausal tableau which meets the
  specified characteristics by a traversal of the deduction tree.

  \begin{enumerate}[label={(\roman*)},leftmargin=2.5em]
  \item \label{ct-case-aaa} If $M$ is a leaf and its clause $C = L_1 \lor
    \ldots \lor L_n$ is in $\FL$, then $\ct{M}$ is a tableau node~$N$ with
    children that are leaves, have side~$\aaa$, target nodes with the same
    side, and literals such that $\nclause{N} = C$. (That target nodes with
    the required side exist is discussed below.)  The value of $\ct{M}$ is
    then the root of a tableau that can be depicted as follows:
    
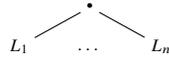
\begin{center}
    \begin{tikzpicture}[scale=\tableauscale,
        sibling distance=8em,level distance=5ex,
        every node/.style = {transform shape,anchor=mid}]]
        \node (a) {\vbar\textbullet}
        child { node (b) {$L_1$} }
        child { node (c) {$L_n$} };
        \path (b) -- node[auto=false]{\ldots} (c);
    \end{tikzpicture}
    \hspace*{2.5em}
    \end{center}    
    
  \item \label{ct-case-bbb} If $M$ is a leaf and its clause $C = L_1 \lor
    \ldots \lor L_n$ is in $\GR$, then $\ct{M}$ is a tableau node~$N$ with
    children that are leaves, have side~$\bbb$, target nodes with the same
    side, and literals such that $\nclause{N} = C$. (That target nodes with
    the required side exist follows analogously to case~\ref{ct-case-aaa}, for
    which this is discussed below.) The value of $\ct{M}$ is then the root of
    a tableau that can be depicted as follows:
    
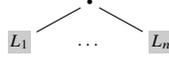
\begin{center}
    \begin{tikzpicture}[scale=\tableauscale,
        sibling distance=8em,level distance=5ex,
        every node/.style = {transform shape,anchor=mid}]]
        \node (a) {\vbar\textbullet}
        child { node (b) {$\tbbb{L_1}$} }
        child { node (c) {$\tbbb{L_n}$} };
        \path (b) -- node[auto=false]{\ldots} (c);
    \end{tikzpicture}
    \hspace*{2.5em}
    \end{center}
        
  \item \label{ct-case-inner} If $M$ is an inner node with children $M_1, M_2$
    and its clause is a resolvent of the clauses of $M_1$ and $M_2$ upon $A$
    in the clause of $M_1$ and $\lnot A$ in the clause of $M_2$, then
    \begin{enumerate}[label={(\alph*)},leftmargin=2.5em]
    \item \label{ct-case-ff} If $A$ is $\aaa$-colored, then $\ct{M}$ is a
      tableau node~$N$ with $\ct{M_1}$ and $\ct{M_2}$ as children, which have
      side~$\aaa$ and literals $\lnot A$ and $A$,
      respectively. The value of $\ct{M}$ is then the root of a tableau that
      can be depicted as follows:
      \begin{center}
      \begin{tikzpicture}[scale=\tableauscale,
          sibling distance=8em,level distance=5ex,
          every node/.style = {transform shape,anchor=mid}]]
          \node (a) {\vbar\textbullet}
          child { node {$\ct{M_1} : \lnot A$} \triaw{1.8em}{}} 
          child { node {$\ct{M_2} : A$} \triaw{1.8em}{}};
      \end{tikzpicture}
      \hspace*{5em}
      \end{center}
    \item \label{ct-case-gg} If $A$ is $\bbb$-colored, then $\ct{M}$ is a
      tableau node~$N$ with $\ct{M_1}$ and $\ct{M_2}$ as children, which have
      side~$\bbb$ and literals~$\lnot A$ and $A$,
      respectively. The value of $\ct{M}$ is then the root of a tableau that
      can be depicted as follows:
      \begin{center}
      \begin{tikzpicture}[scale=\tableauscale,
          sibling distance=8em,level distance=5ex,
          every node/.style = {transform shape,anchor=mid}]]
          \node (a) {\vbar\textbullet}
          child { node {$\ct{M_1} : \tbbb{\lnot A}$} \triaw{1.8em}{}} 
          child { node {$\ct{M_2} : \tbbb{A}$} \triaw{1.8em}{}};
      \end{tikzpicture}
      \hspace*{5em}
      \end{center}
    \item \label{ct-case-transp} If $A$ is transparent, then $\ct{M}$ is a
      tableau node~$N$ with two children $N', N''$ with side~$\aaa$ and
      literals~$\lnot A$ and~$A$, respectively. The children of $N'$ are
      $\ct{M_1}$ and a leaf node, with side~$\bbb$ and literals~$\lnot A$
      and~$A$, respectively. The target of the leaf node is its parent $N'$.
      The children of $N''$ are a leaf node and $\ct{M_2}$, with side~$\bbb$
      and literals~$\lnot A$ and~$A$, respectively.  The target of the leaf
      node is its parent $N''$.  The value of $\ct{M}$ is then the root of a
      tableau that can be depicted as follows:
      
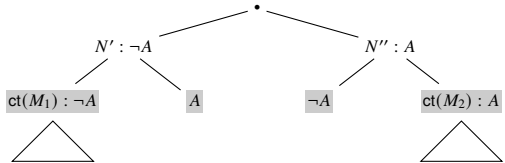
\begin{center}
      \begin{tikzpicture}[scale=\tableauscale,
          sibling distance=15em,level distance=5ex,
          every node/.style = {transform shape,anchor=mid}]]
          \node (a) {\vbar\textbullet}
          child { node {$N' : \lnot A$}
            [sibling distance=8em,level distance=7ex]
            child { node {$\tbbb{\ct{M_1} : \lnot A}$}  \triaw{1.8em}{}}
            child { node {$\tbbb{A}$}}}
          child { node {$N'' : A$}
            [sibling distance=8em,level distance=7ex]
            child { node {$\tbbb{\lnot A}$} }
            child { node {$\tbbb{\ct{M_2} : A}$} \triaw{1.8em}{}}
          };
      \end{tikzpicture}
      \hspace*{5em}
      \end{center}
    \end{enumerate}
  \end{enumerate}

  \noindent
  Observe that the tableaux rooted at $\ct{M_0}$ is in cut normal form, with
  side labelings according to the cases
  \ref{ct-case-aaa}--\ref{ct-case-inner}.  For any leaf with label $L_i$
  introduced according to case~\ref{ct-case-aaa} there exists a target with
  the same side $\aaa$, which can be shown as follows: Since $L_i$ is in a
  clause in $\FL$, its atom must be $\aaa$-colored or transparent.  If it is
  $\aaa$-colored, its target must be one of the nodes with literals $\lnot A$
  or $A$ according to case \ref{ct-case-inner}.\ref{ct-case-ff}, which have
  side~$\aaa$.  If it is transparent, its target must be one of the nodes with
  literals $\lnot A$ or $A$ according to case
  \ref{ct-case-inner}.\ref{ct-case-transp}. In this case, nodes with
  literal~$\lnot A$ or~$A$, respectively, are present as ancestors with each
  side, $\aaa$ and $\bbb$, and the one with side $\aaa$ can be selected as
  target.  This argumentation for leaves introduced according to
  case~\ref{ct-case-aaa} applies analogously to case~\ref{ct-case-bbb}.

  Equation~(\ref{eq-huang-ipol}) can be verified on the basis of the following
  observations about tableaux obtained as values of $\fct$: Case
  \ref{ct-case-inner}.\ref{ct-case-transp} specifies two leaves with side
  $\bbb$ whose targets, which are their parents, have side $\aaa$. By the
  definition of $\f{ipol}$, the literals of these parents enter the
  constructed interpolant.  In fact, these are the \emph{only} leaves of the
  tableau whose target nodes have different side, because leaves according to
  cases~\ref{ct-case-aaa} and~\ref{ct-case-bbb} always have targets with their
  own side.
  \qed
\end{proof}

\noindent
The following example illustrates the simulation.
\begin{examp}[Simulation of Huang's Method by Clausal Tableaux]
  \label{examp-huang-simul}
  Let $F = \fp \land (\lnot \fp \lor \fq)$ and let $\GR = (\lnot \fq \lor \fr)
  \land \lnot \fr$. Figure~\ref{fig-res-huang} shows a deduction tree for $F
  \land \GR$, displayed as a conventional tree, instead of the usual upside
  down presentation of deduction trees \cite{chang:lee}. Nodes are labeled by
  clauses.  \tccctxt{Transparency} is indicated by a frame and \tbbbtxt{color
    $\bbb$} by gray background.  The remaining symbols have color $\aaa$.  The
  values of $\f{pi}$ as determined by Huang's method are annotated in
  brackets, where values are shown as specified in Definition~\ref{def-huang}
  except that truth-value simplification (see footnote~\ref{foot-tv}) is
  applied.  Figure~\ref{fig-sim-huang} shows the corresponding leaf-closed
  \sided clausal tableau that simulates Huang's method, related by the
  $\f{ct}$ mapping. \tbbbtxt{Side $\bbb$} is indicated by gray background.
  The values of $\f{ipol}$ are annotated in brackets, again after truth value
  simplification. They are underlined for nodes that are values of $\f{ct}$
  for some node of the resolution tree.  The target labels of the leaves are
  symbolized by arrows. (Observe in particular that the leaf with literal
  $\fq$ has as target not its closest ancestor with literal $\lnot \fq$
  because the side of that ancestor is not $\aaa$.)  The clausal tableau is
  for $F \land (\lnot \fp \lor \fp) \land (\lnot \fq \lor \fq)$ and $G \land
  (\lnot \fq \lor \fq) \land (\lnot \fr \lor \fr)$.

  \smallskip
  \noindent
  \begin{minipage}[t]{0.33\textwidth}
      \centering
      \hspace{-6pt}%
  \begin{tikzpicture}[scale=0.85,
      baseline=(a.north),
      sibling distance=6.2em,level distance=\extabld,
      every node/.style = {transform shape,anchor=mid}]]
      \node (a) {$\Box$\nannot{$\fq$}}
      child { node {$\tbbb{\fr}$\nannot{$\fq$}} 
        child { node {$\tccc{\fq}$\nannot{$\false$}}
          child { node {$\taaa{\fp}$\nannot{$\false$}}
          }
          child { node {$\taaa{\lnot\fp} \lor \tccc{\fq}$\nannot{$\false$}}}
        }          
        child { node {$\tccc{\lnot \fq} \lor \tbbb{\fr}$%
            \nannot{$\true$}}}
      }
      child { node {$\tbbb{\lnot \fr}$\nannot{$\true$}} };
  \end{tikzpicture}
  
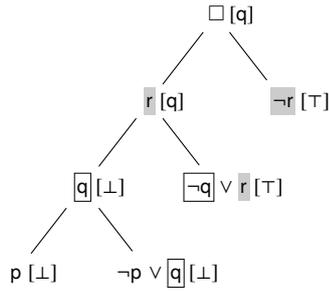
\captionof{figure}{Resolution deduction tree
    and partial interpolant annotation
    according to Huang's method.}
  \label{fig-res-huang}
  \end{minipage}
  \hspace{0.39cm}%
  \begin{minipage}[t]{0.63\textwidth}
      \centering
      \hspace{-11pt}
  \begin{tikzpicture}[scale=0.85,
      baseline={([yshift=-1.5pt]a.north)},
      sibling distance=12em,level distance=\extabld,
      every node/.style = {transform shape,anchor=mid}]]
      \node (a) {\mm{\vbar\textbullet}\nannotwmm{$\fq$}}
      child { node (t7) {\mm{$\tbbb{\lnot \fr}$}\nannotmm{$\fq$}}
        [sibling distance=12em]
        child { node (t3) {\mm{$\taaa{\lnot \fq}$}\nannot{$\false$}}
          [sibling distance=6.2em]
          child { node {\mm{$\tbbb{\lnot \fq}$}\nannotmm{$\false$}}
            child { node (t1) {\mm{$\taaa{\lnot \fp}$}\nannotmm{$\false$}}
              child { node (l1) {\mm{$\taaa{\fp}$}\nannot{$\false$}}}
            }
            child { node (t2) {\mm{$\taaa{\fp}$}\nannotmm{$\false$}}
              [sibling distance=4em]
              child { node (l2) {\mm{$\taaa{\lnot \fp}$}\nannot{$\false$}}}
              child { node (l3) {\mm{$\taaa{\fq}$}\nannot{$\false$}}}
            }
          }
          child { node (l4) {\mm{$\tbbb{\fq}$}\nannot{$\lnot \fq$}}}
        }
        child { node (t5) {\mm{$\taaa{\fq}$}\nannot{$\fq$}}
          [sibling distance=6.2em]
          child { node (l5) {\mm{$\tbbb{\lnot \fq}$}\nannot{$\fq$}}}
          child { node (t6) {\mm{$\tbbb{\fq}$}\nannotmm{$\true$}}
            child { node (l6) {\mm{$\tbbb{\lnot \fq}$}\nannot{$\true$}}}
            child { node (l7) {\mm{$\tbbb{\fr}$}\nannot{$\true$}}
            }
          }
        }
      }
      child { node (t8) {\mm{$\tbbb{\fr}$}\nannotmm{$\true$}}
        child { node (l8) {\mm{$\tbbb{\lnot \fr}$}\nannot{$\true$}}}
      };

      \draw[->,>=stealth] (l1) edge[bend left] node {} (t1);
      \draw[->,>=stealth] (l2) edge[bend left] node {} (t2);
      \draw[->,>=stealth] (l3) edge[bend right=20] node {}
      ([yshift=-8pt,xshift=2pt]t3);
      \draw[->,>=stealth] (l4) edge[bend right] node {}
      ([yshift=-8pt,xshift=10pt]t3);
      \draw[->,>=stealth] (l5) edge[bend left] node {} (t5);
      \draw[->,>=stealth] (l6) edge[bend left] node {} (t6);
      \draw[->,>=stealth] (l7) edge[bend right=40] node {}
      ([yshift=-8pt,xshift=20pt]t7);
      \draw[->,>=stealth] (l8) edge[bend right] node {} (t8);
  \end{tikzpicture}
  
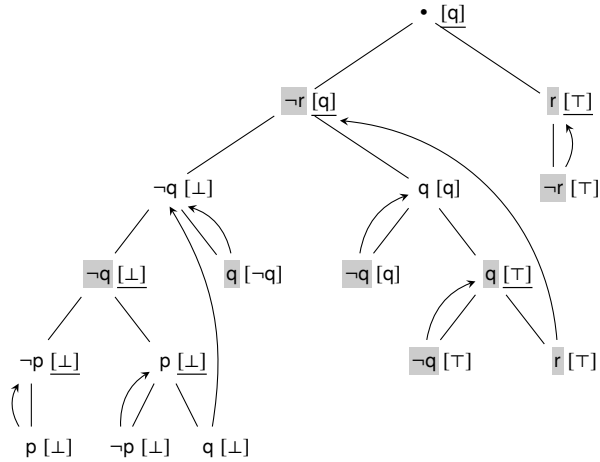
\captionof{figure}{Simulating Huang's method with a clausal
    tableau.}
  \label{fig-sim-huang}
  \end{minipage}
\end{examp}

\bigskip

As already indicated, the mapping $\fct$ is a modification of the known
conversion of deduction trees into clausal tableaux in cut normal form that is
used to show that clausal tableaux with atomic cut and semantic trees can
linearly simulate each other, which, since semantic trees can linearly
simulate tree resolution, implies that tableaux with atomic cut can linearly
simulate tree resolution \cite[Chapter~7]{letz:habil}.  Differences are the
additional labeling of nodes with a side and the introduction of stacked cuts
upon the same atom, but with different side labels for the
case~\ref{ct-case-inner}.\ref{ct-case-transp}.  Through these stacked cuts the
tableau is \name{not regular} \cite{handbook:ar:haehnle}, that is, different
nodes on a branch are labeled with the same literal, however with different
side labeling.

The \name{HKPYM interpolation system} shown in
\cite{bonacina:11,bonacina:15:ground} is similar to Huang's method for the
propositional case, but had also been discovered and analyzed independently by
several other authors (see \cite{bonacina:15:ground} for references, the name
\name{HKPYM} represents initials of these authors, with \name{H} for
\name{Huang}).  It differs from Huang's method in that the case
\ref{huang-case-inner}.\ref{huang-case-transp} of Definition~\ref{def-huang}
is replaced with $\huang{M} \eqdef (A \lor \huang{M_1}) \land (\lnot A \lor
\huang{M_2}).$ If we flip the sides of the first conjunct and take $\huang{M}
\eqdef (\huang{M_1} \lor A) \land (\lnot A \lor \huang{M_2})$ instead, that
method can also be simulated with clausal tableaux such that syntactically
identical interpolants are obtained: We just have to modify
case~\ref{ct-case-inner}.\ref{ct-case-transp} of the definition of $\fct$ such
that the upper cut has side $\bbb$ and the two lower cuts have side $\aaa$.

\name{Optimized Huang} \cite{benedikt:2017} is a variant of Huang's method for
computing Craig-\emph{Lyndon} interpolants which specializes the provisional
interpolation system of \cite[Def.~15]{bonacina:15:on}.  It uses a labeling
with respect to \emph{literal occurrences}, as discussed on
page~\pageref{page-occurrences}.  The case corresponding to
\ref{huang-case-inner}.\ref{huang-case-transp} of Definition~\ref{def-huang}
is split into subcases depending on this labeling, for example (considering
the algorithm just for propositional inputs) $\huang{M} \eqdef \huang{M_1}
\lor (A \land \huang{M_2})$.\footnote{To take account of merging duplicate
  literals, the method as presented in \cite{benedikt:2017} has to be
  supplemented by the explicit consideration of factoring steps with the
  dedicated assignment of partial interpolants from
  \cite[Def.~15]{bonacina:15:on}, or, alternatively, it must be possible to
  label an occurrence with \emph{both} provenance values at the same time, as
  noted in \cite{huang:95}, and retaining Huang's original definition for case
  \ref{huang-case-inner}.\ref{huang-case-transp} for the subcase where the
  resolution step is upon two literals which each have both provenances.  As
  an example where merging is necessary consider $F = (p \lor q) \land (\lnot
  p \lor r)$ and $\GR = (p \lor \lnot q) \land (\lnot p \lor \lnot r)$.}  It
is straightforward to adapt $\fct$ to these cases, as depicted in
Figure~\ref{fig-optimized} for the example case:\footnote{The argument in the
  proof of Theorem~\ref{thm-simul-huang} to show that in
  case~\ref{ct-case-aaa} for any leaf with label $L_i$ there exists a target
  with the same side $\aaa$ can also be applied here, but the analogy for
  case~\ref{ct-case-bbb} has to be shown differently. It follows since in
  case~\ref{ct-case-bbb} a leaf with label $L_i$ must have an ancestor with
  complementary literal that was introduced by the encoding of a resolution
  step \emph{upon} $L_i$ and its complement. If the label of the involved
  occurrence of the complement is~$\aaa$, then the encoding of the resolution
  step would result in two successive nodes with the complement of $L_i$, one
  with side $\aaa$ and one with side $\bbb$.}
\begin{figure}[H]
  \centering
  \vspace{-0.3cm}
  \begin{tikzpicture}[scale=\tableauscale,
      baseline=(a.north),
      sibling distance=15em,level distance=5ex,
      every node/.style = {transform shape,anchor=mid}]]
      \node (a) {\vbar\textbullet}
      child { node {$\ct{M_1} : \lnot A$} \triaw{1.8em}{}} 
      child { node {$A$}
        [sibling distance=8em,level distance=7ex]
        child { node {$\tbbb{\lnot A}$} }
        child { node {$\tbbb{\ct{M_2} : A}$} \triaw{1.8em}{}}
      };
  \end{tikzpicture}
  \caption{The case of \name{optimized Huang} where $\huang{M} \eqdef
    \huang{M_1} \lor (A \land \huang{M_2})$ in tableau simulation.}
  \vspace{-0.2cm}
  \label{fig-optimized}
\end{figure}
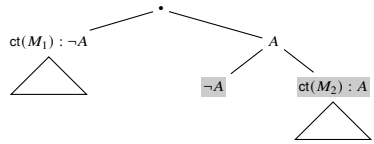  

The encoding of deduction trees into clausal tableaux for interpolant
computation with $\f{ipol}$ suggests further variants of interpolant
extraction from resolution proofs, waiting for exploration or providing a
simulation of further known interpolation methods for resolution proofs.  It
is, for example, possible to translate the case
\ref{ct-case-inner}.\ref{ct-case-transp} not into stacked cuts (i.e., a cut
with side~$\aaa$ where cuts with side~$\bbb$ are attached to its children) as
for the simulation of Huang's method but like cases
\ref{ct-case-inner}.\ref{ct-case-ff} and \ref{ct-case-inner}.\ref{ct-case-gg}
into a single cut, where cuts upon a colored symbol receive the corresponding
side label, and cuts upon a transparent symbol receive an arbitrary side
label.  The resulting tableau is then also a leaf-closed \sided clausal
tableau for $\FL \land \CUTS_{\aaa}\;$ and $\;\GR \land \CUTS_{\bbb}$ that
permits to apply $\f{ipol}$.  It is not difficult to verify that if
side~$\bbb$ is chosen for all cuts upon transparent symbols in this deduction
tree translation, then a linear simulation of McMillan's method
\cite{mcmillan:2003,mcmillan:2005}, presented as interpolation system
\name{MM} in \cite{bonacina:15:ground}, is obtained. Here is an example.
\begin{examp}[Simulation of McMillan's Method by Clausal Tableaux]
  \label{examp-mm-simul}
  This example is analogous to Example~\ref{examp-huang-simul}, but for
  McMillan's instead of Huang's method. The deduction trees in
  Figure~\ref{fig-res-huang} and~\ref{fig-res-mcmillan} differ only in the
  annotated partial interpolants. For McMillan's method they are according to
  \cite[Def.~21]{bonacina:15:ground}. In the tableau translation for Huang's
  method (Figure~\ref{fig-sim-huang}), below the left child of the root
  ``stacked cuts'' upon the transparent atom $\fq$ are inserted, according to
  case \ref{ct-case-inner}.\ref{ct-case-transp} of
  Definition~\ref{def-huang}. For the simulation of McMillan's method shown in
  Figure~\ref{fig-sim-mcmillan}, only a single cut upon~$\fq$ is inserted at
  this position. Its side is~$\bbb$, which is chosen for all cuts upon
  transparent atoms in the simulation of McMillan's method. In the simulation,
  the target of a leaf is always closest ancestor with complementary literal.

  \smallskip
  \noindent
  \begin{minipage}[t]{0.35\textwidth}
      \centering
      \hspace{-14pt}%
    \noindent
    \begin{tikzpicture}[scale=0.85,
        baseline=(a.north),
        sibling distance=6.2em,level distance=\extabld,
        every node/.style = {transform shape,anchor=mid}]]
        \node (a) {$\Box$\nannot{$\fq$}}
        child { node {$\tbbb{\fr}$\nannot{$\fq$}} 
          child { node {$\tccc{\fq}$\nannot{$\fq$}}
            child { node {$\taaa{\fp}$\nannot{$\false$}}
            }
            child { node {$\taaa{\lnot\fp} \lor \tccc{\fq}$\nannot{$\fq$}}}
          }          
          child { node {$\tccc{\lnot \fq} \lor \tbbb{\fr}$%
              \nannot{$\true$}}}
        }
        child { node {$\tbbb{\lnot \fr}$\nannot{$\true$}} };
    \end{tikzpicture}
  
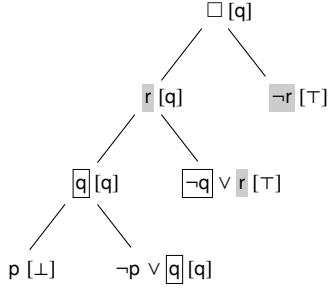
\captionof{figure}{Resolution deduction tree
    and partial interpolant annotation
    according to McMillan's method.}
  \label{fig-res-mcmillan}
  \end{minipage}
  \hspace{0.45cm}%
  \begin{minipage}[t]{0.63\textwidth}
      \centering
      \hspace{-11pt}
    \begin{tikzpicture}[scale=0.85,
      baseline={([yshift=-1.5pt]a.north)},
      sibling distance=12em,level distance=\extabld,
      every node/.style = {transform shape,anchor=mid}]]
      \node (a) {\mm{\vbar\textbullet}\nannotwmm{$\fq$}}
      child { node (t5) {\mm{$\tbbb{\lnot \fr}$}\nannotmm{$\fq$}}
        [sibling distance=12em]
        child { node (t3) {\mm{$\tbbb{\lnot \fq}$}\nannotmm{$\fq$}}
          [sibling distance=6.2em]
          child { node (t1) {\mm{$\taaa{\lnot \fp}$}\nannotmm{$\false$}}
            child { node (l1) {\mm{$\taaa{\fp}$}\nannot{$\false$}}}
            }
          child { node (t2) {\mm{$\taaa{\fp}$}\nannotmm{$\fq$}}
            [sibling distance=4em]
            child { node (l2) {\mm{$\taaa{\lnot \fp}$}\nannot{$\false$}}}
            child { node (l3) {\mm{$\taaa{\fq}$}\nannot{$\fq$}}}
          }
        }
        child { node (t4) {\mm{$\tbbb{\fq}$}\nannotmm{$\true$}}
          [sibling distance=6.2em]
          child { node (l4) {\mm{$\tbbb{\lnot \fq}$}\nannot{$\true$}}}
          child { node (l5) {\mm{$\tbbb{\fr}$}\nannot{$\true$}}
          }
        }
      }
      child { node (t6) {\mm{$\tbbb{\fr}$}\nannotmm{$\true$}}
        child { node (l6) {\mm{$\tbbb{\lnot \fr}$}\nannot{$\true$}}}
      };
      
      \draw[->,>=stealth] (l1) edge[bend left] node {} (t1);
      \draw[->,>=stealth] (l2) edge[bend left] node {} (t2);
      \draw[->,>=stealth] (l3) edge[bend right=40] node {}
      ([yshift=-8pt,xshift=20pt]t3);
      \draw[->,>=stealth] (l4) edge[bend left] node {} (t4);

      \draw[->,>=stealth] (l5) edge[out=75, in=-12] node {}
       ([yshift=-8pt,xshift=20pt]t5);
      \draw[->,>=stealth] (l6) edge[bend right] node {} (t6);
    \end{tikzpicture}
  
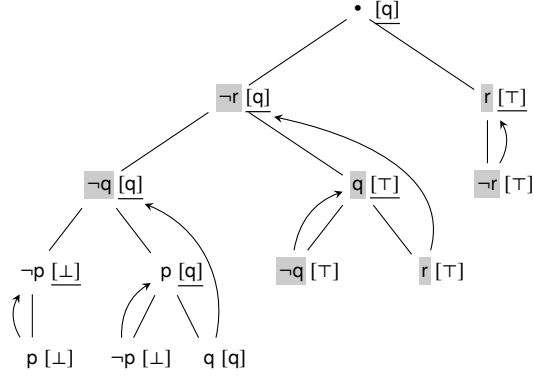
\captionof{figure}{Simulating McMillan's method with a clausal
    tableau.}
  \label{fig-sim-mcmillan}
  \end{minipage}
  \vspace{-0.5cm}    
\end{examp}

\section{Lifting Ground Interpolants -- Additional Notation}
\label{sec-lift-intro}

We now turn to the computation of \emph{first-order} interpolants, beyond the
ground case.  As outlined in Sect.~\ref{sec-intro-lifting}, we follow a
two-stage approach. The second stage, interpolant lifting, is discussed in the
following Sects.~\ref{sec-lift-thm}--\ref{sec-lift-proof}.  Related work on
interpolant lifting will be discussed in Sect.~\ref{sec-lift-related}.  As
interpolant lifting concerns terms, we need additional notation that was not
relevant in the previous sections.

\denseparagraph{The Subterm Relationship.}
We write $\strictsubterm{s}{t}$ to express that $s$ is a strict subterm of
$t$.\footnote{The symbol~$\lhd$ is an adaptation of $t \unrhd s$ for $s$
  \name{is a subterm of} $t$, common in the literature on term rewriting
  \cite{dershowitz:notations:1991}.  We use it in the strict version and in
  flipped direction because it is then in direct correspondence with ordering
  constraints on terms that determine the order of quantifications in our
  Theorem~\ref{thm-lifting}.}

\denseparagraph{Terms with Outermost Function Symbol in a Given Set.}
If $\symset{S}$ is a set of function symbols, then a term whose outermost
function symbol is in $\symset{S}$ is called an \defname{$\symset{S}\sterm$}.
The plural form \defname{$\symset{S}\sterms$} is also used to denote the set
of all $\symset{S}\sterms$.

\denseparagraph{Substitutions.}
A \defname{substitution} is a mapping from variables to terms which is almost
everywhere equal to identity.  If $\sigma$ is a substitution, then the
\defname{domain} of $\sigma$ is the set of variables $\dom{\sigma} \eqdef \{x
\mid x\sigma \neq x\}$, the \defname{range} of $\sigma$ is $\rng{\sigma}
\eqdef \bigcup_{x \in \dom{\sigma}} \{x\sigma\}$, and the
\defname{restriction} of $\sigma$ to a set~$\XS$ of variables, denoted by
$\sigma|_{\XS}$, is the substitution which is equal to the identity everywhere
except over $\XS \cap \dom{\sigma}$, where it is equal to~$\sigma$.  The
identity substitution is denoted by $\emptysubst$.  We write the
set~$\var{\rng{\sigma}}$ of variables in the range of substitution~$\sigma$
also as $\varrng{\sigma}$.  If all members of the range of a substitution are
ground, then the substitution is called a \defname{ground substitution}.

A substitution can be represented as a function by a set of assignments to the
variables in its domain, e.g., $\{x_1 \mapsto t_1, \ldots, x_n \mapsto t_n\}$.
The application of a substitution~$\sigma$ to a term or a formula~$E$ is
written as $E\sigma$, $E\sigma$ is called an \defname{instance} of $E$ and $E$
is said to \defname{subsume} $E\sigma$.  Composition of substitutions is
written as juxtaposition. Hence, if $\sigma$ and $\gamma$ are both
substitutions, then $E\sigma\gamma$ stands for $(E\sigma)\gamma$.

\denseparagraph{Maximal Occurrences and Inverse Application of an Injective
  Substitution.} An occurrence of a member~$t$ of a set~$\mathcal{T}$ of terms
in a term or a formula~$E$ is called \defname{$\mathcal{T}\!$-maximal} if it is
not within an occurrence of another member of $\mathcal{T}$.  If $\sigma$ is
an injective substitution and $E$ is a term or a formula, then
$\invsubst{E}{\sigma}$ denotes $E$ with all $\rng{\sigma}$-maximal occurrences
of terms~$t \in \rng{\sigma}$ replaced by the variable that is mapped by
$\sigma$ to $t$.
As an example let $\sigma = \{x \mapsto \ff(\fa),\, y \mapsto
\fg(\ff(\fa))\}$. Then
\[\invsubst{\fp(\fh(\ff(\fa),\fg(\ff(\fa))))}{\sigma} = \fp(\fh(x,y)).\]
The inverse application of an injective substitution~$\sigma$ to a term or
formula~$E$ can be described operationally as follows: (1)~construct the
inverse of~$\sigma$, which is $\{\ff(\fa) \mapsto x,\, \fg(\ff(\fa)) \mapsto
y\}$ in our example.  As can be seen with the example, the inverse is in
general not a substitution because arbitrary terms and not just variables can
be in its domain.  (2)~Order the members of the inverse into a sequence of
pairs such that whenever for two pairs~$t \mapsto v$ and~$t' \mapsto v'$ it
holds that $\strictsubterm{t'\!}{t}$, then $t \mapsto v$ appears to the
left of~$t' \mapsto v'$.  For the example we obtain $\la \fg(\ff(\fa)) \mapsto
y,\, \ff(\fa) \mapsto x\ra$.  (3)~For each member~$t \mapsto v$ of the
sequence, proceeding from left to right, replace all occurrences of $t$ in~$E$
with $v$.  In our example we first obtain $\fp(\fh(\ff(\fa),y))$ as
intermediate result and then the final result $\fp(\fh(x,y))$.
We note the following properties of inversely applied substitutions.
 \begin{prop}[Properties of Inverse Substitutions]
   Let $\sigma$ be an injective substitution. Then:
      
   \slab{prop-invsubst-subst} If $E$ is a term or a quantifier-free formula
   such that $\dom{\sigma} \cap \var{E} = \emptyset$ then
   $\invsubst{E}{\sigma}\sigma = E$.

   \slab{prop-invsubst-entails} If $F, G$ are quantifier-free formulas such
   that $F \entails G$ then $\invsubst{F}{\sigma} \entails
   \invsubst{G}{\sigma}$.
\end{prop}

\begin{proof}[Sketch]
  (\ref{prop-invsubst-subst}) Easy to see.  (\ref{prop-invsubst-entails})
  Recall that we only consider formulas without equality.  A quantifier-free
  formula is valid if and only if it is \name{propositionally valid}, that is,
  the propositional formula obtained by taking each atom~$p(t_1,\ldots,t_n)$
  as name of a propositional variable is valid.  If $A$ is an atom occurring
  in a quantifier-free formula, then the inverse application of an injective
  substitution~$\sigma$ effects that either no occurrence of $A$ is replaced
  or all occurrences of~$A$ are replaced by the \emph{same} atom
  $A\invsubstpost{\sigma}$.  This implies that the inverse application of an
  injective substitution to a valid quantifier-free formula yields again a
  valid quantifier-free formula: Any proof that shows the propositional
  validity of the original formula can be modified to a proof of the
  propositional validity of the formula after the inverse application by just
  replacing the propositional variables that name replaced atoms~$A$ with the
  propositional variables that name the respective replacement atoms
  $A\invsubstpost{\sigma}$.  (This holds independently from whether names of
  some of the replacing atoms~$A\invsubstpost{\sigma}$ were already used in
  the original proof.)  Assuming that $F$ and $G$ are quantifier-free formulas
  such that $F \entails G$ we can conclude that $F \imp G$ is a valid
  quantifier-formula, hence $(F \imp G)\invsubstpost{\sigma}$, which identical
  to $F\invsubstpost{\sigma} \imp G\invsubstpost{\sigma}$, is valid, hence
  $F\invsubstpost{\sigma} \entails G\invsubstpost{\sigma}$.  \qed
\end{proof}

\section{Interpolant \LBases and Statement of the Interpolant Lifting Theorem}
\label{sec-lift-thm}

As already indicated informally, the lifting to first-order interpolants can
be specified on the basis of Craig-Lyndon interpolants of ground formulas that
are in a certain relationship to the original first-order interpolation
inputs, where the arguments for the correctness of the lifting transformation
are independent of a particular calculus.
Instead, the correctness proof is based on second-order Skolemization and
Herbrand's theorem.  The concept of \name{interpolant lifting base}, defined
below, provides an interface between formula constructions by second-order
Skolemization and quantifier expansion according to Herbrand's theorem on the
one hand and those abstract properties of the constructed formulas that are
needed to justify interpolant lifting on the other hand.  It should be noted,
that these formula constructions are applied only to prove the correctness of
interpolant lifting, without implying the necessity to actually perform them
at interpolant computation. Thus, the correctness of interpolant lifting is
shown via the existence of certain constructible formulas.

\renewcommand{\theprlabcounter}{\text{\roman{prlabcounter}}}

\newcommand{\tightop}[1]{\hspace{1pt}{#1}\hspace{1pt}}

An interpolant lifting base gathers certain components and relates them by
constraints. These components are the original interpolation inputs~$F, G$,
intermediate formulas~$\FE, \GE$ (constructible from $F$ and $G$ by
Skolemization and \emph{expansion} according to Herbrand's theorem), sets of
function symbols~$\ffs, \ggs$ (symbols that are not permitted to occur in the
lifted interpolant and include the Skolem functions in $\FE$ and $\GE$), a
ground substitution~$\sth$ whose domain is the set of the variables occurring
in $\FE$ and $\GE$, and a \emph{ground} interpolant~$\HG$ of the ground
formulas $\FE\sth$ and $\GE\sth$.  In the lifting base these components are
arranged as a tuple such that the components $F, G, \ffs, \ggs, \HG$, which
are those that need to be materialized at the interpolant lifting formula
transformation, precede the other ones $\FE, \GE, \sth$.  The following
definition makes \name{interpolant lifting base} precise.
\begin{defn}[Interpolant Lifting Base]
  \label{def-ib}
  \prlReset{def-ib} An \defname{interpolant lifting base} (briefly
  \defname{\lbase}) is a tuple
  \[\la F, G, \ffs, \ggs, \HG, \FE, \GE, \sth\ra,\]
  where $F, G$ are first-order sentences, $\ffs, \ggs$ are disjoint sets of
  function symbols, $\HG$ (the subscript \textsc{grd} suggesting
  \name{ground}) is a ground formula, $\FE, \GE$ (the subscript \textsc{exp}
  suggesting \name{expansion}) are quantifier-free formulas and $\sth$ is a
  ground substitution such that:
  \[
  \hspace{-3pt}
  \begin{array}{r@{\hspace{0.4em}}l@{\hspace{0.5em}}r@{\hspace{0.4em}}l}
    \ibn{ib:sem} &
    F \entails \exists \ffs \forall \us\, \FE, &
    \ibm{ib:sem} &
    \forall \ggs \exists \vs\, \GE \entails G,\\
    & \text{where } \us = \var{\FE}. &
    & \text{where } \vs = \var{\GE}.\\
    \ibn{ib:pred} &
    \pred{\FE} \subseteq \pred{F}.
    & \ibm{ib:pred} &
    \pred{\GE} \subseteq \pred{G}.\\
    \ibn{ib:fe} &
    \fun{\FE} \subseteq (\fun{F} \tightop{\cap} \fun{G}) \tightop{\cup} \ffs.
    & \ibm{ib:fe} &
    \fun{\GE} \subseteq (\fun{F} \tightop{\cap} \fun{G}) \tightop{\cup} \ggs.\\
    \ibn{ib:fcap} &
    \fun{F} \cap \ggs = \emptyset.
    & \ibm{ib:fcap} &
    \fun{G} \cap \ffs = \emptyset.\\
  \end{array}
  \]
  \[
  \hspace{-3pt}
  \begin{array}{r@{\hspace{0.4em}}l}
    \ibn{ib:domh} & \dom{\sth} = \var{\FE} \cup \var{\GE}.\\
    \ibn{ib:rngh} &
    \fun{\rng{\sth}} \subseteq \fun{\FE} \cup \fun{\GE} \cup \{\fc_0\},
    \text{where } \fc_0 \text{ is a constant in } \ffs \cup \ggs.\\
    \ibn{ib:ipol} & \HG \text{ is a Craig-Lyndon interpolant of } \FE\sth
    \text{ and } \GE\sth.\\
  \end{array}
  \]
\end{defn}

\renewcommand{\theprlabcounter}{\arabic{prlabcounter}}

We now consider the constructive aspect of interpolant lifting bases, showing
with Proposition~\ref{prop-lbase-exists} that for given first-order sentences
$F, G$ such that $F \entails G$ a lifting base exists.  The proof shows this
constructively, providing intuitions that flesh out the items abstractly
related by constraints in Definition~\ref{def-ib}.  Further intuition might be
provided with Procedure~\ref{proc-ctif} and Example~\ref{examp-ibase} below.
Proposition~\ref{prop-lbase-exists} is preceded by two auxiliary propositions
about Skolemization and Herbrand's theorem.
\begin{prop}[Second-Order Skolemization]
\label{prop-second-order-skolemization}
Let $F$ be a formula. Assume that $x_1, \ldots, x_n,$ $y$ are variables that do
not occur bound in $F$ and that $f$ is an $n$-ary function symbol that does
not occur in $F$. Then $\forall x_1 \ldots \forall x_n \exists y\, F\;
\equiv\; \exists f \forall x_1 \ldots \forall x_n\, F\{y \mapsto
f(x_1,\ldots,x_n)\}.$
\end{prop}
If $f$ in Proposition~\ref{prop-second-order-skolemization} is nullary, that
is, if $f$ is a Skolem \name{constant}, then the quantification upon~$f$ is
\emph{first-order} quantification upon~$f$ in the role of a
variable. Otherwise, the quantification is second-order quantification upon a
function symbol.  We need second-order quantification to capture the semantics
of Skolemization, not just the commonly used preservation of satisfiability
and unsatisfiability. However, we use second-order quantification only in
contexts where actual reasoning in second-order logic is not required.

The following proposition expresses Herbrand's theorem as used in automated
theorem proving in a form that applies to conjunctions of universal
first-order sentences.
\begin{prop}[A Form of Herbrand's Theorem]
  \label{prop-h}
  A first-order sentence $S$ of the form $\forall \XS_1\, C_1 \land \ldots
  \land \forall \XS_n C_n$, where $\XS_1,\ldots,\XS_n$ are finite sets of
  variables and $C_1, \ldots, C_n$ are quantifier-free formulas is
  unsatisfiable if and only if there exists a finite unsatisfiable conjunction
  of ground instances of formulas $C_i$ with $i \in \{1,\ldots,n\}$, obtained
  by instantiating with terms formed from functions occurring in $S$ (and, if
  no constant occurs in $S$, a fresh constant).\footnote{$C_1, \ldots, C_n$
    can be arbitrary quantifier-free formulas, with clauses as a special
    case.}
\end{prop}
We can now claim the existence of an interpolant base for sentences $F, G$
such that $F \entails G$ and prove it with a construction.
\begin{prop}[Existence of an Interpolant \LBase]
  \label{prop-lbase-exists}
  If $F,G$ are first-order sentences such that $F \entails G$, then there
  exists an interpolant lifting base \[\la F, G, \ffs, \ggs, \HG, \FE, \GE,
  \sth\ra.\]
\end{prop}

\begin{proof}
  \prlReset{prop-ib-existence} Given $F \entails G$, the remaining components
  of an interpolant lifting base can be constructed as follows: Apply
  conversion to prenex form and Skolemization (according to
  Proposition~\ref{prop-second-order-skolemization}) independently to $F$ and
  to $\lnot G$, to obtain disjoint sets of fresh Skolem functions $\ffs',
  \ggs'$, quantifier-free formulas $F',G'$, and sets $\us' = \var{F'}, \vs' =
  \var{G'}$ of variables such that:
  \[
  \begin{arrayprf}
    \prl{prime}         & F \equiv \exists \ffs' \forall \us' F'.\\
    (\prefNumber{prime}') & \lnot G \equiv \exists \ggs' \forall \vs' G'.\\
    \prl{prime-voc} & \voc{F'} \subseteq \voc{F} \cup \ffs'.\\
    (\prefNumber{prime-voc}') &
    \voc{\lnot G'} \subseteq \voc{G} \cup \ggs'.\\
    \prl{prime-fo} & \forall \us' \forall \vs' (F' \land G')
    \entails \false.\\
  \end{arrayprf}
  \]
  If there is a constant occurring in $F' \land G'$, then let $\fc_0$ be such
  a constant and let $\symset{C} \eqdef \emptyset$, else let~$\fc_0$ be a
  fresh constant and let $\symset{C} \eqdef \{\fc_0\}$.  From Herbrand's
  theorem (Proposition~\ref{prop-h}) and \pref{prime-fo} it follows that there
  exist quantifier-free formulas $\FE,\GE$ and a ground substitution~$\sth$
  such that:
  \[
  \begin{arrayprf}
    \prl{p-dom-sth} & \dom{\sth} = \var{\FE} \cup \var{\GE}.\\
    \prl{p-rng-sth} & \fun{\rng{\sth}} \subseteq \fun{\FE} \cup \fun{\GE} \cup
    \{\fc_0\}.\\
    \prl{sth-false} & \FE\sth \land \lnot \GE\sth \entails \false.\\
    \prl{sth-sem} & \forall \us'\ F' \entails \forall \us\, \FE, \text{ where }
    \us = \var{\FE}.\\
   (\prefNumber{sth-sem}') &
    \forall \vs'\ G' \entails \forall \vs\, \lnot \GE, \text{ where }
    \vs = \var{\GE}.\\
    \prl{voc-e} & \voc{\FE} \subseteq \voc{F'}.\\
    (\prefNumber{voc-e}') & \voc{\GE} \subseteq \voc{\lnot G'}.\\
  \end{arrayprf}
  \]
  Steps~\pref{sth-sem}, ($\prefNumber{sth-sem}'$), \pref{voc-e}
  and~($\prefNumber{voc-e}'$) hold since Proposition~\ref{prop-h} justifies
  the construction of $\FE$ as a conjunction of variants of $F'$ (that is,
  formulas obtained from $F'$ by systematically renaming variables) and the
  construction of $\GE$ as the negation of a conjunction of variants of
  $G'$.\footnote{In the case where $F'$ and $G'$ are clausal formulas, there
    also may exist smaller formulas $\FE$ and $\GE$ obtained as conjunction of
    variants of clauses of $F'$ and by negating a conjunction of variants of
    clauses of~$G'$, respectively.}  Construct the sets~$\ffs$ and~$\ggs$ of
  function symbols according to the following definitions:
 \[
  \begin{arrayprf}  
    \prl{sth-funf} & 
    \ffs\; \eqdef\; \ffs' \cup (\fun{F} \setminus \fun{G}) \cup \symset{C}.\\
    (\prefNumber{sth-funf}') &
    \ggs\; \eqdef\; \ggs' \cup (\fun{G} \setminus \fun{F}).
  \end{arrayprf}
  \]
  That is, $\ffs$ and~$\ggs$ extend the sets of Skolem functions for $F$ and
  $\lnot G$ by the functions occurring in $F$ but not in~$G$ and vice versa,
  respectively. In addition, if a fresh symbols is needed as constant $\fc_0$,
  it is added to $\ffs$.

 That the constructed elements match conditions \ibref{ib:sem},
 \ibref{ib:pred}, \ibref{ib:fe} and~\ibref{ib:fcap} of Definition~\ref{def-ib}
 can be shown in the following steps, explained below, where the respective
 conditions are annotated on the right.  Conditions \ibrefm{ib:sem},
 \ibrefm{ib:pred}, \ibrefm{ib:fe} and~\ibrefm{ib:fcap} of
 Definition~\ref{def-ib} can be derived in straightforward analogy.
  \[\begin{array}{Z{2.5em}@{\hspace{1em}}X{25em}@{\hspace{1em}}X{6.7em}}
   \prl{p-ib-sem} & F \entails \exists \ffs \forall \us\, \FE,
  \text{ where } \us = \var{\FE}. & \text{Def.}~\ref{def-ib}.\ibref{ib:sem}\\
  \prl{p-ib-pred} &
  \pred{\FE} \subseteq \pred{F}. & \text{Def.}~\ref{def-ib}.\ibref{ib:pred}\\
  \prl{f00} & \fun{\FE} \subseteq \fun{F} \cup \ffs'.\\
  \prl{f01} & \fun{\FE} \subseteq \fun{F} \cup \ffs' \cup \symset{C}.\\
  \prl{f02} & \hphantom{ =\; } \fun{F} \cup \ffs' \cup \symset{C}\\
      & =\; (\fun{F} \cap \fun{G}) \cup (\fun{F} \setminus \fun{G}) \cup \ffs' \cup
  \symset{C}\\
      &   =\; (\fun{F} \cap \fun{G}) \cup \ffs.\\
  \prl{p-ib-fe} & \fun{\FE} \subseteq (\fun{F} \tightop{\cap} \fun{G})
  \tightop{\cup} \ffs. &
  \text{Def.}~\ref{def-ib}.\ibref{ib:fe}\\
  \prl{p-ib-fcap} & \fun{F} \cap \ggs = \emptyset. & \text{Def.}~\ref{def-ib}.\ibref{ib:fcap}\\
  \end{array}
  \]
  Step~\pref{p-ib-sem} follow from~\pref{prime} and \pref{sth-sem}.
  Steps~\pref{p-ib-pred} and~\pref{f00} both follow from~\pref{voc-e}
  and~\pref{prime-voc}. Step~\pref{f01} follows from~\pref{f00}.
  Step~\pref{f02} follows from general properties of sets and~\pref{sth-funf}.
  Step~\pref{p-ib-fe} follows from~\pref{f01} and~\pref{f02}.
  Step~\pref{p-ib-fcap} follow from step~($\prefNumber{sth-funf}'$), since
  $\ggs'$ contains only fresh symbols.

  Conditions~\ibref{ib:domh} and~\ibref{ib:rngh} of Definition~\ref{def-ib}
  have already been stated as steps~\pref{p-dom-sth} and~\pref{p-rng-sth}.  It
  remains to consider condition~\ibref{ib:ipol}.  Step~\pref{sth-false} can be
  expressed as $\FE\sth \entails \GE\sth$.  By our ground interpolation method
  from Sect.~\ref{sec-ipol-basic} (or any other constructive statement of the
  Craig-Lyndon interpolation theorem for ground formulas) it follows that from
  $\FE\sth$ and~$\GE\sth$ a Craig-Lyndon interpolant $\HG$ of~$\FE\sth$
  and~$\GE\sth$ can be constructed.  Condition~\ibref{ib:ipol} of
  Definition~\ref{def-ib}, which states that $\HG$ is a Craig-Lyndon
  interpolant of~$\FE\sth$ and~$\GE\sth$, is then evidently satisfied.
  \qed
\end{proof}

Here are several examples of interpolant lifting bases that illustrate
different aspects.
\begin{examp}[Interpolant Lifting Base]
  \label{examp-ibase}
  The following are examples of interpolant lifting bases with specific
  properties described as introductory comments:

  \smallskip
  \selab{ex-ib-nonlocal} $\ffs$ contains a non-constant;
  members of $\ffs$ as well as a members of $\ggs$
  occur in $\HG$:
  \[
  \begin{array}{l}
    \la F = \forall x\, \fp(x, \ff(x)),\; G = \exists x\, \fp(\fg, x),\
    \ffs = \{\ff\},\; \ggs = \{\fg\},\; \HG = \fp(\fg, \ff(\fg)),\\
    \hphantom{\la} \FE = \fp(u_1, \ff(u_1)),\; \GE = \fp(\fg, u_2),\;
    \sth = \{u_1 \mapsto \fg, u_2 \mapsto \ff(\fg)\}\ra.
  \end{array}
  \]
  In this example it holds that $\FE\sth = \GE\sth = \HG = \fp(\fg,
  \ff(\fg))$.

  \medskip
  
  \selab{ex-ib-skolem} A member of $\ffs$ (i.e., $\ff_2$) is a Skolem function:
  \[
  \begin{array}{l}
    \la F = \forall x \exists y\, \fp(x, y, \ff_1),\;
    G = \exists x \exists y \, \fp(\fg, x, y),\;
    \ffs = \{\ff_1,\ff_2\},\; \ggs = \{\fg\},\\
    \hphantom{\la} \HG = \fp(\fg,\ff_2(\fg),\ff_1),\;
     \FE = \fp(u_1,\ff_2(u_1),\ff_1),\;
     \GE = \fp(\fg, u_2, u_3),\\
     \hphantom{\la} \sth = \{u_1 \mapsto \fg, u_2 \mapsto \ff_2(\fg),
     u_3 \mapsto \ff_1\}\ra.
  \end{array}
  \]
  In this example it holds that $\FE\sth = \GE\sth = \HG =
  \fp(\fg,\ff_2(\fg),\ff_1)$.

  \medskip
  
  \selab{ex-ib-expansion} $\FE$ is a conjunction of different variants of
  the matrix $\fp(x,\ff)$ of $F$:
  \[
  \begin{array}{l}
    \la F = \forall x\, \fp(x, \ff),\;
    G = \exists x\, \fp(\fg_1, x) \land \exists x\, \fp(\fg_2, x),\;
    \ffs = \{\ff\},\; \ggs = \{\fg_1,\fg_2\},\\
    \hphantom{\la} \HG = \fp(\fg_1, \ff) \land \fp(\fg_2, \ff),\;
     \FE = \fp(u_1, \ff) \land \fp(u_2, \ff),\;
     \GE = \fp(\fg_1, u_3) \land \fp(\fg_2, u_3),\\
     \hphantom{\la} \sth = \{u_1 \mapsto \fg_1, u_2 \mapsto \fg_2,
     u_3 \mapsto \ff\}\ra
  \end{array}
  \]
  In this example it holds that $\FE\sth = \GE\sth = \HG = \fp(\fg_1, \ff)
  \land \fp(\fg_2, \ff)$.
  
  \medskip
  
  \selab{ex-ib-qr} The input formulas~$F,G$ in this example extend those
  of~(i) by additional literals with predicates~$\fq, \fr$ that occur in only
  one of them and a second function symbol in $G$.  Differently from the
  previous three cases, in this case the ground interpolant is not an instance
  of the intermediate formula $\FE$. Values of~$\FE$, $\GE$ and~$\sth$ with
  other features than shown here are also possible. For example, the
  variables~$u_3, u_4$ could be merged with~$u_1$, or $\sth$ could assign
  $u_3, u_4$ to some other ground term. The values here have been chosen, as
  they are suited to illustrate further aspects in Example~\ref{examp-items}
  and the discussion following Definition~\ref{def-i}.
  \[
  \begin{array}{l}
    \la F = \forall x\, \fp(x, \ff(x)) \land
    \forall x \forall y\, \fq(\ff(x),y),\;
    G = \exists x\, (\fp(\fg_1, x) \lor \fr(\fg_2(x))),\\\hphantom{\la} 
    \ffs = \{\ff\},\; \ggs = \{\fg_1, \fg_2\},\;
    \HG = \fp(\fg_1, \ff(\fg_1)),\\\hphantom{\la}
    \FE = \fp(u_1, \ff(u_1)) \land \fq(\ff(u_3), u_4),\;
    \GE = \fp(\fg_1, u_2) \lor \fr(\fg_2(u_2)),\\
    \hphantom{\la} 
    \sth = \{u_1 \mapsto \fg_1, u_2 \mapsto \ff(\fg_1),
    u_3 \mapsto \fg_2(\ff(\fg_1)), u_4 \mapsto \fg_1\}\ra
  \end{array}
  \]
  In this example it holds that $\FE\sth = \fp(\fg_1, \ff(\fg_1)) \land
  \fq(\ff(\fg_2(\ff(\fg_1))), \fg_1)$ and $\GE\sth = \fp(\fg_1, \ff(\fg_1))
  \lor \fr(\fg_2(\ff(\fg_1)))$.
\end{examp}

We now leave the consideration of how an interpolant lifting base can be
constructed and view the interface provided by that concept from its other
side, as a representation of the preconditions of interpolant lifting, that
is, of the construction of a first-order Craig-Lyndon interpolant from the
components of a given lifting base. This is made precise with the following
theorem, which is proven later in the dedicated Section~\ref{sec-lift-proof}.
\begin{thm}[Interpolant Lifting]
  \label{thm-lifting}
  Let $F, G, \ffs, \ggs, \HG$ be the first components of a \lbase.  Let $\fgs$
  stand for $\ffs \cup \ggs$.  Let $\{t_1, \ldots, t_n\}$ be the set of the
  $\fgs\sterms$ with an $\fgs\sterms$-maximal occurrence in $\HG$, ordered
  such that if $\strictsubterm{t_i}{t_j}$, then $i < j$.  Let
  $\{v_1,\ldots,v_n\}$ be a set of fresh variables and let $\stt$ be the
  injective substitution
  \[\stt \eqdef \{v_i \mapsto t_i \mid i \in \{1,\ldots,n\}\}.\]
  For $i \in \{1,\ldots,n\}$ let $Q_i \eqdef \exists$ if $v_i\stt \in
  \ffs\sterms$ and $Q_i \eqdef \forall$ otherwise, that is, if $v_i\stt \in
  \ggs\sterms$. Then
  \[Q_1 v_1 \ldots Q_n v_n\, \HG\invsubstpost{\stt}\] is a Craig-Lyndon
  interpolant of $F$ and $G$.
\end{thm}
Theorem~\ref{thm-lifting} shows a construction of first-order formula from a
given ground formula $\HG$ and two given sets, $\ffs$ and $\ggs$, of function
symbols. The constructed formula is obtained from $\HG$ by replacing
occurrences of $\ffs\sterms$ and $\ggs\sterms$ that are maximal with respect
to $\ffs \cup \ggs$ with variables, and prepending a quantifier prefix upon
these variables that respects certain constraints with respect to the replaced
terms: Variables replacing an $\ffs\sterm$ are existentially quantified,
variables replacing a $\ggs\sterm$ are universally quantified, and whenever
variables $x,y$ replace terms~$s,t$, respectively, such that
$\strictsubterm{s}{t}$, then the quantification upon $x$ precedes that upon
$y$.  The theorem then claims that the constructed first-order formula is a
Craig-Lyndon interpolant of first-order formulas~$F$ and~$G$, provided $F, G,
\ffs, \ggs, \HG$ are the first components of a lifting base, that is, satisfy
certain constraints that are expressed with reference to further related
formulas $\FG, \FE$ and a related substitution~$\sth$.  Interpolants obtained
by applying the theorem are shown in the following example.
\begin{examp}[Interpolant Lifting]
  Consider the interpolant lifting bases from Example~\ref{examp-ibase}.
  Respective Craig-Lyndon interpolants obtained according
  to Theorem~\ref{thm-lifting} are:

  For (\ref{ex-ib-nonlocal}) and (\ref{ex-ib-qr}):
  $\forall v_1 \exists v_2\, \fp(v_1, v_2)$.

  For (\ref{ex-ib-skolem}): $\ex v_1 \all v_2 \ex v_3\, \fp(v_2, v_3, v_1)$.
  Also other orderings of the quantifier prefix are possible according to
  Theorem~\ref{thm-lifting}. The only required condition is (expressed with
  the variable names of the shown value) that $\forall v_2$ must precede
  $\exists v_3$.
  
  For (\ref{ex-ib-expansion}): $\ex v_1 \all v_2 \all v_3\, (\fp(v_2, v_1)
  \land \fp(v_3, v_1))$, which is equivalent to $\ex v_1 \all v_2
  \fp(v_2,v_1)$.  Also arbitrary other orderings of the quantifier prefix are
  possible.
  
\end{examp}

Note that for applying the lifting theorem only some of the components of a
lifting base need actually to be constructed: $F$, $G$, $\ffs$, $\ggs$ and
$\HG$.  For the remaining components, $\FE$, $\GE$ and~$\sth$, it is
sufficient to ensure that they exist. They need not to be materialized.  In
other words, interpolant lifting according to the lifting theorem is
applicable on top of all ground interpolation methods that produce from given
$F$, $G$, $\ffs$ and $\ggs$ a ground formula $\HG$ such that formulas $\FE$,
$\GE$ and a substitutions $\sth$ which satisfy the constraints for an
interpolant lifting base \emph{exist}, without need to actually compute them.
Procedure~\ref{proc-ctif} below gives an example of applying interpolant
lifting on top of a ground interpolation method with clausal tableaux.

We conclude this section with a discussion on how the components of an
interpolant lifting base $\la F, G, \ffs, \ggs, \HG, \FE, \GE, \sth\ra$ can be
matched with the color-based terminology used in some of the literature on
interpolation of refutations by resolution and superposition.  We already
noted aspects of that terminology in footnote~\ref{foot-color}
(page~\pageref{foot-color}), and applied it in Definition~\ref{def-huang}.  It
has been developed mainly for clausal inputs, whereas our interpolation input
formulas~$F$ and~$G$ are subject to Skolemization, which introduces fresh
function symbols. As sketched in the proof of
Proposition~\ref{prop-lbase-exists}, Skolemization for interpolation can be
performed separately on $F$ and $\lnot G$.  Interpolation for clausal inputs
is considered as applied after such a separate Skolemization, implying that
Skolem functions are \emph{colored}.

Thus $\ffs$ may be taken as the set of $\aaa$-colored symbols
and $\ggs$ as the set of $\bbb$-colored symbols, although both sets can
contain Skolem functions that are present neither in $F$ nor in $G$.  However,
with respect to the quantifier-free formulas $\FE$ and $\GE$, the sets $\ffs$
and $\ggs$ are the sets of the $\aaa$- and $\bbb$-colored, respectively,
symbols in the strict original sense.

A compound structure, that is, a term or a formula, is called
$\aaa$-\name{colored} if all symbols occurring in it are either $\aaa$-colored
or transparent and there is at least one occurrence of an $\aaa$-colored
symbol. The definition of $\bbb$-colored for compound structures is analogous.
Hence the formulas~$F$ and~$\FE$ are either $\aaa$-colored or transparent, and
the formulas~$G$ and~$\GE$ are either $\bbb$-colored or transparent.
Transparency of these formulas might be forbidden, as it indicates that inputs
themselves already provide trivial interpolants. In our context this seems,
however, an artificial restriction that is undesired because colored
inputs~$F$ and~$G$ may lead to formulas~$\FE$ and~$\GE$ that are transparent,
and, moreover, structural properties of interpolants for transparent inputs
might be of interest to get insights on the interpolation algorithm and in the
perspective of the computation of interpolants with further properties than
those required of Craig-Lyndon interpolants.

Arbitrary colored and transparent function symbols from~$\FE$ and~$\GE$ are
allowed in the formulas $\HG$, $\FE\sth$ and $\GE\sth$, such that these
formulas may have any of the four possible color status values (transparent,
$\aaa$-colored, $\bbb$-colored, or $\aaa\bbb$-mixed).  All predicates
occurring in~$\HG$ are transparent.  The key to interpolant lifting according
to Theorem~\ref{thm-lifting} is to consider occurrences of $\ffs\sterms$ and
$\ggs\sterms$ that are maximal with respect to $\ffs \cup \ggs$.  An
$\ffs\sterm$ is either $\aaa$-colored or $\aaa\bbb$-mixed, a $\ggs\sterm$ is
either $\bbb$-colored or $\aaa\bbb$-mixed.  However, an exact characterization
of the $\ffs\sterms$ and $\ggs\sterms$ would require an extension of the
color-based terminology that takes the color of the \emph{outermost} symbol of
a term especially into account.

\section{First-Order Interpolation with Clausal Tableaux}
\label{sec-lift-procedure}

Based on the correctness of ground interpolant extraction with the $\f{ipol}$
function (Lemma~\ref{lem-ground-ipol-correct}) and the interpolant lifting
theorem (Theorem~\ref{thm-lifting}) we can now formulate a generalization of
the \CTIG procedure (Procedure~\ref{proc-ctig}) that computes Craig-Lyndon
interpolants from first-order sentences (without equality, except if
represented as predicate, see Sect.~\ref{sec-cli-related-equality}).
We call the procedure \name{\CTIF}, suggesting \name{Clausal Tableau
  Interpolation for First-Order Formulas}.
\begin{proc}[The \CTIF Method for Craig-Lyndon Interpolation]
\label{proc-ctif}

 \algoinput First-order sentences $F$ and $G$ such that $F \entails G$.

 \algomethod
 \begin{enumerate}
 \item \label{step-pre} \name{Skolemization and clausification:} Apply
   conversion to prenex form and Skolemization to $F$ and to $\lnot G$,
   independently to each formula, to obtain disjoint sets of fresh Skolem
   functions $\ffs', \ggs'$, clausal formulas $F',G'$, and sets $\us'
   = \var{F'}, \vs' = \var{G'}$ of variables such that:
   \[\begin{array}{r@{\hspace{1em}}l}
   \text{(a)} & F \equiv \exists \ffs' \forall \us' F'
    \text{ and } \lnot G \equiv \exists \ggs' \forall \vs' G'.\\     
   \text{(b)} & \voc{F'} \subseteq \voc{F} \cup \ffs' \text{ and }
     \voc{\lnot G'} \subseteq \voc{G} \cup \ggs'.\\
   \text{(c)} & \forall \us' \forall \vs' (F' \land G') \entails \false.\\
   \end{array}
   \]
 \item \label{step-tab} \name{Tableau computation:} Compute a closed clausal
   tableau for the clausal formula $F' \land G'$. If the tableau is
   not already leaf-closed, convert it to leaf-closed form by removing all
   edges that originate in closing nodes.
   
 \item \label{step-grounding} \name{Tableau grounding:} (Recall that the
   tableau may have rigid variables, that is, variables whose scope is the
   whole tableau may occur in literal labels.) Instantiate all variables of
   the tableau with arbitrary ground terms constructed from symbols in
   $\fun{F' \land G'}$ and, if there is no constant in this set of
   function symbols, a fresh constant $\fc_0$. (Options for choosing these
   ground terms will be discussed in Sect.~\ref{sec-grounding}.)  Observe that
   the grounded tableau is still a leaf-closed tableau for $F' \land
   G'$.

 \item \label{step-side-assignment} \name{Side assignment:} Convert the ground
   tableau to a \sided tableau for $F'$ and $G'$ by attaching appropriate
   \name{side} labels to all nodes except the root. This is always possible
   because every clause of the tableau is an instance of a clause in $F'$ or
   in $G'$. (It is possible that a clause of the tableau is an instance of a
   clause in $F'$ \emph{and} of a clause in~$G'$.  See
   Sect.~\ref{sec-grounding}.)

 \item \label{step-extract} \name{Ground interpolant extraction:} Let $\HG$ be
   the value of $\nipol{N_0}$, where $N_0$ is the root of the tableau.

 \item \label{step-lifting} \name{Interpolant lifting:} Let $\ffs \eqdef \ffs'
   \cup (\fun{F} \setminus \fun{G}) \cup \{\fc_0\}$ and let $\ggs \eqdef \ggs'
   \cup (\fun{G} \setminus \fun{F})$. ($\fc_0$ needs only to be considered if
   it has been introduced in step~\ref{step-grounding} because there is no
   constant occurring in the clausal formula $F' \land G'$. It may here be
   placed also in $\ggs$ instead of $\ffs$.) Observe that $F, G, \ffs, \ggs,
   \HG$ form the first components of an interpolant lifting base. Let $H$ be
   the Craig-Lyndon interpolant of~$F$ and~$G$ specified in
   Theorem~\ref{thm-lifting} with respect to $F, G, \ffs, \ggs, \HG$ as first
   components of an interpolant lifting base.
 \end{enumerate}

 \algooutput Return $H$.  The output is a
 Craig-Lyndon interpolant of the input sentences.
\end{proc}

\noindent
Skolemization and clausification (step~\ref{step-pre}) might integrate
preprocessing operations such as structural normal form conversion and
versions of well-known clausal preprocessing techniques that, however, need to
be specially adapted as for interpolation it is not sufficient to just
preserve unsatisfiability. This is discussed below in Sect.~\ref{sec-preproc}.
The tableau computation itself (step~\ref{step-tab}) is just a refutation
task, entirely independent from its use for interpolation. That is, a clausal
tableau prover can be used there without making any changes to its internal
workings.  The duplication and instantiation of the input first-order clauses
is in essence performed in this step by the tableau prover.\footnote{This
  contrast to approaches like \cite{christ:2010}, where instantiation is
  performed specifically for interpolation.} Nevertheless, tableau
construction methods may yield tableaux in which variables are instantiated
through unification just as far as required to ensure that the tableau is
closed.  The purpose of the tableau grounding (step~\ref{step-grounding}) is
to instantiate any remaining variables in the closed tableau to ground terms.
This is a mere linear operation, where, however, different options are
possible that have effect on the resulting interpolant as discussed in
Sect.~\ref{sec-grounding}. Side assignment (step~\ref{step-side-assignment})
and ground interpolant extraction (step~\ref{step-extract}) then operate on
the resulting ground tableau.

A lifting base $\la F, G, \ffs, \ggs, \HG, \FE, \GE, \sth \ra$ that justifies
the application of interpolant lifting (step~\ref{step-lifting}) can be
determined as follows: Let $F, G, \ffs, \ggs, \HG$ are as specified in the
procedure description.  To determine the remaining components consider the
tableau after grounding and side assignment.  Take as $\FE\sth$ the
conjunction of the tableau clauses with side $\aaa$ and as $\GE\sth$ the
negation of the conjunction of the tableau clauses with side~$\bbb$.  By
comparison with the clausal formulas $F'$ and $G'$ constructed by the
procedure, the formulas $\FE\sth$ and $\GE\sth$ can be separated into suitable
formulas $\FE$ and $\GE$ and a substitution $\sth$.

The procedure can easily be adapted to handle not just sentences, but also
formulas with free variables as in- and outputs, as required by the full
definition of Craig-Lyndon interpolant: In a preprocessing step, the free
variables of $F$ and~$G$ would be converted to constants and in a
postprocessing step the occurrences in $H$ would be converted back to the
corresponding free variables.

If the invoked method for tableau computation is complete, that is, it
computes a closed clausal tableau for all unsatisfiable inputs, then, like the
\CTIG procedure, the \CTIF procedure is complete, that is, it outputs a
Craig-Lyndon interpolant of its inputs $F$ and $G$ whenever $F \entails G$.
Also as for \CTIG, the size of the result formula is linear in the size
of the clausal tableau, or more precisely in the number of its leaves whose
target has the opposite side label. The lifting step does not change this.
(Of course, as for \CTIG, the size of the tableau itself is not polynomially
bounded.)

\section{Proof of the Interpolant Lifting Theorem}
\label{sec-lift-proof}

In this section we prove Theorem~\ref{thm-lifting}.  The proof statement is
given at the end of the section. It resides on the definition of
\name{interpolant lifting base} (Definition~\ref{def-ib}) and a lemma
concerning the semantic properties of interpolants that is developed
throughout the section.  We start from a given interpolant lifting
base. Further auxiliary formulas, quantifiers, variables and substitutions are
then defined and propositions that relate them are stated.
These auxiliary elements are specified in \name{definition} environments,
which are used here differently from the other sections to specify elements
that are only of relevance within the section. Correspondingly, the
definitions and statements in this section directly refer to the components
from the given interpolant lifting base and elements defined previously in the
section.

To prove Theorem~\ref{thm-lifting}, it has to be shown that the sentence $H =
Q_1 v_1 \ldots Q_n v_n\, \HG\invsubstpost{\stt}$, as specified in the theorem
is a Craig-Lyndon interpolant of sentences $F$ and $G$. This is the case if $F
\entails H$, $H \entails G$ and $\voc{H} \subseteq \voc{F} \cap \voc{G}$.  Of
the two semantic conditions, we focus on $F \entails H$, as $H \entails G$ can
be shown analogously.  We assume in this section a context with a given
interpolant lifting base
\[\la F, G, \ffs, \ggs, \HG, \FE, \GE, \sth\ra\]
and the symbols
  \[\fgs,\; \stt,\; v_1, \ldots, v_n,\; \text{ and } Q_1, \ldots, Q_n\]
defined as in the preconditions of Theorem~\ref{thm-lifting}.
The following definition specifies, along with some shorthands, an injective
ground substitution $\stz$ that extends the \emph{lifting} substitution $\stt$
by ranging over all $\fgs\sterms$ that occur in $\FE\sth$ or in $\HG$.
\begin{defn}[Formula $\bm{\FG}$, Quantified Variables $\bm{\R_1 \w_1, \ldots,
      \R_m \w_m}$, Injective
    Substitution $\bm{\stz}$, Sets of Variables $\bm{\xwf, \ywg}$,
    Formula $\bm{\FQ}$]
  \label{def-stex-ri-fq}

  \sdlab{def-fg} Define the ground formula \[\FG \eqdef \FE\sth.\]
  
  \sdlab{def-stex} Define a set of variables $\{\w_1,\ldots, \w_m\} \supseteq
  \{v_1,\ldots,v_n\}$ and an injective substitution~$\stz$ (the subscript
  \textsc{lft} suggesting \name{lifting}) with the following properties:
\begin{enumerate}[label={(\alph*)},leftmargin=3.5em]
\item \label{sk:1}  $\dom{\stz}\; =\; \{\w_1,\ldots,\w_m\}$.
\item \label{sk:2}  $\rng{\stz}\; =\;  \{t \mid t \text{ is an } \fgs\sterm
  \text{ occurring in } \FG \text{ or in } \HG\}$.
\item \label{sk:3}  If $\strictsubterm{\w_i\stz}{\w_j\stz}$, then $i < j$.
\item \label{sk:4}   $\{v_1,\ldots, v_n\} \subseteq \{\w_1,\ldots, \w_m\}$.
\item \label{sk:5}  If $\w_i = v_j$, then $\w_i\stz = v_j\stz$.
\item \label{sk:6}  If $\w_i = v_j$, $\w_k = v_l$ and $j < l$, then $i < k$.
\end{enumerate}

\sdlab{def-wfg} Define the shorthands $\xwf \eqdef \{\w_i \mid \w_i\stz \in
\ffs\sterms\}$ and $\ywg \eqdef \{\w_i \mid \w_i\stz \in \ggs\sterms\}$.

\medskip

\sdlab{def-ri}
For $i \in \{1,\ldots,m\}$ define $\R_i \eqdef \forall$ if $\w_i \in \ywg$
and define $\R_i \eqdef \exists$ otherwise, that is, if
$\w_i \in \xwf$.

\medskip

\sdlab{def-fq}
Define the quantifier-free formula
\[\FQ\; \eqdef\; \invsubst{\FG}{\stz}.\]

\end{defn}
Definition~\ref{def-stex} characterizes the specified elements just in terms
of properties. The following proposition supplements this by claiming their
existence and supplementing a construction in its proof.
\begin{prop}[Existence of $\bm{\{\w_1,\ldots, \w_m\}}$ and $\bm{\stz}$]
  \label{prop-aux-stex}
  There exists a set $\{\w_1,\ldots, \w_m\}$ of variables and a
  substitution $\stz$ as specified in Definition~\ref{def-stex}.
\end{prop}
\begin{proof} The set $\{\w_1,\ldots, \w_m\}$ of variables and the
  substitution $\stz$ can be constructed as follows: Collect all $\fgs\sterms$
  occurring in $\FG$ or $\HG$. They form a finite set of ground terms, a
  superset of $\{v_1\stt, \ldots, v_n\stt\}$. Let $\{t_1, \ldots, t_m\}$ be
  this set, ordered such that $t_1 < \ldots < t_m$ extends the ordering
  $v_1\stt < \ldots < v_n\stt$ (i.e., if $t_i = v_k\stt$, $t_j = v_l\stt$ and
  $k < l$, then $i < j$) and that if $\strictsubterm{t_i}{t_j}$, then $i < j$.
  Define $\stz \eqdef \{\w_1\mapsto t_1, \ldots, \w_m\mapsto t_m\}$. \qed
\end{proof}
The injective substitution $\stz$ extends $\stt$ that was specified in the
preconditions of the interpolant lifting theorem with pairs that map
additional variables to $\fgs\sterms$ that have an occurrence in $\FG$ or
$\HG$ which is not maximal or is not in $\HG$.  Correspondingly, the
quantifier prefix $\R_1 \w_1 \ldots \R_m \w_m$ ``includes'' $Q_1 v_1 \ldots
Q_n v_n$. That is, $\R_1 \w_1 \ldots \R_m \w_m$ could be obtained from $Q_1
v_1 \ldots Q_n v_n$ by adding to the front, the end, or in-between additional
quantifications upon those variables in $\{\w_1,\ldots,\w_m\}$ that are not in
$\{v_1,\ldots,v_n\}$.  We now define the shorthand~$\HQ$ for the
quantifier-free formula that follows the quantifier prefix in the interpolant
lifting result.

\begin{defn}[Formula $\bm{\HQ}$]
  Define the quantifier-free formula
  \[\HQ \eqdef \invsubst{\HG}{\stz}.\]
\end{defn}
We note the following properties of $\HQ$.
\begin{prop}[Properties of $\bm{\HQ}$]
  \label{prop-hq}
  \smallskip
  
  \slab{prop-hq-alt}
  $\HQ = \invsubst{\HG}{\stt}$.

  \slab{prop-fq-imp-hq}
  $\FQ \entails \HQ$.

  \slab{prop-sk-12}  $\R_1 \w_1 \ldots \R_m \w_m\, \HQ\; \equiv\;
  Q_1 v_1 \ldots Q_n v_n\, \HQ$.
\end{prop}
\begin{proof}
  (\ref{prop-hq-alt}) By the definitions of $\stz$ and $\stt$ the inverse
  application of either substitution to $\HG$ has the effect that all
  $\fgs\sterms$-maximal occurrences of $\fgs\sterms$ are replaced with
  variables, and, moreover, ensured by conditions~\ref{sk:4} and~\ref{sk:5} of
  Definition~\ref{def-stex}, in both cases the same variables.
  (\ref{prop-fq-imp-hq}) By condition~\ibref{ib:ipol} of the characterization
  of \name{lifting base}, Definition~\ref{def-ib}, $\HG$ is a Craig-Lyndon
  interpolant of $\FE\sth$. Hence $\FE\sth \entails \HG$, which, by the
  definition of $\FG$ can be expressed as $\FG \entails \HG$.  By the
  definitions of $\FQ$ and~$\HQ$ the proposition statement can be expressed as
  $\invsubst{\FG}{\stz} \entails \invsubst{\HG}{\stz}$, which follows from
  $\FG \entails \HG$ by Proposition~\ref{prop-invsubst-entails}.
  (\ref{prop-sk-12}) Follows from the definitions of the prefixes $\R_1 \w_1
  \ldots \R_m \w_m$ and $Q_1 v_1 \ldots Q_n v_n$ since $\var{\HQ} \subseteq
  \{v_1,\ldots,v_n\} \subseteq \{\w_1, \ldots, \w_m\}$.  \qed
\end{proof}

The following example illustrate the elements introduced so far.
\begin{examp}[Formulas, Substitutions and Quantifier Prefixes Introduced So Far]
  \label{examp-items}
  The following table shows values for formulas, substitutions and quantifier
  prefixes defined so far in this section for the lifting base from
  Example~\ref{ex-ib-qr} as starting point. Properties stated with
  Proposition~\ref{prop-hq} can be easily verified for the example values.
  \[\begin{array}{r@{\hspace{0.5em}}c@{\hspace{0.5em}}l}
  F & = & \forall x\, \fp(x, \ff(x)) \land \forall x \forall y\,
  \fq(\ff(x),y).\\
    \ffs & = & \{\ff\}.\\
    \ggs & = & \{\fg_1, \fg_2\}.\\
    \HG & = & \fp(\fg_1, \ff(\fg_1)).\\
    \sigma & = & \{v_1 \mapsto \fg_1, v_2 \mapsto \ff(\fg_1)\}.\\
    Q_1 v_1 \ldots Q_m v_m & = & \forall v_1 \exists v_2.\\
    H & = & \forall v_1 \exists v_2\, \fp(v_1, v_2).\\
    \FE & = & \fp(u_1, \ff(u_1)) \land \fq(\ff(u_3), u_4).\\
    \sth & = & \{u_1 \mapsto \fg_1, u_2 \mapsto \ff(\fg_1),
    u_3 \mapsto \fg_2(\ff(\fg_1)), u_4 \mapsto \fg_1\}.\\
    \FG = \FE\sth & = &
    \fp(\fg_1, \ff(\fg_1)) \land \fq(\ff(\fg_2(\ff(\fg_1))), \fg_1).\\
    \stz & = &
    \{\w_1 \mapsto \fg_1, \w_2 \mapsto \ff(\fg_1),
    \w_3 \mapsto \fg_2(\ff(\fg_1)),
    \w_4 \mapsto \ff(\fg_2(\ff(\fg_1))))\},\\
    && \text{where } \w_1 = v_1 \text{ and } \w_2 = v_2.\\
    \xwf & = & \{\w_2, \w_4\}.\\
    \ywg & = & \{\w_1, \w_3\}.\\
    \R_1 \w_1 \ldots \R \w_m & =
    & \forall \w_1 \exists \w_2 \forall \w_3 \exists \w_4.\\
    \FQ = \invsubst{\FG}{\stz} & =
    & \fp(\w_1, \w_2) \land \fq(\w_4,\w_1).\\
    \HQ = \invsubst{\HG}{\stz} & = & \fp(\w_1, \w_2).\\
  \end{array}
  \]
\end{examp}

\medskip

Based on variable sets $\xwf$ and $\ywg$ we now define a series of
substitutions and a series of subsets of variables that will be used
later in an induction.
\begin{defn}[Substitutions $\bm{\stski}_i$ and Variable Sets $\bm{\ywg}_i$]
  \label{def-i}
  For $i \in \{0,\ldots,m\}$ define substitutions $\stski_i$ and sets $\ywg_i$
  of variables as follows:

  \medskip
  
\sdlab{def-i-stsk}
$\stski_i\; \eqdef\;
\{\w_j \mapsto \invsubst{\w_j\stz}{\stz |_{\ywg}} \mid
   \w_j \in \xwf \text{ and } j > i \}.$

\sdlab{def-i-xs} $\ywg_i\; \eqdef\; \{\w_j \mid \w_j \in \ywg \text{ and } j >
i\}.$

\sdlab{def-stsk} In addition, define the shorthand $\stsk \eqdef \stski_0$.
(Observe that $\ywg = \ywg_0$.)

\end{defn}
The domain of substitution $\stski_i$ is the set of those $\w_j$ with $j$
strictly larger than~$i$ that are mapped by $\stz$ to an $\ffs\sterm$.  The
value of $\w_j\stski_i$ is the value of $\w_j\stz$ (an $\ffs$\sterm) after
replacing all maximal occurrences of $\ggs\sterms$ with the respective
variables $\w_k$ that are mapped by~$\stz$ to them. For
Example~\ref{examp-items}, we obtain
\[
\begin{array}{r@{\hspace{0.5em}}c@{\hspace{0.5em}}l}
  \stsk = \stski_0 = \stski_1 & = &
  \{\w_2 \mapsto \ff(\w_1), \w_4 \mapsto \ff(\w_3)\}.\\
  \stski_2 = \stski_3 & = & \{\w_4 \mapsto \ff(\w_3)\}.\\
  \stski_4 = \stski_m & = & \emptysubst.\\
\end{array}
\]
The set of variables $\ywg_i$ is the sets of those~$\w_j$ with $j$ strictly
larger than~$i$ that are mapped by $\stz$ to a $\ggs\sterm$.  For
Example~\ref{examp-items} the values of $\ywg_i$ are $\ywg = \ywg_0 =
\{\w_1,\w_3\}$, $\ywg_1 = \ywg_2 = \{\w_3\}$, and $\ywg_3 = \ywg_4 = \ywg_m =
\emptyset$.  Members of $\xwf$ have the following properties, which will be
used later in the proof of Lemma~\ref{lem-lift-sem}.
\begin{prop}[Properties of Members of $\bm{\xwf}$]
  \label{prop-stsk-aux}
  For all $\w_i \in \xwf$ it holds that

  \smallskip
  
  \slab{prop-stsk-aux-rng}
  $\w_i \notin \varrng{\stski_i}$.
  
  \slab{prop-stsk-aux-ywg}
  $\ywg_i \cap \var{\w_i\stski_{i-1}} = \emptyset$.
\end{prop}

\begin{proof}
  \prlReset{prop-stsk-aux-ywg}

  \smallskip
  (\ref{prop-stsk-aux-rng}) From the definitions of~$\stski_i$ and~$\ywg$ and
  given that $\stz$ is a ground substitution it follows that
  $\varrng{\stski_i} \subseteq \ywg$. Since $\xwf \cap \ywg = \emptyset$,
  the proposition then follows from the precondition $\w_i \in \xwf$.

  (\ref{prop-stsk-aux-ywg}) We derive a contradiction from assuming that,
  contrary to the proposition, there exists a $j$ such that:
  \[
  \begin{arrayprf}
    \prl{prop-stsk-a1} & \w_j \in \ywg_i.\\
    \prl{prop-stsk-a2} & \w_j \in \var{\w_i\stski_{i-1}}.
  \end{arrayprf}
  \]
  The derivation proceeds in the following steps explained below:
  \[
  \begin{arrayprf}
    \prl{prop-stsk-j-gt-i} & j > i.\\
    \prl{prop-stsk-stz-empty} & \var{\w_i\stz} = \emptyset.\\
    \prl{prop-stsk-undo} &
    \w_i\stz = \w_i\stz\invsubstpost{\stz |_{\ywg}}\stz |_{\ywg}.\\
    \prl{prop-stsk-stz} &
    \w_i\stz = \w_i\stz\invsubstpost{\stz |_{\ywg}}\stz.\\
    \prl{prop-stsk-inv} &
    \w_i\stski_{i-1} = \w_i\stz\invsubstpost{\stz |_{\ywg}}.\\
    \prl{prop-stsk-stz-stski} &
    \w_i\stz = \w_i\stski_{i-1}\stz.\\
    \prl{prop-stsk-subterm} &
    \strictsubterm{\w_j\stz}{\w_i\stski_{i-1}\stz}.\\
    \prl{prop-stsk-w} &
    \strictsubterm{\w_j\stz}{\w_i\stz}.\\
    \prl{prop-stsk-j-lt-i} &
    j < i.\\
  \end{arrayprf}
  \]
  Step~\pref{prop-stsk-j-gt-i} follows from \pref{prop-stsk-a1} and the
  definition of $\ywg_i$.  Step~\pref{prop-stsk-stz-empty} holds since $\stz$
  is a ground substitution and $\w_i \in \dom{\stz}$.
  Step~\pref{prop-stsk-undo} follows from
  Proposition~\ref{prop-invsubst-subst}, whose precondition $\dom{\stz
    |_{\ywg}} \cap \var{\w_i\stz} = \emptyset$ is implied by
  \pref{prop-stsk-stz-empty}. Step~\pref{prop-stsk-stz} follows from
  \pref{prop-stsk-undo} and \pref{prop-stsk-stz-empty};
  step~\pref{prop-stsk-inv} from the definition of $\stski_{i-1}$;
  step~\pref{prop-stsk-stz-stski} from \pref{prop-stsk-stz} and
  \pref{prop-stsk-inv}; step~\pref{prop-stsk-subterm} from~\pref{prop-stsk-a2}
  and the definition of~$\stski_{i-1}$, which precludes that
  $\w_i\stski_{i-1}$ is just a variable; step~\pref{prop-stsk-w} from
  \pref{prop-stsk-subterm} and \pref{prop-stsk-stz-stski}; and, finally,
  step~\pref{prop-stsk-j-lt-i}, which contradicts~\pref{prop-stsk-j-gt-i},
  from \pref{prop-stsk-w} and property~\ref{sk:3} of
  Definition~\ref{def-stex}.  \qed
\end{proof}

Both formulas $\FE$ and $\FQ$ generalize the ground formula~$\FG$, but in
different ways, which are reconciled by observing that $\FQ\stsk$ is a
(possibly non-ground) instance of $\FE$.  As an example consider $\FG =
\fp(\ff(\fg(\ff(\fa))),\fg(\ff(\fa)),\fg(\ff(\fa)))$, which is an instance of
both $\FE = \fp(\ff(u_1),u_1,u_2)$ and $\FQ = \fp(x,y,y)$.  If $\stsk = \{x
\mapsto \ff(y)\}$, then $\FQ\stsk = \fp(\ff(y),y,y)$ is more general than
$\FG$ and an instance of both $\FE$ and $\FQ$.  (This example will be fleshed
out further in Example~\ref{examp-lifting} below.)  The following definition
specifies the substitution~$\stmerge$, which is then used in
Proposition~\ref{prop-stmerge} to justify that $\FQ\stsk$ is an instance of
$\FE$.
\begin{defn}[Substitution $\bm{\stmerge}$]
  \label{def-stmerge}
  Define the substitution~$\stmerge$ as
\[\stmerge \eqdef \{u \mapsto u\sth\invsubstpost{\stz |_{\ywg}}\mid u \in \var{\FE}\}.\]
\end{defn}

\begin{prop}[Key Property of $\bm{\stmerge}$]
  \label{prop-stmerge}
  \[\FE\stmerge = \FQ\stsk.\]
\end{prop}
\begin{proof}
  \prlReset{lem-stmerge} We show this with the help of two auxiliary
  substitutions $\stv$ and~$\stren$ and an auxiliary quantifier-free formula
  $\FA$ that generalizes $\FE$ and $\FQ$.  These auxiliary objects are not
  used elsewhere in the proof of Theorem~\ref{thm-lifting}.  Let $\{t_1,
  \ldots, t_k\}$ be the set of all $\ffs\sterms$ with an $\ffs\sterms$-maximal
  occurrence in $\FE$.  Let $u_1, \ldots, u_k$ be fresh variables, define the
  substitution~$\stv$ as
\[\stv\; \eqdef\;
\{u_i \mapsto t_i \mid i \in \{1, \ldots, k\}\},\] and define the
quantifier-free formula~$\FA$ as
  \[\FA \eqdef \invsubst{\FE}{\stv}.\]  Since
  $\dom{\stv} \cap \var{\FE} = \emptyset$ it follows from
  Proposition~\ref{prop-invsubst-subst} that $\invsubst{\FE}{\stv}\stv = \FE$. Thus:
  \[\begin{arrayprf}
  \prl{fa-fe} & \FA\stv = \FE.
  \end{arrayprf}
  \]
  Define the substitution~$\stren$ as
  \[\stren\; \eqdef\; \{u \mapsto \invsubst{u\stv\sth}{\stz} \mid
  u \in \var{\FA}\}.\]
  Because the range of~$\stz$ contains only $\fgs\sterms$ and in~$\FA$ there
  are no occurrences of $\fgs\sterms$, inversely applying~$\stz$ to an
  instance of~$\FA$, say $\FA\xi$, yields the same result as applying to $\FA$
  the substitution that maps each variable $u$ occurring in $\FA$ to the value
  of $u\xi$ after inversely applying~$\stz$.  For the particular case of
  $\stv\sth$ in the role of~$\xi$ this can be stated formally as
  \[\begin{arrayprf}
  \prl{fa-stren} & \FA\stv\sth\invsubstpost{\stz} = \FA\{u \mapsto
  u\stv\sth\invsubstpost{\stz} \mid u \in \var{\FA}\}.
  \end{arrayprf}
  \]
  Observe that the right side of~\pref{fa-stren} is equal to $\FA\stren$.  By
  \pref{fa-fe} and by taking into account the definitions of $\FG$ and $\FQ$
  we can conclude via $\FA\stren = \FA\stv\sth\invsubstpost{\stz} =
  \FE\sth\invsubstpost{\stz} = \FG\invsubstpost{\stz} = \FQ$ that
  \[\begin{arrayprf}
  \prl{fa-fq} & \FA\stren = \FQ.
  \end{arrayprf}
  \]
  Recall that $\FG$ is a ground formula and that $\stz$ is an injective
  substitution whose range includes all $\fgs\sterms$ occurring in $\FG$.
  Inversely applying the restriction of $\stz$ to $\ggs\sterms$ to $\FG$ has
  the same result as inversely applying $\stz$ to $\FG$ followed by replacing
  those variables that are mapped by $\stz$ to an $\ffs\sterm$ with the result
  of inversely applying to that respective $\ffs\sterm$ the restriction of
  $\stz$ to $\ggs\sterms$. With the first description on the right side,
  this equality is formally stated as
  \[\begin{arrayprf}
  \prl{aux-subst} & \FG\invsubstpost{\stz}\{x \mapsto x\stz\invsubstpost{\stz
    |_{\ywg}} \mid x \in \xwf\} = \FG\invsubstpost{\stz |_{\ywg}}.
  \end{arrayprf}
  \]
  That $\FE\stmerge = \FQ\stsk$ can now be shown in the following steps,
  explained below, proceeding from the right to the left side:
  \[
  \begin{arrayprfeq}
    \prl{e:1} & & \FQ\stsk\\
    \prl{e:2} & = & \FA\stren\stsk\\
    \prl{e:3a} & = & \FA\stv\sth\invsubstpost{\stz}\stsk\\
    \prl{e:3} & = & \FA\stv\sth\invsubstpost{\stz}
    \{x \mapsto x\stz\invsubstpost{\stz |_{\ywg}} \mid x \in \xwf \}\\
    \prl{e:4} & = & \FA\stv\sth\invsubstpost{\stz |_{\ywg}}\\
    \prl{e:5} & = & \FA\stv\{u \mapsto u\sth\invsubstpost{\stz |_{\ywg}}\mid u \in \var{\FA\stv}\}\\
    \prl{e:61} & = & \FE\{u \mapsto u\sth\invsubstpost{\stz |_{\ywg}}\mid u \in \var{\FE}\}\\
      \prl{e:7} & = & \FE\stmerge.
  \end{arrayprfeq}
  \]
  Equality of~\pref{e:1} to~\pref{e:2} follows from~\pref{fa-fq}.  Equality
  to~\pref{e:3a} follows since by the definition of~$\stren$
  and~\pref{fa-stren} it holds that $\FA\stren =
  \FA\stv\sth\invsubstpost{\stz}$. Equality to~\pref{e:3} follows from the
  definition of $\stsk$ (Definition~\ref{def-i}).  Equality to~\pref{e:4}
  follows from~\pref{aux-subst}, since $\FA\stv\sth = \FE\sth = \FG$, where
  the first equality in the sequence follows from~\pref{fa-fe} and the second
  one holds by~Definition~\ref{def-fg}.  Equality to~\pref{e:5} follows with
  similar arguments as step~\pref{fa-stren} above: Because the range of
  $\stz|_{\ywg}$ contains only $\ggs\sterms$ and there are no occurrences of
  $\ggs\sterms$ in $\FA\stv$, inversely applying~$\stz|_{\ywg}$ to an instance
  of~$\FA\stv$, say $\FA\stv\xi$, yields the same result as applying to
  $\FA\stv$ the substitution that maps each variable $u$ occurring in
  $\FA\stv$ to the value of $u\xi$ after inversely applying~$\stz|_{\ywg}$.
  For the particular case of $\sth$ in the role of~$\xi$ this is formally
  stated as the equality of~\pref{e:4} and~\pref{e:5}.  Equality
  to~\pref{e:61} follows from~\pref{fa-fe}, and equality to~\pref{e:7} by
  contracting the definition of~$\stmerge$, Definition~\ref{def-stmerge}.
  \qed
\end{proof}
As a side remark, we note that the proof of Proposition~\ref{prop-stmerge}
implies with its step~\prefGlobal{lem-stmerge:e:4} a third characterization of
the formulas equated by the proposition, in terms of the ground formula $\FG$
with an inversely applied substitution: $\FE\stmerge = \FQ\stsk =
\FG\invsubstpost{\stz |_{\ywg}}$.

The following example illustrates the intermediate formulas and substitutions
introduced so far to prove the interpolant lifting theorem.
\begin{examp}[Formulas and Substitutions Involved
    in the Proof of Interpolant Lifting]
  \label{examp-lifting}
  The diagram in Figure~\ref{fig-diagram} shows instance
  relationships of some of the quantifier-free formulas and substitutions used
  to prove interpolant lifting, along with examples.  A formula connected with
  a downward line to another formula is, after applying the substitution that
  labels the line, identical to the lower formula.  Injective substitutions
  are indicated by thick lines.  The shown symbolic formula and substitution
  names refer to an implicitly given interpolant lifting base with $\ffs =
  \{\ff\}$ and $\ggs = \{\fg\}$ and to the definitions in this section.
  (Exceptions are $\FA$ and~$\stv, \stren$, which are only used locally within
  the proof of Proposition~\ref{prop-stmerge}.)
  \begin{figure}[H]
    \centering
  \noindent
  \hspace{-0.68cm}
\begin{tikzpicture}[xscale=0.38,yscale=0.38,
    baseline={([yshift=-8.5pt]a.north)}]
  \draw [thick] (0,0) -- node[above]
        {$\stv\;\;\;\;$\al} (6,4);
   \draw (6,4) -- node[above]
         {$\;\;\;\;\stren$\al} (12,0);
   \draw (0,0) -- node[above] {$\;\;\;\;\stmerge$\al} (6,-4);
   \draw (12,0) -- node[above] {$\stsk\;\;\;\;$\al} (6,-4);
   \draw (0,0) -- node[below] {\al$\sth\;\;\;\;\,$} (6,-9);
   \draw (12,0) [thick] -- node[below] {$\;\;\;\;\;\stz$\al} (6,-9);
   \draw [thick] (6,-4) -- node[fill=white,inner sep=0pt,above=-4pt] {$\!\!\!\stz |^{}_{\ywg}\!\!\!$} (6,-9);
   \node (a) [fill=white] at (6,4) {$\FA = \fp(u_3,u_1,u_2)$};
   \node [fill=white] at (0,0) {$\hspace{0.5cm}\FE = \fp(\ff(u_1),u_1,u_2)$};
   \node [fill=white] at (12,0) {$\FQ = \fp(x,y,y)\hspace{0.5cm}$};
   \node [fill=white] at (6,-9)
         {$\FG = \fp(\ff(\fg(\ff(\fa))),\fg(\ff(\fa)),\fg(\ff(\fa)))$};
   \node [fill=white] at (6,-4) {$\fp(\ff(y),y,y)$};         
\end{tikzpicture}%
\hspace{0.21cm}%
$\begin{array}[t]{rcl}
  \multicolumn{3}{c}{\text{Substitutions}}\\[3pt]
  \stv & = & \{u_3 \mapsto \ff(u_1)\}\\
  \stren & = & \{u_3 \mapsto x, u_1 \mapsto y, u_2 \mapsto y\}\\
  \sth & = & \{u_1 \mapsto \fg(\ff(\fa)), u_2 \mapsto \fg(\ff(\fa))\}\\
  \stz & = & \{x \mapsto \ff(\fg(\ff(\fa))), y \mapsto \fg(\ff(\fa))\}\\
  \stsk & = & \{x \mapsto \ff(y)\}\\
  \stmerge & = & \{u_1\mapsto y, u_2 \mapsto y\}
\end{array}$
\caption{Instance relationships between formulas involved in
  the proof of interpolant lifting with example values.}
\label{fig-diagram}
  \end{figure}
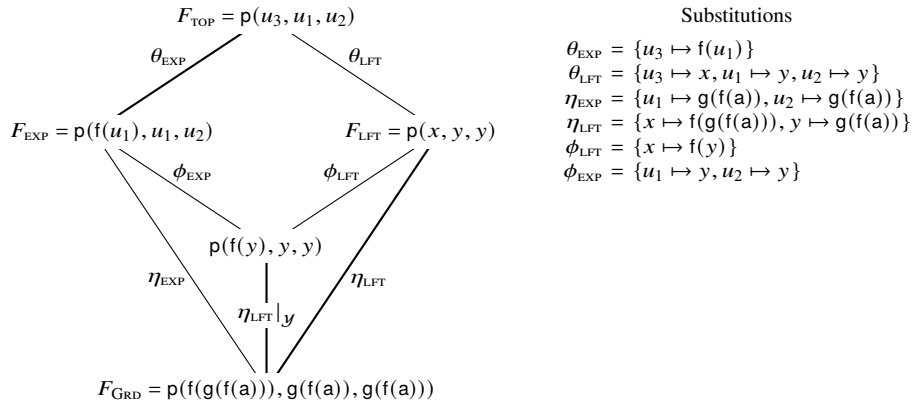
\end{examp}

\medskip

\noindent
We are now ready to prove the required semantic property of the formula
obtained by interpolant lifting with the following lemma.
\begin{lem}[Semantic Justification of Interpolant Lifting: From $\bm{F}$ to
  the Interpolant]
  \label{lem-lift-sem}
  \[F \entails Q_1 v_1 \ldots Q_n v_n\, \invsubst{\HG}{\stt}.\]
\end{lem}

\begin{proof}
  \prlReset{lem-lift-sem} That the lemma holds under the assumption of the
  following statement~\pref{lss:1} can be shown in the subsequent steps
  explained below:
  \[
  \begin{arrayprfmeq}
    \prl{lss:1}  && F \entails \exists \ffs \R_1 \w_1 \ldots \R_m \w_m \forall \ywg_m\,
    \FQ\stski_m\\
    \prl{lss:2}  & \miff & F \entails \R_1 \w_1 \ldots \R_m \w_m\,
    \FQ\\
    \prl{lss:3}  & \mimp & F \entails \R_1 \w_1 \ldots \R_m \w_m\,
    \HQ\\
    \prl{lss:4}  & \miff & F \entails Q_1 v_1 \ldots Q_n v_n\,
    \HQ\\
    \prl{lss:5}  & \miff & F \entails Q_1 v_1 \ldots Q_n v_n\,
    \invsubst{\HG}{\stt}.
  \end{arrayprfmeq}  
  \]
  The equivalence of~\pref{lss:1} to~\pref{lss:2} holds since it follows from
  the definitions of $\stski_i$ and $\ywg_i$ that $\stski_m = \emptysubst$ and
  $\ywg_m = \emptyset$, and from the definitions of $\FQ$ and $\stz$ that
  members of $\ffs$ do not occur in $\FQ$. The implication of~\pref{lss:3}
  follows from Proposition~\ref{prop-fq-imp-hq}, equivalence to~\pref{lss:4}
  from Proposition~\ref{prop-sk-12} and equivalence to~\pref{lss:5}, the
  statement to prove, from unfolding the definition of $\HQ$.

  It remains to show assumption~\pref{lss:1}. We show by induction that
  actually for all $i \in \{0, \ldots, m\}$ it holds that
  \begin{equation}
    \label{eq-lss-1-indu}
  F \entails \exists \ffs \R_1 \w_1 \ldots \R_i \w_i
  \forall \ywg_i \FQ\stski_i,\tag{$*$}
  \end{equation}
  which, of course, with the case $i=m$ includes assumption~\pref{lss:1}.  We
  first consider the base case where $i=0$.  Let $\us \eqdef \var{\FE}$.  From
  the key property of $\stmerge$ (Proposition~\ref{prop-stmerge}) it follows
  that
  \[\begin{arrayprf}
  \prl{merge} & \forall \us\, \FE\; \entails\; \forall \ywg\, \FE\stmerge
  \;=\; \forall \ywg\, \FQ\stsk.\\
  \end{arrayprf}  
  \]
  Condition~\ibref{ib:sem} of the definition of \name{interpolant
    lifting base} (Definition~\ref{def-ib}) states that $F \entails \exists
  \ffs \forall \us\, \FE$. With \pref{merge} this implies $F \entails \exists
  \ffs \forall \ywg\, \FQ\stsk$, which, since $\ywg = \ywg_0$ and $\stsk =
  \stski_0$, can be written as $F \entails \exists \ffs \forall \ywg_0\,
  \FQ\stski_0$, that is, the statement to show for the base case.

  We now consider the induction step. As induction hypothesis we assume that
  (\ref{eq-lss-1-indu}) holds for some $i \in \{0,\ldots,m-1\}$. The induction
  conclusion, that is, (\ref{eq-lss-1-indu}) with $i+1$ in place of $i$,
  follows from the hypothesis since for all $i \in \{0, \ldots, m-1\}$ it
  holds that
  \begin{equation}
    \label{eq-lift-step}
  \R_1 \w_1 \ldots \R_i \w_i \forall \ywg_i\, \FQ\stski_i\; \entails\;
  \R_1 \w_1 \ldots \R_{i+1} \w_{i+1} \forall \ywg_{i+1}\, \FQ\stski_{i+1},\tag{$**$}
  \end{equation}
  which we now show. Let $i$ be a member of $\{0, \ldots, m-1\}$.  The
  variable~$\w_{i+1}$ is then either in~$\xwf$ or in~$\ywg$. We show
  (\ref{eq-lift-step}) for both cases separately:
  \begin{itemize}
  
  \item Case $\w_{i+1} \in \xwf$: Then $\R_i = \exists$, $\stski_i = \{\w_{i+1}
    \mapsto \w_{i+1}\stski_i\}\stski_{i+1}$ and $\ywg_{i+1} = \ywg_{i}$.
    Moreover, by Proposition~\ref{prop-stsk-aux} it holds that:
    \[
    \begin{arrayprf}
      \prl{b:22} & \w_{i+1} \notin \varrng{\stski_{i+1}}.\\
      \prl{b:23} & \ywg_{i+1} \cap \var{\w_{i+1}\stski_i} = \emptyset.
    \end{arrayprf}
    \]
    Hence, 
    \[
    \begin{arrayprfeq}
      \prl{c:01} &&   \R_1 \w_1 \ldots \R_i \w_i \forall \ywg_i\, \FQ\stski_i\\
      \prl{c:02} & = &  \R_1 \w_1 \ldots \R_i \w_i \forall \ywg_{i+1}\, \FQ\{\w_{i+1} \mapsto \w_{i+1}\stski_i\}\stski_{i+1}\\
      \prl{c:03} & \entails &  \R_1 \w_1 \ldots \R_i \w_i \exists \w_{i+1} \forall \ywg_{i+1}\, \FQ\stski_{i+1}\\
      \prl{c:04} & = &  \R_1 \w_1 \ldots \R_{i+1} \w_{i+1} \forall \ywg_{i+1}\, \FQ\stski_{i+1},\\
    \end{arrayprfeq}    
\]
where the entailment of~\pref{c:03} by~\pref{c:02} is justified by~\pref{b:22}
and \pref{b:23}.  This concludes the proof of the induction step for the case
$\w_{i+1} \in \xwf$.

\item Case $\w_{i+1} \in \ywg$: Then $\R_i = \forall$, $\stski_{i+1} =
  \stski_i$ and $\ywg_i = \{\w_{i+1}\} \cup \ywg_{i+1}$.
  Hence
  \[
  \begin{arrayprfeq}
   \prl{d:01} && \R_1 \w_1 \ldots \R_i \w_i \forall \ywg_i\, \FQ\stski_i\\
\prl{d:02} & =  & \R_1 \w_1 \ldots \R_i \w_i \forall  \w_{i+1} \forall \ywg_{i+1}\, \FQ\stski_{i+1}\\
\prl{d:03} & = & \R_1 \w_1 \ldots \R_{i+1} \w_{i+1} \forall \ywg_{i+1}\, \FQ\stski_{i+1},
  \end{arrayprfeq}
  \]
  which concludes the proof of the induction step for the case $\w_{i+1} \in
  \ywg$.   \qed
  \end{itemize}
  \end{proof}

\pagebreak
\begin{coro}[Semantic Justification of Interpolant Lifting:
    From the Interpolant to~$\bm{G}$]
  \label{cor-lift-sem}
  \[Q_1 v_1 \ldots Q_n v_n\, \invsubst{\HG}{\stt} \entails G.\]
\end{coro}

\begin{proof}[Sketch] A formula~$H$ is a
  Craig-Lyndon interpolant of formulas $F$ and~$G$ if and only if $\lnot H$ is
  a Craig-Lyndon interpolant of $\lnot G$ and $\lnot F$. Let $\NQ_i \eqdef
  \forall$ if $Q_i = \exists$ and $\NQ_i \eqdef \exists$ if $Q_i =
  \forall$. The corollary statement then can then be expressed as
  \begin{equation}
    \label{eq-cor-lift-sem}
    \lnot G \entails \NQ_1 v_1 \ldots
    \NQ_n v_n\, \lnot \invsubst{\HG}{\stt}.\tag{$*$}
  \end{equation}
  Given that $\la F,
  G, \ffs, \ggs, \HG, \FE, \lnot \GE, \sth\ra$ is a lifting base,
  it is not difficult to verify that also
  \[\la \lnot G, \lnot F, \ggs, \ffs, \lnot \HG, \lnot \GE, \lnot \FE,
  \sth\ra\] forms a lifting base.  If we derive Lemma~\ref{lem-lift-sem} from
  this lifting base and $\NQ_1, \ldots, \NQ_n$ instead of $Q_1,\ldots, Q_n$ as
  starting points, then the lemma reads exactly as (\ref{eq-cor-lift-sem}).
  \qed
\end{proof}

Finally, the proof of the interpolant lifting theorem can be stated, based on
Lemma~\ref{lem-lift-sem} and considerations on syntactic properties of the
formula obtained by interpolant lifting.
\begin{proof}[Theorem~\ref{thm-lifting}]
  Let $H \eqdef Q_1 v_1 \ldots Q_n v_n\, \invsubst{\HG}{\stt}$. According to
  the definition of \name{Craig-Lyndon interpolant} (Definition~\ref{def-cli})
  we have to verify semantic and syntactic properties.  The semantic
  properties are $F \entails H$, which has been shown as
  Lemma~\ref{lem-lift-sem}, and $H \entails G$, stated as
  Corollary~\ref{cor-lift-sem}. The syntactic property is $\voc{H} \subseteq
  \voc{F} \cap \voc{G}$, or, equivalently, $\pred{H} \subseteq \pred{F} \cap
  \pred{G}$ \textit{and} $\fun{H} \subseteq \fun{F} \cap \fun{G}$.  The
  requirement on predicate occurrences hold since by conditions
  \ibref{ib:ipol}, \ibref{ib:pred} and~\ibrefm{ib:pred} of the
  definition of interpolant lifting base (Definition~\ref{def-ib}) it follows
  that
  \[\pred{H} \subseteq \pred{\HG} \subseteq \pred{\FE}
  \cap \pred{\GE} \subseteq \pred{F} \cap \pred{G}.
  \]
  The requirement on function symbols can be shown as follows: Since the
  construction of $H$ involves the inverse application of the
  substitution~$\stt$ to the ground interpolant~$\HG$, which causes all terms
  in which members of $\ffs \cup \ggs$ occur to be replaced by variables, it
  holds that $\fun{H} \cap (\ffs \cup \ggs) = \emptyset$.  From conditions
  \ibref{ib:ipol}, \ibref{ib:rngh}, \ibref{ib:fe} and~\ibrefm{ib:fe} of the
  definition of interpolant lifting base it follows that
  \[\fun{H} \subseteq \fun{\HG} \subseteq
  \fun{\FE} \cup \fun{\GE} \cup \ffs \cup \ggs \subseteq (\fun{F} \cap
  \fun{G}) \cup \ffs \cup \ggs,\] which together with $\fun{H} \cap (\ffs \cup
  \ggs) = \emptyset$ implies $\fun{H} \subseteq \fun{F} \cap \fun{G}$.  \qed
\end{proof}

\section{Related Work, Refinements and Issues}
\label{sec-cli-related}

\subsection{Related Work on Interpolant Lifting}
\label{sec-lift-related}

The construction of the lifting result according to Theorem~\ref{thm-lifting}
has in essence been already shown by Huang \cite{huang:95}.  A minor
difference is that Huang orders variables in the quantifier prefix by the
\emph{length} of the associated terms, which is more constrained than the
subterm relationship used in Theorem~\ref{thm-lifting}.

Although the construction of the lifting result can be expressed as a simple
formula conversion, independently of any particular calculus, its correctness
seems not trivial to prove and subtle issues arise.  The proof in
\cite{huang:95} depends with \cite[Lemma 12]{huang:95} on an induction over a
specific representation of a proof by resolution, paramodulation and
factoring.  The formula on which lifting is applied is characterized there by
two properties: First, it is a \name{relational interpolant}, which means that
it satisfies the constraints on a Craig interpolant except that function
symbols which are not shared by both input formulas are permitted.  Second, it
was obtained in the first stage with a specific inductive algorithm from a
propositional deduction tree that in turn was obtained by a specific method
from the first-order proof.  The correctness of Huang's interpolation method
for first-order logic with equality seems so far not settled. In
\cite{huang:95} it is not thoroughly proven and, as noted in
\cite[Sect.~6]{bonacina:15:on}, its specification \cite{huang:95} leaves
ambiguities concerning paramodulation.  In \cite[Sect.~6]{kovacs:17} it is
suggested that Huang's method is in presence of equality incorrect, but the
example given to substantiate this is not traceable. It may be adapted,
however, to illustrate that the notion of \name{relational interpolant} alone
is, at least in presence of equality, not sufficient to characterize the
lifted formulas: If the objective is to compute a Craig interpolant of $F =
(\fa=\fb)$ and $G = (\fg(\fa) = \fg(\fb))$, then Huang's first stage would
yield the relational interpolant $\fa=\fb$, which is also Craig
interpolant. It would not produce $\fg(\fa) = \fg(\fb)$, which is another
relational interpolant but whose lifting, that is, $\forall v_1 \forall v_2\,
v_1\!=v_2$, is not a Craig interpolant of~$F$ and~$G$ because it is not
entailed by~$F$.  Seemingly independently from \cite{huang:95}, the
correctness of interpolant lifting has been proven (for the case without
equality) in \cite[Lemma~8.2.2]{baaz:11} based on natural deduction proofs as
data structures (lifting is called \name{abstraction} in \cite{baaz:11}).  In
contrast, our justification of the lifting step (for the case without
equality) is based more abstractly on Herbrand's theorem instead of a
resolution or natural deduction proof structure. The formula to be lifted is
characterized as an actual Craig-Lyndon interpolant of two intermediate
formulas whose existence is ensured but which do not have to be constructed at
interpolant computation.

Huang's relational interpolants permit free variables upon which extra
quantifiers will be added after lifting.  As indicated in
\cite[p.~188]{huang:95}, this can be done in an arbitrary way: the extra
quantifiers can be existential or universal, at any position in the prefix.
In our formalization, the base formulas used for lifting have to be ground.
The effects described by Huang appear to be subsumed by the alternate
possibilities to instantiate non-ground tableaux delivered by provers as
discussed in Sect.~\ref{sec-grounding}.  The method of \cite{baaz:11} to
construct the results of the first stage, called \name{weak interpolants}
there, involves certain cases where quantified variables are introduced.

The special case of the lifting theorem where \emph{constants} are the only
functions to be replaced by variables was shown already by Craig in his proof
of the interpolation theorem for equality-free relational formulas
\cite[Lemma~1]{craig:uses} (as observed in \cite{bonacina:15:on}) and for
application to an interpolation method for proofs by resolution and
superposition as well as support for theory reasoning in
\cite{bonacina:15:on}. Both proofs are independent of a particular calculus or
proof data structures.

\subsection{Choices in Grounding and Side Assignment}
\label{sec-grounding} 

The \CTIF method (Procedure~\ref{proc-ctif}) for the construction of
first-order interpolants leaves at several stages alternate choices that have
effect on the formula returned as interpolant.
We discuss some of these here, although a thorough investigation of ways to
integrate the exploration and evaluation of these into interpolant
construction seems a nontrivial topic on its own.

The first considered choice concerns the tableau grounding step
(step~\ref{step-grounding} of the procedure).  Typically, provers instantiate
variables just as much ``as needed'' by the calculus to compute a closed
tableau.  To match with our inductive interpolant extraction and interpolant
lifting, these rigid variables in the literal labels of such free-variable
tableaux have to be instantiated by ground terms. There are different
possibilities to do so, all preserving the property that the tableau is a
leaf-closed \sided tableau for $F'$ and $G'$, but leading to different
interpolants: A variable can be instantiated by a term whose functions all
occur in both interpolation inputs.  The term may then occur in the
interpolant.  Alternatively, the variable can be instantiated by a term whose
outermost function symbol occurs in just one of the input formulas or has been
introduced at Skolemization.  By interpolant lifting the term will then be
replaced with a variable whose kind, existential or universal, depends on the
outermost symbol of the term, and whose quantifier position in the prefix is
constrained by its subterms.

Aside of these alternate possibilities that concern the instantiation of each
variable individually, there are also choices to instantiate different
variables by the same term or by different terms: Arbitrary subsets of the
free variables of the literal labels of the tableau can be instantiated with
the same ground term, leading in the interpolant to fewer quantified variables
but to more ``variable sharing'', that is, an increase in the number of
occurrences of each variable throughout the formula.

A second choice point concerns the side assignment
(step~\ref{step-side-assignment} of the procedure): A tableau clause can be an
instance of some clause $F'$ as well as of some clause in $G'$, such that its
side (i.e., the side labels of the nodes labeled with its literals) can be
assigned to $\aaa$ or~$\bbb$, where both assignments may lead to different
interpolants. This possibility may occur if a clause in $F'$ and a clause in
$G'$ have a common instance, including the cases where they are identical or
are variants, that is, identical modulo systematic variable renaming.  Effects
of alternate choices at side assignment were illustrated with
Example~\ref{examp-sides} on page~\pageref{examp-sides}.

\subsection{Preprocessing and Structure-Preserving Normalization}
\label{sec-preproc}

Sophisticated preprocessing is a crucial component of automated reasoning
systems with high performance. While formula simplifications such as removal
of subsumed clauses and removal of tautological clauses preserve equivalence,
others only preserve unsatisfiability.  For example, \name{purity
  simplification}, that is, removal of clauses that contain a literal with a
predicate that occurs with only one polarity in the formula.  Many
simplifications of the latter kind actually preserve not just
unsatisfiability, but, moreover, equivalence \emph{with respect to a set of
  predicates}, or, more precisely, a second-order equivalence
\begin{equation}
\label{eq-so-simp}
\exists p_1 \ldots \exists p_n\, F\; \equiv\;
\exists p_1 \ldots \exists p_n\, \f{simplify}(F),
\tag{$*$}
\end{equation}
where $\f{simplify}(F)$ stands for the result of the simplification operation
applied to~$F$. One might say that the \emph{semantics of the predicates not
  in $\{p_1,\ldots,p_n\}$ is preserved} by the simplification. For the
computation of Craig-Lyndon interpolants it is possible to preprocess the
first as well as the negated second input formula independently from each
other in ways such that the semantics of the predicates occurring in both
formulas is preserved in this sense.  The used preprocessors then should
support parameterization with the set of these predicates (see
\cite[Sect.~2.5]{cw-pie} for a discussion).

For clausal tableau methods some of these simplifications are particularly
relevant as they complement tableau construction with techniques which break
apart and join clauses and may thus introduce some of the benefits of
resolution.
Techniques for propositional logic that preserve
equivalence~(\ref{eq-so-simp}) for certain sets of predicates include variable
elimination by resolution \cite{biere:elim} and blocked clause elimination
\cite{jarv:blocked}. The preprocessing of \name{Prover9} \cite{prover9-mace4}
applies by default a form of predicate elimination.  In general, for
first-order generalizations of such elimination-based techniques the handling
of equality seems the most difficult issue.  Predicate elimination can
introduce equality also for inputs without equality.  In a semantic framework
where the Herbrand universe is taken as domain this can be avoided to some
degree, as shown in \cite{cw-skp} with a variant of the SCAN algorithm
\cite{scan} for predicate elimination.  Blocked clause elimination in
first-order logic \cite{blocked:fol:2017} comes in two variants, for formulas
without and with equality, respectively.

Another way to use equivalence~(\ref{eq-so-simp}) is by \emph{introducing}
fresh ``definer'' predicates for example by structure-preserving (also known
as \name{definitional}) normal forms such as the Tseitin transformation and
first-order generalizations of it
\cite{scott:twovars,tseitin,eder:def:85,plaisted:greenbaum}.  If disjoint sets
of definer predicates are used for the first and for the second interpolation
input, then, by the definition of \name{Craig-Lyndon interpolant}, definer
predicates do not occur in the interpolant.

In certain situations, which need further investigation, it might be useful to
relax this constraint.  For example, if two definer predicates have the same
subformula as definiens, it is in general useful to identify both predicates,
that is, to remove the defining formulas for one of them and replace its
definer predicate by the other, retained, definer predicate. If the two
definer predicates each stem from separate preprocessing of the first and
second interpolation inputs, respectively, the merged definer predicate would
occur in both inputs and might occur in the interpolant.  Another example
would be allowing definers occurring in the interpolant in cases where this
permits a condensed representation of a formula whose equivalent without the
definers would be much larger but straightforward to obtain.

\subsection{Equality Handling}
\label{sec-cli-related-equality}

So far we considered only first-order logic \emph{without equality}.
Nevertheless, our method to compute interpolants can be used together with the
well-known encoding of equality as a binary predicate with axioms that express
its reflexivity, symmetry and transitivity as well as axioms that express
substitutivity of predicates and functions.
If the input formulas of interpolant computation involve equality, these
axioms have to be added. The interpolation inputs are then formulas $E_{\aaa}
\land F$ and~$E_{\bbb} \imp G$ instead of $F$ and $G$, respectively, where
$E_{\aaa}$ and~$E_{\bbb}$ are conjunctions of equality axioms: The
substitutivity axioms for predicate and function symbols that occur only in
one of~$F$ or $G$ are placed in $E_{\aaa}$ or $E_{\bbb}$, respectively,
whereas the axioms that express reflexivity, symmetry and transitivity as well
as substitutivity axioms for symbols that occur in both~$F$ and~$G$ can be
placed arbitrarily in~$E_{\aaa}$, in~$E_{\bbb}$, or in both of them.

For formulas without function symbols with exception of constants in which
equality occurs only in one of the inputs, say $F$, more can be said about the
\emph{polarity} in which it can occur in the interpolant: The axioms
expressing reflexivity, symmetry and transitivity can be added to
$E_{\aaa}$. After adding the substitutivity axioms, for example $\forall x
\forall y\, (\fp(x) \land x=y \imp \fp(y))$, to $E_{\aaa}$ and $E_{\bbb}$ as
described above, all occurrences of equality in $E_{\bbb} \imp G$ are in
substitutivity axioms in which they have negative polarity.  Hence, in
$E_{\bbb} \imp G$ they have positive polarity and in a Craig-Lyndon
interpolant of $E_{\aaa} \land F$ and $E_{\bbb} \imp G$ they must also have
positive polarity. Analogously it can be shown that if equality occurs only in
the other input $G$, then it can occur in the interpolant only with negative
polarity.  Stronger constraints on interpolants with respect to equality are
stated in an interpolation theorem due to Oberschelp and Fujiwara (see
\cite{motohashi:84}).

Equality handling with goal-directed clausal tableau provers is notoriously
difficult.  The modern \name{leanCoP} system \cite{leancop} indeed encodes
equality just as a predicate with axioms as mentioned, with some practical
success for the price of loosing completeness due to \name{restricted
  backtracking}, a technique to cut off parts of the explored search space.
Recently an equality preprocessing technique for \name{leanCoP} has been
described \cite{oliver:otten:equality:2020}, which can be understood as
deriving specific clauses involving equality followed by performing
simplifications.  The best experimental results were obtained with an
incomplete variant in which also certain non-redundant clauses are deleted and
no equality axioms are added.  The historic \name{SETHEO} system
\cite{setheo:97} successfully integrated a refinement of Brand's
STE-modification \cite{brand:75}, a transformation of the source axioms that
makes the equality axioms redundant.  A further refinement of Brand's method
that takes term ordering constraints into account is described in
\cite{steq}. \CMProver \cite{cw-pie,cw-pie:2020} supports both, the axiom-based
equality representation and an implementation of the transformation from
\cite{steq}.  In principle, variants of Brand's transformation can be expected
to be applicable for interpolation, although their semantic properties shown
in the literature \cite{brand:75,steq} are just up to the preservation of
satisfiability and unsatisfiability, whereas for interpolation equivalence, or
at least the preservation of the second-order equivalence discussed in
Sect.~\ref{sec-preproc}, is required.

\subsection{Issues with Top-Down and Bottom-Up Clausal Tableau Provers} 

As demonstrated with the two different tableaux for the same inputs in
Example~\ref{ex-ipol-prop} on page~\pageref{ex-ipol-prop}, there exist in
general quite different closed clausal tableaux for a given clausal formula,
leading to different extracted interpolants.
For top-down methods such as model elimination \cite{loveland:1978} and the
connection method \cite{bibel:ar:1982,bo:C46}, the constructed tableau is
often largely determined by the chosen start clause, that is, the clause
attached to the root. The addition of further clauses is then guided by the
requirement that it closes an open branch through the connection condition,
that is, the last literal on the branch is the complement of a literal in the
clause.  Such provers typically consider a specific subset of the input
clauses as start clauses.  Without loss of completeness the set of negative
clauses can, for example, be taken as this subset, or, if a theorem is to be
proven from a consistent set of axioms, the clauses representing the (negated)
theorem. It remains to be investigated what choices of start clauses are
particularly useful for the computation of interpolants.

Bottom-up methods such as the \hypertableau calculus \cite{hypertab} as such
typically construct variations of clausal tableaux that do not match our
requirements as they may contain non-rigid variables.  Translations into
tableaux with only rigid variables need to be developed, which should be
facilitated by the constraint that different literals in the clause of a
\hypertableau are not allowed to share variables.  When applied to a clausal
formula that is range restricted (all variables in a clause occur in a
negative literal), an important case in practice, the \hypertableau calculus
constructs a tableau that is \emph{ground} -- and thus is trivially a tableau
with only rigid variables that can be directly used in the \CTIF procedure.

\subsection{An Implementation} 
\label{sec-cli-implem}

An implementation of the \CTIF method, which is available as free software, is
integrated in the \name{PIE} environment \cite{cw-pie,cw-pie:2020}. The
construction of the tableaux for interpolation is performed there with the
first-order prover \CMProver that proceeds in the goal-oriented
top-down way.\footnote{Inspired by the \name{Prolog Technology Theorem Prover}
  \cite{pttp} and \name{SETHEO} \cite{setheo:92}, \CMProver was originally
  written in 1993 but had been revived in 1996 \cite{cw-mathlib} and in
  2016. It was evaluated in 2018 on all suitable TPTP problems, that is,
  problems that have a distinguished theorem, are not classified as
  satisfiable and are in clausal or quantified first-order form: Of these, it
  can solve about 76\% of the 2143 problems without equality (in 9
  configurations) and about 26\% of the 11321 problems with equality (in 4
  configurations). The timeout was 600s, the TPTP version was 7.1.0. See
  \url{http://cs.christophwernhard.com/pie/cmprover/} for details.}
Experimental support is also provided for bottom-up tableau construction with
the \name{Hyper} \cite{cw-krhyper,cw-ekrhyper,hyper:2013} \hypertableau
system.  The clausal tableaux used for interpolant extraction are represented
as Prolog terms, providing an interface to integrate further provers.

Configurable preprocessing provides simplifications that respect preservation
of predicate semantics as required for interpolation, supports
structure-preserving clausification and handles the adding of equality axioms.
Configurable postprocessing allows to integrate into the interpolant
extraction ground formula simplifications that are aware of equality, e.g.,
rewrite $\fa=\fa$ to $\true$, and to apply different first-order
simplifications to the overall result.  \name{Symmetric} interpolation
\cite[Lemma~2]{craig:uses} (the name is due to \cite{mcmillan:symmetric}) with
consideration of predicate polarity is implemented as iterated Craig-Lyndon
interpolation.

So far, the implemented \CTIF method has not yet been experimentally compared
with other implementations of interpolation in first-order logic, such as an
extension of \name{Vampire} \cite{vampire:interpol:2012} and \name{Princess}
\cite{ruemmer:ipol:jar:2011}. Both of them are mainly targeted at applications
in verification that involve theory reasoning.  Experiments with applications
of interpolation to ground formulas for query reformulation are described in
\cite{benedikt:2017}. Aside of \name{Vampire} also implementations of
variations of the inductive algorithms from
\cite{huang:95,bonacina:11,mcmillan:2003} have been tried in
\cite{benedikt:2017} on the basis of resolution proofs returned by the
\name{MathSAT} SMT solver \cite{mathsat:13} and the \name{E} first-order
prover \cite{eprover:19}. In addition, a method based on the chase technique,
implemented with \name{DLV} \cite{dlv} as model generator was compared.  In
cases where the expected reformulation is a disjunctive normal form, the
\name{DLV}-based approach was only slightly worse than the best
resolution-based approach.  Another lesson reported in \cite{benedikt:2017}
was that the requirements on interpolation for query reformulation seem quite
different than for verification.  The \name{Vampire} extension, for example,
seems to compute just Craig (in contrast to Craig-Lyndon) interpolants.

\section{Conclusion}
\label{sec-conclusion}

We have investigated the computation of Craig interpolants with automated
theorem provers that compute clausal first-order tableaux.  The presented
method proceeds in two stages, similar to some interpolation methods for
resolution proofs \cite{huang:95,bonacina:15:on}. In the first stage an
intermediate formula is computed by an induction on the proof representation
returned by a theorem prover. The proof representation is in our case a
clausal tableau.  In the second stage the intermediate formula is converted to
an actual interpolant by lifting, that is, replacing terms with variables and
prepending a quantifier prefix.

The involved induction on clausal tableaux is an adaptation of an
interpolation method for analytic tableaux that deconstructs them bottom-up
\cite{smullyan:book:68}.  The version for clausal tableaux reveals striking
parallels with the induction on resolution deduction trees performed by the
interpolation methods surveyed in \cite{bonacina:15:ground}.  The parallels
involve some dualities and seem related to inherent correspondences of
resolution and clausal tableaux.  Exploring them in depth is an issue for
future research.
Based on a known linear simulation of tree resolution by clausal tableaux
\cite{letz:habil} we have shown that in propositional logic interpolation with
clausal tableaux can linearly simulate the most prominent inductive
interpolation methods for resolution proofs.  The obtained clausal tableaux
need to permit \name{atomic cuts}, that is, tautological clauses of the form
$\lnot A \lor A$. On the resolution side, tree-shaped proofs are required,
such that the potential sharing of subproofs in dags seems not retainable in
the simulations.
While the integration of cuts into theorem proving methods that
\emph{construct} clausal tableaux is nontrivial \cite{letz:cut:1994}, it is
straightforward to convert a \emph{given} resolution proof into a clausal
tableau with atomic cuts.
The simulations offer an interesting view on different interpolation methods
for resolution proofs, as their differences are reflected in small variations
of the clausal tableau translation. Future research might investigate the
variations systematically, potentially discovering further methods for
resolution proofs as backward translations of variations.

Our interpolation method computes Craig-\emph{Lyndon} interpolants, that is,
predicates occur in interpolants only in polarities in which they occur in
both input formulas.  In particular for applications in query reformulation
also further potential properties of interpolants are of interest. For
example, the Horn property or that quantifiers occur only together with atoms
that relativize their variables, as for example in range restricted formulas
(e.g., \cite{vgt}) or in access interpolants \cite{benedikt:book}.  Clausal
tableaux seem to provide a suitable basis for versions of interpolation that
ensure such properties.  As indicated in \cite{cw-report}, this is facilitated
if they are structured in such a way that the leaves are exactly the nodes
with negative literals (like \hypertableaux \cite{hypertab}, except that
variables are rigid).  Tableau conversions may be applied to tableaux returned
by provers to achieve that form \cite{cw-report}. This is an issue for future
research.

Interpolant lifting only creates interpolants that are in prenex form, which
might be undesirable or, for forms of interpolation where all occurrences of
quantifications in an interpolant must be relativized by atoms, not
possible. The construction of interpolants with relativized quantifications in
\cite{benedikt:book} is based on an analytic tableau method, where the
introduction of quantifiers may be considered as a special form of lifting
that is applied during interpolant construction to subformulas.  The
comparison and exploration of possible ways of lifting is another issue for
future research.

That the presented approach indeed provides a basis for implementing
interpolation with efficient machine-oriented theorem provers for first-order
logic that can be understood as constructing clausal tableaux has been
demonstrated with an implementation \cite{cw-pie,cw-pie:2020}, which, however,
has so far not been much tested or compared with other systems that are
capable of computing first-order interpolants.

A major drawback compared to approaches based on resolution with
paramodulation and superposition might be the inherently poor equality support
of clausal tableau provers, in particular provers that proceed top down and
create tableaux with rigid variables.  The simulations of propositional
interpolation methods for resolution proofs by clausal tableaux suggest, if
equality is important, another course to practical interpolation by computing
a proof with a system based on resolution and superposition, followed by
translating its equality reasoning steps into applications of equality axioms
and then converting the proof into a clausal tableau for interpolant
extraction.  Investigation of this idea, whether the conversion of a
first-order proof with equality reasoning steps can also be done efficiently,
and whether, in analogy to the propositional case, known interpolation methods
for resolution proofs can be simulated is a further issue for future research.

\subsubsection*{Acknowledgments}
This research was in part supported by Deutsche Forschungsgemeinschaft (DFG)
with grant~\mbox{WE~5641/1-1}.  The author thanks anonymous reviewers of
previous versions for numerous very helpful suggestions and remarks.

\bibliographystyle{spmpsci}
\bibliography{bibelim11short}

\closeout\plabelsfile
\closeout\plabelslogfile
\end{document}